\documentclass{lmcs} %%% last changed 2014-08-20
\pdfoutput=1

\usepackage{lastpage}
\lmcsdoi{17}{1}{12}
\lmcsheading{}{\pageref{LastPage}}{}{}%
{Dec.~19,~2019}{Feb.~03,~2021}{}

\usepackage[utf8]{inputenc}


\usepackage{graphicx}

\usepackage{subcaption}
\usepackage{wrapfig}

\usepackage{shuffle}
\DeclareFontFamily{U}{shuffle}{}
\DeclareFontShape{U}{shuffle}{m}{n}{ <-8>shuffle7 <8->shuffle10}{}
\usepackage{syntax}

\usepackage{mathtools}

\usepackage{array}
\newcolumntype{L}[1]{>{\raggedright\let\newline\\\arraybackslash\hspace{0pt}}m{#1}}
\newcolumntype{C}[1]{>{\centering\let\newline\\\arraybackslash\hspace{0pt}}m{#1}}
\newcolumntype{R}[1]{>{\raggedleft\let\newline\\\arraybackslash\hspace{0pt}}m{#1}}
\newcommand{\tuple}[1]{\left\langle#1\right\rangle}

\usepackage{tikz}
\usetikzlibrary{arrows,automata}
\usetikzlibrary{arrows.meta}
\usetikzlibrary{decorations.shapes}

\usepackage{algpseudocode}
\makeatletter
\renewcommand{\ALG@beginalgorithmic}{\footnotesize}
\makeatother
\usepackage{eqnarray}
\usepackage{alltt}
\usepackage{textcomp}

\usepackage{mathsprograms}

\usepackage{listings}
\lstdefinelanguage{JavaScript}{%
  keywords = { async, await, break, case, catch, class, const, continue, debugger, default, delete, do, each, else, export, finally, for, function, if, import, in, instanceof, let, new, of, return, switch, this, throw, try, typeof, var, void, while, with, yield },
  morecomment = [l]{//},
  morecomment = [s]{/*}{*/},
  morestring  = [b]',
  morestring  = [b]",
  sensitive   = true,
}

\lstdefinelanguage{Java10}{
  language      = Java,
  morekeywords  ={ var },
}

\theoremstyle{plain} %\crefname{satz}{Satz}{S\"atze}
\def\cf{{\em cf.}}

\begin{document}

\title[Robustness Against Transactional Causal Consistency]{Robustness Against Transactional Causal Consistency}
%\titlecomment{{\lsuper*}OPTIONAL comment concerning the title, \eg,
 % if a variant or an extended abstract of the paper has appeared elsewhere.}

\author[S.M.~Beillahi]{Sidi Mohamed Beillahi}	%required
%\thanks{thanks 1, optional.}	%optional

\author[A.~Bouajjani]{Ahmed Bouajjani}	%optional
%\address{IRIF, University Paris Diderot \& CNRS}	%optional
%\email{@irif.fr}  %optional
%\thanks{thanks 2, optional.}	%optional

\author[C.~Enea]{Constantin Enea}	%optional
%\address{IRIF, University Paris Diderot \& CNRS}	%optional
%\urladdr{@irif.fr}  %optional

\address{Universit\'{e} de Paris, IRIF, CNRS, F-75013 Paris, France}	%required
\email{beillahi@irif.fr}  %optional
\email{abou@irif.fr}  %optional
\email{cenea@irif.fr}  %optional

\thanks{This work is supported in part by the European Research Council (ERC) under the European Union's Horizon 2020 research and innovation programme (grant agreement No 678177).}	%optional

%% etc.

%% required for running head on odd and even pages, use suitable
%% abbreviations in case of long titles and many authors:

%%%%%%%%%%%%%%%%%%%%%%%%%%%%%%%%%%%%%%%%%%%%%%%%%%%%%%%%%%%%%%%%%%%%%%%%%%%

%% the abstract has to PRECEDE the command \maketitle:
%% be sure not to issue the \maketitle command twice!

\begin{abstract}
Distributed storage systems and databases are widely used by various types of applications. Transactional access to these storage systems is an important abstraction allowing application programmers to consider blocks of actions (i.e., transactions) as executing atomically. For performance reasons, the consistency models implemented by modern databases are weaker than the standard serializability model, which corresponds to the atomicity abstraction of transactions executing over a sequentially consistent memory. Causal consistency for instance is one such model that is widely used in practice.

In this paper, we investigate application-specific relationships between several variations of causal consistency and we address the issue of verifying automatically if a given transactional program is robust against causal consistency, i.e., all its behaviors when executed over an arbitrary causally consistent database are serializable. We show that programs without write-write races have the same set of behaviors under all these variations, and we show that checking robustness is polynomial time reducible to a state reachability problem in transactional programs over a sequentially consistent shared memory. A surprising corollary of the latter result is that causal consistency variations which admit incomparable sets of behaviors admit comparable sets of robust programs. This reduction also opens the door to leveraging existing methods and tools for the verification of concurrent programs (assuming sequential consistency) for reasoning about programs running over causally consistent databases. Furthermore, it allows to establish that the problem of checking robustness is decidable when the programs executed at different sites are finite-state.
%\keywords{Distributed Databases   \and Causal Consistency \and Model Checking}
\end{abstract}
%For x, we compare our model to a preceding operational model in which we found a flaw.
%It explores, formalizes and studies the differences between three variations of causal consistency and highlights them
%We examine the limits of consistency in
%We consider a class of consistency models
\maketitle

%% start the paper here:
%\everymath{\displaystyle}

%!TEX root = draft.tex
\section{Introduction}
Distribution and replication are widely adopted in order to implement storage systems and databases offering performant and available services. The implementations of these systems must ensure consistency guarantees allowing to reason about their behaviors in an abstract and simple way. Ideally, programmers of applications using such systems would like to have strong consistency guarantees, i.e., all updates occurring anywhere in the system are seen immediately and executed in the same order by all sites. Moreover, application programmers also need an abstract mechanism such as transactions, ensuring that blocks of actions (writes and reads) of a site can be considered as executing atomically without interferences from actions of other sites. For transactional programs, the consistency model offering strong consistency is {\em serializability} \cite{DBLP:journals/jacm/Papadimitriou79b}, i.e., every computation of a program is equivalent to another one where transactions are executed serially one after another without interference. In the non-transactional case this model corresponds to {\em sequential consistency} (SC) \cite{DBLP:journals/tc/Lamport79}. However, while serializability and SC are easier to apprehend by application programmers, their enforcement (by storage systems implementors) requires the use of global synchronization between all sites, which is hard to achieve while ensuring availability and acceptable performances \cite{DBLP:journals/jacm/FischerLP85, DBLP:journals/sigact/GilbertL02}.
%maintaining availability is in general hard (or even impossible) to achieve for both theoretical and practical reasons, especially in the setting of distributed systems over networks with failures and/or high latency \cite{}.
For this reason, modern storage systems ensure weaker consistency guarantees. %which allow computations that are not serializable/sequentially consistent.
In this paper, we are interested in studying {\em causal consistency} \cite{DBLP:journals/cacm/Lamport78}.

Causal consistency is a fundamental consistency model implemented in several production databases, e.g., AntidoteDB, CockroachDB, and MongoDB, and extensively studied in the literature~\cite{DBLP:conf/eurosys/AlmeidaLR13,DBLP:conf/cloud/DuE0Z13,DBLP:conf/sosp/LloydFKA11,DBLP:conf/nsdi/LloydFKA13,DBLP:conf/srds/PreguicaZBDBBS14}. Basically, when defined at the level of actions, it guarantees that every two  causally related actions, say $a_1$ is causally before (i.e., it has an influence on) $a_2$, are executed in that same order, i.e., $a_1$ before $a_2$, by all sites.
%whenever these two actions are seen and executed by any site, they are executed in that same order, i.e., $a_1$ before $a_2$.
The sets of updates visible to different sites may differ and read actions may return values that cannot be obtained in SC executions. The definition of causal consistency can be lifted to the level of transactions, assuming that transactions are visible to a site in their entirety (i.e., all their updates are visible at the same time), and they are executed by a site in isolation without interference from other transactions. In comparison to serializability, causal consistency allows that conflicting transactions, i.e., which read or write to a common location, be executed in different orders by different sites as long as they are not causally related. %The causal relation can be ensured using causal delivery of messages, which contain writes of transactions, between sites.
Actually, we consider three variations of causal consistency introduced in the literature, weak causal consistency (\wcct{})~\cite{DBLP:conf/ppopp/PerrinMJ16,DBLP:conf/popl/BouajjaniEGH17}, causal memory (\scct{})~\cite{DBLP:journals/dc/AhamadNBKH95,DBLP:conf/ppopp/PerrinMJ16}, and causal convergence (\ccvt{})~\cite{DBLP:journals/ftpl/Burckhardt14}.
%Two of these versions are incomparable, while a third one is weaker than both of them.
%In the rest of this section, we use causal consistency as a generic name for these three variations. Which
%, i.e., a transaction is causally dependent on another transaction if the latter affects the former or the two are issued by the same site.

\begin{figure}[t]
\lstset{basicstyle=\ttfamily\scriptsize}
\begin{subfigure}{44mm}
\begin{minipage}[c]{15mm}
\begin{lstlisting}[language=Java10]
t1 [z = 1
    x = 1]
t2 [y = 1]
\end{lstlisting}
\end{minipage}
\begin{minipage}[c]{4mm}
\footnotesize{$||$}
\end{minipage}
\begin{minipage}[c]{21mm}
\begin{lstlisting}[language=Java10]
t3 [x = 2
    r1 = z] //0
t4 [r2 = y  //1
    r3 = x] //2
\end{lstlisting}
\end{minipage}
\caption{\ccvt{} but not \scct{}.}
\label{fig:Lprogs1}
\end{subfigure}
\hspace{4mm}
\begin{subfigure}{49mm}
\begin{minipage}[c]{21mm}
\begin{lstlisting}[language=Java10]
t1 [x = 1]
t2 [r1 = x] //2	
\end{lstlisting}
\end{minipage}
\begin{minipage}[c]{4mm}
\footnotesize{$||$}
\end{minipage}
\begin{minipage}[c]{21mm}
\begin{lstlisting}[language=Java10]
t3 [x = 2]
t4 [r2 = x] //1
\end{lstlisting}
\end{minipage}
\caption{\scct{} but not \ccvt{}.}
\label{fig:Lprogs2}
\end{subfigure}
\hspace{4mm}
\begin{subfigure}{44mm}
\begin{minipage}[c]{15mm}
\begin{lstlisting}[language=Java10]
t1 [x = 2]
\end{lstlisting}
\end{minipage}
\begin{minipage}[c]{4mm}
\footnotesize{$||$}
\end{minipage}
\begin{minipage}[c]{21mm}
\begin{lstlisting}[language=Java10]
t2 [x = 1]
t3 [r1 = x] //2
t4 [r2 = x] //1
\end{lstlisting}
\end{minipage}
\caption{\wcct{} but not \scct{} nor \ccvt{}.}
\label{fig:Lprogs3}
\end{subfigure}
\caption{Program computations showing the relationship between \wcct{}, \ccvt{} and \scct{}. Transactions are delimited using brackets and the transactions issued on the same site are aligned vertically. The values read in a transaction are given in comments.}
\label{fig:Lprogs}
\end{figure}

The weakest variation of causal consistency, namely \wcct{},
%Weak causal consistency is the weakest model and it only guarantees causal dependency and
allows speculative executions and roll-backs of transactions which are not causally related (concurrent).
%, i.e., the order in which these transactions are applied on a given site may change over time.
For instance, the computation in Fig. \ref{fig:Lprogs3} is only feasible under \wcct{}: the site on the right applies ${\tt t2}$ after ${\tt t1}$ before executing ${\tt t3}$ and roll-backs ${\tt t2}$ before executing ${\tt t4}$.  \ccvt{} and \scct{} offer more guarantees. \ccvt{} enforces a total \emph{arbitration} order between all transactions which defines the order in which delivered concurrent transactions are executed by \emph{every} site. This guarantees that all sites reach the same state when all transactions are delivered. \scct{} ensures that \emph{all} values read by a site can be explained by an interleaving of transactions consistent with the causal order, enforcing thus PRAM consistency \cite{MISC:tr/princeton/Lipton88} on top of \wcct{}.
%the reads of a transaction are consistent with the reads done by transactions issued in the past on the same site.
%It implies that the causal past of each site must have linearizable history that explains all the reads that were done by this site.
Contrary to \ccvt{}, \scct{} allows that two sites diverge on the ordering of concurrent transactions, but both models do not allow roll-backs of concurrent transactions. Thus, \ccvt{} and \scct{} are incomparable in terms of computations they admit.
The computation in Fig. \ref{fig:Lprogs1} is not admitted by \scct{} because there is no interleaving of those transactions that explains the values read by the site on the right: reading $0$ from ${\tt z}$ implies that the transactions on the left must be applied after ${\tt t3}$ while reading $1$ from ${\tt y}$ implies that both ${\tt t1}$ and ${\tt t2}$ are applied before ${\tt t4}$ which contradicts reading 2 from ${\tt x}$. However, this computation is possible under \ccvt{} because ${\tt t1}$ can be delivered to the right after executing ${\tt t3}$ but arbitrated before ${\tt t3}$, which implies that the write to ${\tt x}$ in ${\tt t1}$ will be lost. The \scct{} computation in Fig.  \ref{fig:Lprogs2} is not possible under \ccvt{} because there is no arbitration order that could explain both reads from ${\tt x}$.
% the read of $1$ in $\atr 4$ and the read of $2$ in $\atr 3$ implies that there cannot be a total order between transactions $\atr 1$ and $\atr 2$ that write to $x$.
%The computation given Figure \ref{fig:Lprogs3} is only feasible under \wcct{}, since transaction $\atr 3$ reads $2$ then after a roll-back transaction $\atr 4$ reads $1$.
%TODO REFINE THIS: In the literature, \ccvt{} is usually used for applications that require semantics that provide causal consistency and eventual consistency \cite{DBLP:journals/ftpl/Burckhardt14}.
%\scct{} is mainly used for applications that require semantics that provide causal consistency and PRAM consistency \cite{MISC:tr/princeton/Lipton88}.  Several existing key-value stores including COPS \cite{DBLP:conf/sosp/LloydFKA11}, Orbe \cite{DBLP:conf/cloud/DuE0Z13}, and SwiftCloud \cite{DBLP:conf/srds/PreguicaZBDBBS14} adopt consistency semantics that are based on the combination of \ccvt{} and \scct{} semantics. In the rest of this section, we use causal consistency as a generic name for these three variations.

As a first contribution of our paper, we show that the three causal consistency models coincide for transactional programs containing no \emph{write-write races}, i.e., concurrent transactions writing on a common variable. %Thus,
We also show that if a transactional program has a write-write race under one of these models, then it must have a write-write race under any of the other two models. This property is rather counter-intuitive since \wcct{} is strictly weaker than both \ccvt{} and \scct{}, and \ccvt{} and \scct{} are incomparable (in terms of admitted behaviors).
Notice that each of the computations in Figures \ref{fig:Lprogs1}, \ref{fig:Lprogs2}, and \ref{fig:Lprogs3} contains a write-write race which explains why none of these computations is possible under all three models.

Then, we %consider the issue of ensuring that
%Here, ``weaker'' means ``more permissive'' or ``allowing more behaviors'' in general.
%Since, causal consistency allows computations that are not serializable/sequentially consistent. Then, an important issue is to ensure that
%the level of consistency needed by a given application coincides with the one that is guaranteed by its infrastructure, i.e., the storage system it uses. One way to tackle this issue is to
investigate the problem of checking {\em robustness} of application programs against causal consistency relaxations: Given a program $P$
%and two consistency models $S$ and $W$ such that $S$ is stronger than $W$,
and a causal consistency variation $X$,
we say that $P$ is robust against $X$ if the set of computations of $P$ when running under $X$ is the same as its set of computations when running under \emph{serializability}. This means that it is possible to reason about the behaviors of $P$ assuming the simpler serializability model and no additional synchronization is required when $P$ runs under $X$ such that it maintains all the properties satisfied under serializability.
%is not sensitive to the consistency relaxation from $S$ to $W$, and therefore (1) it is possible to reason about the behaviors of $P$ assuming that it is running over $S$, and (2) no additional synchronization is required when $P$ runs over the weak consistency model $W$ such that it maintains all its properties satisfied with the model $S$.
Checking robustness is not trivial, it can be seen as a form of checking program equivalence. However, the equivalence to check is between two versions of the same program, obtained using two different semantics, one more permissive than the other one. The goal is to check that this permissiveness has actually no effect on the particular program under consideration. The difficulty in checking robustness is to apprehend the extra behaviors due to the reorderings introduced by the relaxed consistency model w.r.t. serializability. This requires a priori reasoning about complex order constraints between operations in arbitrarily long computations, which may need maintaining unbounded ordered structures, and make the problem of checking robustness hard or even undecidable.

%This is the first work to study the semantics of three causal consistency models and to show the relationships between them in term of data races.
We show that verifying robustness of transactional programs against causal consistency can be reduced in polynomial time to the reachability problem in concurrent programs over SC. This allows
to reason about distributed applications running on causally consistent storage systems using the existing verification technology and it implies that
% (1) to avoid explicit handling of the causality relation between operations along computations (since this may imply memorizing an unbounded amount of information), and (2) to leverage available tools for verifying invariants/reachability problems on concurrent programs. Moreover, this implies that
 the robustness problem is decidable for finite-state programs; the problem is PSPACE-complete when the number of sites is fixed, and EXPSPACE-complete otherwise.
 This is the first result on the decidability and complexity of verifying robustness against causal consistency.
 In fact, the problem of verifying robustness has been considered in the literature for several consistency models of distributed systems, including causal consistency~\cite{DBLP:conf/concur/0002G16,DBLP:conf/popl/BrutschyD0V17,DBLP:conf/pldi/BrutschyD0V18,DBLP:journals/jacm/CeroneG18,DBLP:conf/concur/NagarJ18}.
These works provide (over- or under-)approximate analyses for checking robustness, but none of them provides precise (sound and complete) algorithmic verification methods for solving this problem, nor addresses its decidability and complexity.

The approach we adopt for tackling this verification problem is based on a precise characterization of the set of robustness violations, i.e., executions that are causally consistent but not serializable. For both \ccvt{} and \scct{}, we show that it is sufficient to search for a special type of robustness violations, that can be simulated by serial (SC) computations of an instrumentation of the original program. These computations maintain the information needed to recognize the pattern of a violation that would have occurred in the original program under a causally consistent semantics (executing the same set of operations). A surprising consequence of these results is that a program is robust against \scct{} iff it is robust against \wcct{}, and robustness against \scct{} implies robustness against \ccvt{}. This shows that the causal consistency variations we investigate can be incomparable in terms of the admitted behaviors, but comparable in terms of the robust applications they support.

% We reuse the results of the first part and show that the robustness against \wcct{} corresponds to the robustness against \scct{}. Moreover, using this characterization, we show that given a program $P$, deciding whether $P$ is not robust can be done by exploring only serial computations: We consider a program $P'$ obtained from $P$ by a linear-size instrumentation. The latter maintains along serial computations of $P'$ (where accesses to the main memory are done in a sequentially consistent way) the information needed to recognize the pattern of a violation that would have occurred in $P$ under a causally consistent semantics (executing the same set of operations).

%The key property we prove is the existence of minimal violations that obey to a finite number of patterns

%The complete formalizations and proofs are given in the Appendix.

%!TEX root = draft.tex
\section{Causal Consistency}\label{sec:CCSpec}

\subsection{Program syntax} 

We consider a simple programming language where a program is parallel composition of \emph{processes} distinguished using a set of identifiers $\mathbb{P}$. Our simple programming language syntax is given in Fig.~\ref{Figure:syntax}. Each process is a sequence of \emph{transactions} and each transaction is a sequence of \emph{labeled instructions}. Each transaction starts with a \plog{begin} instruction and finishes with an \plog{end} instruction.
Each other instruction is either an assignment to a process-local \emph{register} from a set $\mathbb{R}$ or to a \emph{shared variable} from a set $\mathbb{V}$, or an \plog{assume} statement.
The assignments use values from a data domain $\mathbb{D}$.
An assignment to a register $\langle reg\rangle := \langle var\rangle$ is called a \emph{read} of $\langle var\rangle$ and an assignment to a shared variable $\langle var\rangle := \langle reg\text{-}expr\rangle$ is called a \emph{write} to $\langle var\rangle$ ($\langle reg\text{-}expr\rangle$ is an expression over registers).
The statement \plog{assume} $\langle bexpr\rangle$ blocks the process if the Boolean expression $\langle bexpr\rangle$ over registers is false.
Each instruction is followed by a \plog{goto} statement which defines the evolution of the program counter. Multiple instructions can be associated with the same label which allows us to write non-deterministic programs and multiple \plog{goto} statements can direct the control to the same label which allows us to mimic imperative constructs like loops and conditionals. We assume that the control cannot pass from one transaction to another without going as expected through \plog{begin} and \plog{end} instructions.

\begin{figure}[t!]
\begin{minipage}[c]{0.49\textwidth}
{\scriptsize
\setlength{\grammarindent}{5em}
\setlength{\grammarparsep}{1pt}
\begin{grammar}
<prog>  ::= \plog{program} <process>$^{*}$

<process> ::= \plog{process} <pid> \plog{regs} <reg>$^{*}$ \\ <ltxn>$^{*}$

<ltxn> ::=  <binst> <linst>$^{*}$ <einst>

<binst> ::= <label>":" \plog{begin}";" \plog{goto} <label>";"

<einst> ::= <label>":" \plog{end}";" \plog{goto} <label>";"
\end{grammar}}
\end{minipage}
\hfill
\begin{minipage}[c]{0.47\textwidth}
{\scriptsize
\setlength{\grammarindent}{5em}
\begin{grammar}
<linst> ::= <label>":" <inst>";" \plog{goto} <label>";"

<inst> ::= <reg> ":=" <var>
  \alt <var> ":=" <reg-expr>
  \alt \plog{assume} <bexpr>
\end{grammar}}
\end{minipage}
\caption{Program syntax. $a^{*}$ indicates zero or more occurrences of $a$.  $\langle pid\rangle$, $\langle reg\rangle$, $\langle label \rangle$, and  $\langle var\rangle$ represent a process identifier, a register, a label, and a shared variable respectively. $\langle reg\text{-}expr \rangle$ is an expression over registers while $\langle bexpr \rangle$ is a Boolean expression over registers.}
\label{Figure:syntax}
\end{figure}

\subsection{Program Semantics Under Causal Memory} \label{sec:operationalModel1}

Informally, the semantics of a program under causal memory is defined as follows. The shared variables are replicated across each process, each process maintaining its own local valuation of these variables. During the execution of a transaction in a process, the shared-variable writes are stored in a \emph{transaction log} which is visible only to the process executing the transaction and which is broadcasted to all the processes at the end of the transaction\footnote{For simplicity, we assume that every transaction commits. The effects of aborted transactions shouldn't be visible to any process.}. To read a shared variable $\anaddr$, a process $\apr$ first accesses its transaction log and takes the last written value on $\anaddr$, if any, and then its own valuation of the shared variables, if  $\anaddr$ was not written during the current transaction. Transaction logs are delivered to every process in an order consistent with the \emph{causal delivery} relation between transactions, i.e., the transitive closure of the union of the \emph{program order} (the order in which transactions are executed by a process), and the \emph{delivered-before} relation (a transaction $\atr_1$ is delivered-before a transaction $\atr_2$ iff the log of $\atr_1$ has been delivered at the process executing $\atr_2$ before $\atr_2$ starts). By an abuse of terminology, we call this property \emph{causal delivery}. Once a transaction log is delivered, it is immediately applied on the shared-variable valuation of the receiving process. Also, no transaction log can be delivered to a process $\apr$ while $\apr$ is executing another transaction, we call this property \emph{transaction isolation}.

Formally, a program configuration is a triple $\gsconf = (\lsconf,\msgsconf)$ where $\lsconf: \mathbb{P} \rightarrow \lstatesconf$ associates a local state in $\lstatesconf$ to each process in $\mathbb{P}$, and $\msgsconf$ is a set of messages in transit.
A local state is a tuple  $\tuple{\pcconf,\storeconf,\valconf,\txnwrsconf}$ where
$\pcconf \in\labconf$ is the program counter, i.e., the label of the next instruction to be executed, $\storeconf: \mathbb{V} \rightarrow \mathbb{D}$ is the local valuation of the shared variables, $\valconf: \mathbb{R} \rightarrow \mathbb{D}$ is the valuation of the local registers, and $\txnwrsconf \in (\mathbb{V} \times \mathbb{D})^{*}$ is the transaction log, i.e., a list of variable-value pairs. For a local state $s$, we use $s.\pcconf$ to denote the program counter component of $s$, and similarly for all the other components of $s$.
A message $m=\tuple{\atr,\alog}$ is a transaction identifier $\atr$ from a set $\mathbb{T}$ together with a transaction log $\alog\in (\mathbb{V} \times \mathbb{D})^{*}$. We let $\mathbb{M}$ denote the set of messages.

Then, the semantics of a program $\aprog$ under causal memory is defined using a labeled transition system (LTS) $[\aprog]_{\scct{}}=(\gstatesconf,\eventsconf,\gsconf_0,\rightarrow)$ where $\gstatesconf$ is the set of program configurations, $\eventsconf$ is a set of transition labels called \emph{events}, $\gsconf_0$ is the initial configuration, and $\rightarrow\subseteq \gstatesconf\times \eventsconf\times \gstatesconf$ is the transition relation. As it will be explained later in this section, the executions of $\aprog$ under causal memory are a subset of those generated by $[\aprog]_{\scct{}}$.
The set of events is defined by:
\begin{align*}
\eventsconf =\ & \{\ \beginact(\apr,\atr), \loadact(\apr,\atr,\anaddr,\aval), \issueact(\apr,\atr,\anaddr,\aval), \storeact(\apr,\atr), \commitact(\apr,\atr): \apr\in \mathbb{P}, \atr\in \mathbb{T}, \anaddr\in \mathbb{V}, \aval\in \mathbb{D}\}
\end{align*}
where $\beginact$ and $\commitact$ label transitions corresponding to the start, resp., the end of a transaction, $\issueact$ and $\loadact$ label transitions corresponding to writing, resp., reading, a shared variable during some transaction, and $\storeact$ labels transitions corresponding to applying a transition log to the local state of the process issuing the transaction or to the state of another process that received the log. An event $\issueact$ is called an \emph{issue} while an event $\storeact$  is called a \emph{store}.

The transition relation $\rightarrow$ is partially defined in Fig.~\ref{Table:MPRules} (we will present additional constraints later in this section). The events labeling a transition are written on top of $\rightarrow$. A $\beginact$ transition will just reset the transaction log while an $\commitact$ transition will add the transaction log together with the transaction identifier to the set $\msgsconf$ of messages in transit. An $\loadact$ transition will read the value of a shared-variable looking first at the transaction log $\txnwrsconf$ and then, at the shared-variable valuation $\storeconf$, while an $\issueact$ transition will add a new write to the transaction log. Finally, a $\storeact$ transition represents the delivery of a transaction log that was in transit which is applied immediately on the shared-variable valuation $\storeconf$. 
%apply the writes stored in the transaction log on the shared-variable valuation $\storeconf$  and
\begin{figure}[!t]
\small\addtolength{\tabcolsep}{-5pt}
\therules{
\scriptsize
\therule
{$\text{\plog{begin}}\in\instrOf(\lsconf(\apr).\pcconf)$\quad $s = \lsconf(\apr)[\txnwrsconf \mapsto \epsilon,\pcconf\mapsto \mathsf{next}(\pcconf)]$}
{$(\lsconf,\msgsconf)\mpitrans{\beginact(\apr,\atr)} (\lsconf[\apr\mapsto s],\msgsconf)$}
\scriptsize
\dfrac
{\text{$\theload{\areg}{\anaddr}\in\instrOf(\lsconf(\apr).\pcconf)$\quad
$\mathit{eval}(\lsconf(\apr),\anaddr) = \aval$ \quad
$\mathit{rval} = \lsconf(\apr).\valconf[\areg\mapsto \aval]$\quad 
$s = \lsconf(\apr)[\valconf \mapsto \mathit{rval},\pcconf\mapsto \mathsf{next}(\pcconf)]$}}
{\text{$(\lsconf,\msgsconf)\mpitrans{\loadact(\apr,\atr,\anaddr,\aval)} (\lsconf[\apr\mapsto s],\msgsconf)$}}\\[8mm]
\scriptsize
\therule
{$\thestore{\anaddr}{\aval}\in\instrOf(\lsconf(\apr).\pcconf)$\quad
$\mathit{log} = (\lsconf(\apr).\txnwrsconf) \cdot (\anaddr,\aval)$\quad
$s = \lsconf(\apr)[\txnwrsconf \mapsto \mathit{log},\pcconf\mapsto \mathsf{next}(\pcconf)]$}
{$(\lsconf,\msgsconf)\mpitrans{\issueact(\apr,\atr,\anaddr,\aval)} (\lsconf[\apr\mapsto s],\msgsconf)$}
\scriptsize
\dfrac
{\text{$\text{\plog{end}}\in\instrOf(\lsconf(\apr).\pcconf)$\quad
%$\mathit{store} = \lsconf(\apr).\storeconf[\anaddr\mapsto \mathit{eval}(\lsconf(\apr),\anaddr): \anaddr\in\mathbb{V}]$\quad
%\storeconf \mapsto \mathit{store},
 $s = \lsconf(\apr)[\pcconf\mapsto \mathsf{next}(\pcconf)]$}}
{\text{$(\lsconf,\msgsconf)\mpitrans{\commitact(\apr,\atr)} (\lsconf[\apr\mapsto s],\msgsconf\cup\{(\atr,\lsconf(p).\txnwrsconf)\})$}}\\[8mm]
\scriptsize
\dfrac
{\text{$\tuple{t,\mathit{log}}\in \msgsconf$\quad
$\mathit{store} = \lsconf(\apr).\storeconf[\anaddr\mapsto \mathit{last}(\mathit{log},\anaddr): \anaddr\in\mathbb{V}, \mathit{last}(\mathit{log},\anaddr)\neq\bot]$ \quad $s = \lsconf(\apr)[\storeconf \mapsto \mathit{store}]$}}
{\text{$(\lsconf,\msgsconf)\mpitrans{\storeact(\apr,\atr)} (\lsconf[\apr\mapsto s],\msgsconf)$}}
}
\caption{The set of transition rules defining the causal memory semantics. We assume that all the events which come from the same transaction use a unique transaction identifier $\atr$. For a function $f$, we use $f[a\mapsto b]$ to denote a function $g$ such that $g(c)=f(c)$ for all $c\neq a$ and $g(a)=b$. The function $\instrOf$ returns the set of instructions labeled by some given label while $\mathsf{next}$ gives the next instruction to execute. We use $\cdot$ to denote sequence concatenation. The function $\mathit{eval}(\lsconf(\apr),\anaddr)$ returns the value of $\anaddr$ in the local state $\lsconf(\apr)$: (1) if $\lsconf(\apr).\txnwrsconf$ contains a pair $(\anaddr,\aval)$, for some $\aval$, then $\mathit{eval}(\lsconf(\apr),\anaddr)$ returns the value of the last such pair in $\lsconf(\apr).\txnwrsconf$, and (2) $\mathit{eval}(\lsconf(\apr),\anaddr)$ returns $\lsconf(\apr).\storeconf(\anaddr)$, otherwise. Also, $\mathit{last}(\mathit{log},\anaddr)$ returns the value $v$ in the last pair $(\anaddr,\aval)$ in $\mathit{log}$, and $\bot$, if such a pair does not exist.}
\label{Table:MPRules}
\end{figure}

We say that an execution $\rho$ satisfies \emph{transaction isolation} if no transaction log is delivered to a process $\apr$ while $\apr$ is executing a transaction, i.e., if an event $\event=\storeact(\apr,\atr)$ occurs in $\rho$ before an event $\event'=\commitact(\apr,\atr')$ with $\atr' \neq \atr$, then $\rho$ contains an event $\event''=\beginact(\apr,\atr')$ between $\event$ and $\event'$. For an execution $\rho$ satisfying transaction isolation, we assume w.l.o.g. that transactions executed by different processes do not interleave, i.e., if an event $\event$ associated to a transaction $\atr$ (an event of the process executing $\atr$ or the delivery of the transaction log of $\atr$) occurs in $\rho$ before $\event'=\commitact(\apr',\atr')$, then $\rho$ contains an event $\event''=\beginact(\apr',\atr')$ between $\event$ and $\event'$.
Formally, we say that an execution $\rho$ satisfies \emph{causal delivery} if the following hold:
\begin{itemize}
  \item for any event $\beginact(\apr,\atr)$, and for any process $\apr'$, $\rho$ contains at most one event $\storeact(\apr',\atr)$,
  \item for any two events $\beginact(\apr,\atr)$ and $\beginact(\apr,\atr')$, if $\beginact(\apr,\atr)$ occurs in $\rho$ before $\beginact(\apr,\atr')$, then the event $\storeact(\apr,\atr)$ occurs before $\beginact(\apr,\atr')$ in $\rho$. This ensures that when $\apr$ issues $\atr$ it must store the writes of $\atr$ in its local state before issuing another transaction $\atr'$; 
	\item for any events $\event_1\in\{\storeact(\apr,\atr_1),\commitact(\apr,\atr_1)\}$, $\event_2=\beginact(\apr,\atr_2)$, and $\event_2'=\storeact(\apr',\atr_2)$ with $\apr\neq\apr'$, if $\event_1$ occurs in $\rho$ before $\event_2$, then there exists $\event_1'=\storeact(\apr',\atr_1)$ such that $\event_1'$ occurs before $\event_2'$ in $\rho$.
\end{itemize}
An execution $\rho$ satisfies \emph{causal memory} if it satisfies transaction isolation and causal delivery. The set of executions of $\aprog$ under causal memory, denoted by $\executionsconf_{\scct{}}(\aprog)$, is the set of executions of $[\aprog]_{\scct{}}$ satisfying causal memory.

%Let $\rho$  be an execution under causal memory, a sequence $\tau$ of events $\issueact(\apr,\atr)$ and $\storeact(\apr,\atr)$ with $\apr\in\mathbb{P}$ and $\atr\in\mathbb{T}$ is called a \emph{summary of $\rho$} if it is obtained from the projection of $\rho$ on $\issueact$ and $\storeact$ events by substituting every sub-sequence of transitions in $\rho$ delimited by a $\beginact$ and an $\commitact$ transition, with a single ``macro-event'' $\issueact(\apr,\atr)$.
Fig.~\ref{firstexamplecm} shows an execution under \scct{}. This execution satisfies transaction isolation since no transaction is delivered while another transaction is executing.

\begin{figure}[t]
  \begin{subfigure}{144mm}
  \begin{minipage}{144mm}
  \scalebox{0.61}
  {\begin{tikzpicture}[shape=rectangle,draw=none,font=\large,align=left]
    \node (A) at (2,0)  [] {$\beginact(\apr 1,\atr 1)$\\ $\ \issueact(\apr 1,\atr 1,x,1)$\\ $\commitact(\apr 1,\atr 1)$};
    \node (AA1) at (3.5,0)  [] {$\cdot$};
    \node (A1) at (4.5,0)  [] {$\storeact(\apr 1,\atr 1)$};
    \node (A1B) at (5.8,0)  [] {$\cdot$};
    \node (B) at (7.6,0)  [] {$\beginact(\apr 2,\atr 3)$\\ $\ \issueact(\apr 2,\atr 3,x,2)$\\ $\commitact(\apr 2,\atr 3)$};
    \node (BB1) at (9.1,0)  [] {$\cdot$};
    \node (B1) at (10.1,0)  [] {$\storeact(\apr 2,\atr 3)$};
    \node (B1C) at (11.1,0)  [] {$\cdot$};
    \node (C) at (12.1,0)  [] {$\storeact(\apr 1,\atr 3)$};
    \node (CD) at (13.4,0)  [] {$\cdot$};
    \node (D) at (14.7,0)  [] {$\storeact(\apr 2,\atr 1)$};
    \node (DE) at (16,0)  [] {$\cdot$};
    \node (E) at (17.8,0)  [] {$\beginact(\apr 2,\atr 4)$\\ $\ \loadact(\apr 2,\atr 4,x,1)$\\ $\commitact(\apr 2,\atr 4)$};
    \node (EF) at (19.5,0)  [] {$\cdot$};
    \node (F) at (21.3,0)  [] {$\beginact(\apr 1,\atr 2)$\\ $\ \loadact(\apr 1,\atr 2,x,2)$\\ $\commitact(\apr 1,\atr 2)$};
  \end{tikzpicture}}
  \end{minipage}
  \caption{\scct{} execution of the program in Fig.~\ref{fig:Lprogs2}.}
  \label{firstexamplecm}
  \end{subfigure}
  
  \begin{subfigure}{144mm}
  \begin{minipage}{144mm}
  \scalebox{0.61}
  {\begin{tikzpicture}[shape=rectangle,draw=none,font=\large,align=left]
    \node (A) at (2,0)  [] {$\beginact(\apr 1,\atr 1)$\\ $\ \issueact(\apr 1,\atr 1,z,1)$\\ $\ \issueact(\apr 1,\atr 1,x,1)$\\ $\commitact(\apr 1,\atr 1)$};
    \node (AA1) at (3.5,0)  [] {$\cdot$};
    \node (A1) at (4.5,0)  [] {$\storeact(\apr 1,\atr 1)$};
    \node (A1B) at (5.6,0)  [] {$\cdot$};
    \node (B) at (7.1,0)  [] {$\beginact(\apr 2,\atr 3)$\\ $\ \issueact(\apr 2,\atr 3,x,2)$\\ $\ \loadact(\apr 2,\atr 3,z,0)$\\ $\commitact(\apr 2,\atr 3)$};
    \node (BB1) at (8.6,0)  [] {$\cdot$};
    \node (B1) at (9.6,0)  [] {$\storeact(\apr 2,\atr 3)$};
    \node (B1C) at (10.6,0)  [] {$\cdot$};
    \node (C) at (11.6,0)  [] {$\storeact(\apr 1,\atr 3)$};
    \node (CD) at (12.7,0)  [] {$\cdot$};
    \node (D) at (13.8,0)  [] {$\storeact(\apr 2,\atr 1)$};
    \node (DE) at (14.9,0)  [] {$\cdot$};
    \node (E) at (16.4,0)  [] {$\beginact(\apr 1,\atr 2)$\\ $\ \issueact(\apr 1,\atr 2,y,1)$\\ $\commitact(\apr 1,\atr 2)$};
    \node (EE1) at (17.9,0)  [] {$\cdot$};
    \node (E1) at (18.9,0)  [] {$\storeact(\apr 1,\atr 2)$};
    \node (E1F) at (19.9,0)  [] {$\cdot$};
    \node (F) at (20.9,0)  [] {$\storeact(\apr 2,\atr 2)$};
    \node (FG) at (22,0)  [] {$\cdot$};
    \node (G) at (23.5,0)  [] {$\beginact(\apr 2,\atr 4)$\\ $\ \loadact(\apr 2,\atr 4,y,1)$\\ $\ \loadact(\apr 2,\atr 4,x,2)$\\ $\commitact(\apr 2,\atr 4)$};
  \end{tikzpicture}}
  \end{minipage}
  \caption{\ccvt{} execution of the program in Fig.~\ref{fig:Lprogs1}.}
  \label{firstexampleccv}
  \end{subfigure}
  \begin{subfigure}{144mm}
    \begin{minipage}{144mm}
    \scalebox{0.61}
    {\begin{tikzpicture}[shape=rectangle,draw=none,font=\large,align=left]
      \node (A) at (2,0)  [] {$\beginact(\apr 1,\atr 1)$\\ $\ \issueact(\apr 1,\atr 1,x,2)$\\ $\commitact(\apr 1,\atr 1)$};
      \node (AA1) at (3.5,0)  [] {$\cdot$};
      \node (A1) at (4.5,0)  [] {$\storeact(\apr 1,\atr 1)$};
      \node (A1B) at (5.8,0)  [] {$\cdot$};
      \node (B) at (7.6,0)  [] {$\beginact(\apr 2,\atr 2)$\\ $\ \issueact(\apr 2,\atr 2,x,1)$\\  $\commitact(\apr 2,\atr 2)$};
      \node (BB1) at (9.1,0)  [] {$\cdot$};
      \node (B1) at (10.1,0)  [] {$\storeact(\apr 2,\atr 2)$};
      \node (B1C) at (11.4,0)  [] {$\cdot$};
      \node (C) at (12.7,0)  [] {$\storeact(\apr 2,\atr 1)$};
      \node (CE) at (14,0)  [] {$\cdot$};
      \node (E) at (15.8,0)  [] {$\beginact(\apr 2,\atr 3)$\\ $\ \loadact(\apr 2,\atr 3,x,2)$\\ $\commitact(\apr 2,\atr 3)$};
      \node (FG) at (17.4,0)  [] {$\cdot$};
      \node (G) at (19.2,0)  [] {$\beginact(\apr 2,\atr 4)$\\ $\ \loadact(\apr 2,\atr 4,x,1)$\\  $\commitact(\apr 2,\atr 4)$};
      \node (GD) at (21,0)  [] {$\cdot$};
      \node (D) at (22.3,0)  [] {$\storeact(\apr 1,\atr 2)$};
    \end{tikzpicture}}
    \end{minipage}
    \caption{\wcct{} execution of the program in Fig.~\ref{fig:Lprogs3}.}
    \label{firstexamplewcc}
    \end{subfigure}
    \caption{For readability, the sub-sequences of events delimited by $\beginact$ and $\commitact$ are aligned vertically, the execution-flow advancing from left to right and top to bottom.}
\end{figure}

\subsection{Program Semantics Under Causal Convergence}\label{sec:operationalModel2}

Compared to causal memory, causal convergence ensures eventual consistency of process-local copies of the shared variables. Each transaction log is associated with a
timestamp and a process applies a write on some variable $\anaddr$ from a transaction log only if it has a timestamp larger than the timestamps of all the transaction logs it has already applied and that wrote the same variable $\anaddr$. For simplicity, we assume that the transaction identifiers play the role of timestamps, which are totally ordered according to some relation $<$. \ccvt{} satisfies both \emph{causal delivery} and \emph{transaction isolation} as well. Assuming that transactions are constituted of either a read alone or a write alone, \ccvt{} is equivalent to Strong Release-Acquire (SRA), a strengthening of the standard Release-Acquire of the C11 memory model~\cite{DBLP:conf/popl/LahavGV16}\footnote{This equivalence excludes the atomic read-modify-write (also know as compare-and-swap) operation which is not provided by \ccvt{}.}.

Formally, we define a variation of the LTS $[\aprog]_{\scct{}}$, denoted by $[\aprog]_{\ccvt{}}$, where essentially, the transition identifiers play the role of timestamps and are ordered by a total order $<$, each process-local state contains an additional component $\timest$ storing the largest timestamp the process has seen for each variable, and a write on a variable $\anaddr$ from a  transaction log is applied on the local valuation $\storeconf$ only if it has a timestamp larger than $\timest(\anaddr)$.
Also, a $\beginact(\apr,\atr)$ transition will choose a transaction identifier $\atr$ greater than those in the image of the $\timest$ component of $\apr$'s local state.
The transition rules of $[\aprog]_{\ccvt{}}$ that change w.r.t. those of $[\aprog]_{\scct{}}$ are given in Fig.~\ref{Table:MPRules-ccv}.
\begin{figure}[tt]
\small\addtolength{\tabcolsep}{-5pt}
\therules{
\scriptsize
\therule
{$\Phi_1$\quad
$\mathsf{img}(\lsconf(\apr).\timest) < \atr$}
{$(\lsconf,\lkconf,\msgsconf)\mpitrans{\beginact(\apr,\atr)} (\lsconf[\apr\mapsto s],\lkconf,\msgsconf)$}
%\scriptsize
%\dfrac
%{\text{$\text{\plog{end}}\in\instrOf(\lsconf(\apr).\pcconf)$\quad
%$\mathit{tstamp} = \lsconf(\apr).\timest[\anaddr\mapsto t: \anaddr\in\mathbb{V}, \mathit{last}(\mathit{log},\anaddr)\neq\bot]$ \quad 
%\timest\mapsto \mathit{tstamp},
%$s = \lsconf(\apr)[\pcconf\mapsto \mathsf{next}(\pcconf)]$
%$\mathit{store} = \lsconf(\apr).\storeconf[\anaddr\mapsto \mathit{eval}(\lsconf(\apr),\anaddr): \anaddr\in\mathbb{V}]$ 
%\storeconf \mapsto \mathit{store}, 
%}}
%{\text{}}}
%{\text{$(\lsconf,\lkconf,\msgsconf)\mpitrans{\commitact(\apr,\atr)} (\lsconf[\apr\mapsto s],\lkconf,\msgsconf\cup\{(\atr,\lsconf(p).\txnwrsconf)\})$}}\\[8mm]
\scriptsize
\dfrac
{\splitdfrac{\text{$\tuple{t,\mathit{log}}\in \msgsconf$\quad
$\mathit{store} = \lsconf(\apr).\storeconf[\anaddr\mapsto \mathit{last}(\mathit{log},\anaddr): \anaddr\in\mathbb{V}, \mathit{last}(\mathit{log},\anaddr)\neq\bot, \timest(\anaddr) < t]$}}
{\text{$\mathit{tstamp} = \lsconf(\apr).\timest[\anaddr\mapsto t: \anaddr\in\mathbb{V}, \mathit{last}(\mathit{log},\anaddr)\neq\bot, \timest(\anaddr) < t]$\quad 
$s = \lsconf(\apr)[\storeconf \mapsto \mathit{store},\timest\mapsto \mathit{tstamp}]$}}}
{\text{$(\lsconf,\lkconf,\msgsconf)\mpitrans{\storeact(\apr,\atr)} (\lsconf[\apr\mapsto s],\lkconf,\msgsconf)$}}%\\[2em]
}
\caption{Transition rules for defining causal convergence. $\Phi_1$ is the hypothesis of the $\beginact(\apr,\atr)$ transition rule in Fig.~\ref{Table:MPRules}, and $\mathsf{img}$ denotes the image of a function.}
\label{Table:MPRules-ccv}
\end{figure}

The set of executions of $\aprog$ under causal convergence, denoted by $\executionsconf_{\ccvt{}}(\aprog)$, is the set of executions of $[\aprog]_{\ccvt{}}$ satisfying transaction isolation, causal delivery, and the fact that every process $\apr$ generates monotonically increasing transaction identifiers. 

The execution in Fig.~\ref{firstexamplecm} is not possible under causal convergence since $\atr 4$ and $\atr 2$ read $2$ and $1$ from $\anaddr$, respectively. This is possible only if $\atr 1$ and $\atr 3$ write $\anaddr$ at $\apr 2$ and $\apr 1$, respectively, which contradicts the definition of $\storeact$ transition given in Fig.~\ref{Table:MPRules-ccv} where we cannot have both $\atr 1 < \atr 3$ and  $\atr 3 < \atr 1$ at the same time. Fig.~\ref{firstexampleccv} shows an execution under \ccvt{} (we assume $\atr_1<\atr_2<\atr_3<\atr_4$). Notice that $\storeact(\apr 2,\atr 1)$ did not result in an update of $x$ because the timestamp $\atr_1$ is smaller than the timestamp of the last transaction that wrote $x$ at $\apr_2$, namely $\atr_3$, a behavior that is not possible under \scct{}. The two processes converge and store the same shared variable copy at the end of the execution.

\subsection{Program Semantics Under Weak Causal Consistency} \label{sec:operationalModel3}

Compared to the previous semantics, \wcct{} allows that reads of the same process observe concurrent writes as executing in different orders. Each process maintains a \emph{set} of values for each shared variable, and a read returns any one of these values non-deterministically. 
Transaction logs are associated with \emph{vector clocks}~\cite{DBLP:journals/cacm/Lamport78} which represent the causal delivery relation, i.e., a transaction $\atr_1$ is before $\atr_2$ in causal-delivery iff the vector clock of $\atr_1$ is smaller than the vector clock of $\atr_2$.
We assume that transactions identifiers play the role of vector clocks, which are partially ordered according to some relation $<$.
In applying the log of a transaction $\atr$ on the local state of the receiving process $\apr$, the final \emph{set} of values for each shared variable in $\apr$ will be constituted of the value in the log of $\atr$ and the values that were written by \emph{concurrent} transactions (not related by causal delivery to $\atr$). \wcct{} satisfies both \emph{causal delivery} and \emph{transaction isolation}.

Formally, in \wcct{} semantics, the local valuation of the shared variables  $\storeconf: \mathbb{V} \rightarrow (\mathbb{D} \times \mathbb{T})^{*}$ is a map that accepts a shared variable and returns a set of pairs. The pairs are constituted of values that were written concurrently and identifiers of the transactions that wrote those values.  When applying a transaction log on the local valuation store, we keep the values that were written by transactions that are concurrent with the current transaction. 
Additionally, in the \wcct{} semantics, the local state of a process has an additional component $\storeconfcopy: \mathbb{V} \rightarrow (\mathbb{D} \times \mathbb{T})$ that maps each shared variable to a single pair. $\storeconfcopy$ is obtained by taking a ``consistent'' snapshot from $\storeconf$ when a new transaction starts. Such a snapshot corresponds to a linearization of the transactions that were delivered to the process, which is consistent with the vector clock order. The snapshot associates to each variable the last value written in this linearization.  %such that there exists $\stowcc$ a partial order where for every $\anaddr \in \mathbb{V}$ such that $\storeconfcopy(\anaddr)=(\aval,\atr)$ we have $(\aval,\atr)\in\storeconf(\anaddr)$ and for every other pair $(\aval',\atr')\in\storeconf(\anaddr)$ we have $(\atr',\atr)\in\stowcc$. 
When a process does a read from a shared variable $\anaddr$, it looks first at the transaction log $\txnwrsconf$ and then, at the variable valuation $\storeconfcopy$. 
%In this way, $\storeconfcopy$ ensures that reads from the same transaction are consistent with each other.
In Fig.~\ref{Table:MPRules-cc}, we provide the transition rules of $[\aprog]_{\wcct{}}$ that change w.r.t. those of $[\aprog]_{\ccvt{}}$ and $[\aprog]_{\scct{}}$.

\begin{figure}[t]
\small\addtolength{\tabcolsep}{-10pt}
\therules{
\dfrac
{\splitdfrac{\text{$\text{\plog{begin}}\in\instrOf(\lsconf(\apr).\pcconf)$ \quad
$\mathsf{img}(\lsconf(p).\timest) < t$}}
{\text{$s = \lsconf(\apr)[\txnwrsconf \mapsto \epsilon,\storeconfcopy \mapsto \buildsnapshot(\storeconf),\pcconf\mapsto \mathsf{next}(\pcconf)]$}}}
{\text{$(\lsconf,\msgsconf)\mpitrans{\beginact(\apr,\atr)} (\lsconf[\apr\mapsto s],\msgsconf)$}}\\[2em]
\dfrac
{\splitdfrac{\text{$\theload{\areg}{\anaddr}\in\instrOf(\lsconf(\apr).\pcconf)$\quad
$\mathit{cceval}(\lsconf(\apr),\anaddr) = (\aval,\atr')$ \quad $\mathit{rval} = \lsconf(\apr).\valconf[\areg\mapsto \aval]$}}% \quad
{\text{$s = \lsconf(\apr)[\valconf \mapsto \mathit{rval},\pcconf\mapsto \mathsf{next}(\pcconf)]$}}}
{\text{$(\lsconf,\msgsconf)\mpitrans{\loadact(\apr,\atr,\anaddr,\aval)} (\lsconf[\apr\mapsto s],\msgsconf)$}}\\[2em]
\dfrac
{\text{$\text{\plog{end}}\in\instrOf(\lsconf(\apr).\pcconf)$\quad
%$\mathit{store} = \lsconf(\apr).\storeconf[\anaddr\mapsto \mathit{last}(\mathit{log},\anaddr):\ \anaddr\in\mathbb{V}, \mathit{last}(\mathit{log},\anaddr)\neq\bot]$ \storeconf \mapsto \mathit{store},
$s = \lsconf(\apr)[\storeconfcopy \mapsto \epsilon,\pcconf\mapsto \mathsf{next}(\pcconf)]$}}%\quad
%{\text{}}}
{\text{$(\lsconf,\msgsconf)\mpitrans{\commitact(\apr,\atr)} (\lsconf[\apr\mapsto s],\msgsconf\cup\{(\atr,\lsconf(p).\txnwrsconf)\})$}}\\[2em]
\dfrac
{\splitdfrac{\text{$\tuple{t,\mathit{log}}\in \msgsconf$\quad
$\mathit{store} = \lsconf(\apr).\storeconf[\anaddr\mapsto \mathit{update}(\lsconf(\apr),\anaddr,\atr,\mathit{last}(\mathit{log},\anaddr)): \anaddr\in\mathbb{V}]$}}%\quad
{\text{ $s = \lsconf(\apr)[\storeconf \mapsto \mathit{store},\pcconf\mapsto \mathsf{next}(\pcconf)]$}}}
{\text{$(\lsconf,\msgsconf)\mpitrans{\storeact(\apr,\atr)} (\lsconf[\apr\mapsto s],\msgsconf)$}}
}
\caption{Transition rules for defining weak causal consistency semantics: $\buildsnapshot(\storeconf)$ returns a consistent snapshot of $\storeconf$.
$\mathit{cceval}(\lsconf(\apr),\anaddr)$ returns the pair $(\mathit{last}(\mathit{log},\anaddr),\atr)$ if $\mathit{last}(\mathit{log},\anaddr)\neq\bot$, and returns the pair $(\aval,\atr')$ in $\lsconf(\apr).\storeconfcopy(\anaddr)$, otherwise. $\mathit{update}(\lsconf(\apr),\anaddr,\atr,\mathit{last}(\mathit{log},\anaddr))$ returns the result of appending the pair  $(\mathit{last}(\mathit{log},\anaddr),\atr)$ to the set $\lsconf(\apr).\storeconf(\anaddr)$ after removing all pairs that contain values overwritten by $\atr$.}
\label{Table:MPRules-cc}
\end{figure}

The set of executions of $\aprog$ under weak causal consistency model, denoted by $\executionsconf_{\wcct{}}(\aprog)$, is the set of executions of $[\aprog]_{\wcct{}}$ satisfying transaction isolation and causal delivery. We denote by $\tracesconf(\aprog)_{\wcct{}}$ the set of traces of executions of a program $\aprog$ under weak causal consistency.

Fig.~\ref{firstexamplewcc} shows an execution under \wcct{}, which is not possible under \ccvt{} and \scct{} because $\atr 3$ and $\atr 4$ read $2$ and $1$, respectively. Since the transactions $\atr_1$ and $\atr_2$ are concurrent, $\apr_2$ stores both values $2$ and $1$ written by these transactions. A read of $x$ can return any of these two values.

\subsection{Execution Summary} 

Let $\rho$  be an execution under $\textsf{X} \in \{\ccvt{},\ \scct{},\ \wcct{}\}$, a sequence $\tau$ of events $\issueact(\apr,\atr)$ and $\storeact(\apr,\atr)$ with $\apr\in\mathbb{P}$ and $\atr\in\mathbb{T}$ is called a \emph{summary of $\rho$} if it is obtained from $\rho$ 
by substituting every sub-sequence of transitions in $\rho$ delimited by a $\beginact$ and an $\commitact$ transition, with a single ``macro-event'' $\issueact(\apr,\atr)$. For example, $\issueact(\apr 1,\atr 1)\cdot\issueact(\apr 2,\atr 3)\cdot\storeact(\apr 1,\atr 3)\cdot\storeact(\apr 2,\atr 1)\cdot\issueact(\apr 2,\atr 4)\cdot\issueact(\apr 1,\atr 2)$ is a summary of the execution in Fig.~\ref{firstexamplecm}. 
%A summary satisfies causal delivery if it is the summary of an execution that satisfies causal delivery (the notion of causal delivery can be applied directly to summaries by replacing the occurrences of $\beginact(\apr,\atr)$ and $\commitact(\apr,\atr)$ with $\issueact(\apr,\atr)$ in the definition from Section~\ref{sec:operationalModel1}).

We say that a transaction $\atr$ in $\rho$ performs an \emph{external read} of a variable $\anaddr$ if $\rho$ contains an event $\loadact(\apr,\atr,\anaddr,\aval)$ which is not preceded by a write on $\anaddr$ of $\atr$, i.e., an event $\issueact(\apr,\atr,\anaddr,\aval)$. Under \scct{} and \wcct{}, a transaction $\atr$ \emph{writes} a variable $\anaddr$ if $\rho$ contains an event $\issueact(\apr,\atr,\anaddr,\aval)$, for some $\aval$. In Fig.~\ref{firstexamplecm}, both $\atr 2$ and $\atr 4$ perform external reads and $\atr 2$ writes to $y$. A transaction $\atr$ executed by a process $\apr$ \emph{writes $\anaddr$ at process $\apr'$} if $\atr$ writes $\anaddr$ and $\rho$ contains an event $\storeact(\apr',\atr)$ (e.g., in Fig.~\ref{firstexamplecm}, $\atr 1$ writes $\anaddr$ at $\apr 2$).
Under \ccvt{}, we say that a transaction $\atr$ executed by a process $\apr$ \emph{writes $\anaddr$ at process $\apr'$} if $\atr$ writes $\anaddr$ and $\rho$ contains an event $\storeact(\apr',\atr)$ which is not preceded by an event $\storeact(\apr',\atr')$ with $\atr<\atr'$ and $\atr'$ writing $\anaddr$ (if it would be preceded by such an event then the write to $\anaddr$ of $\atr$ will be discarded). For example, in Fig.~\ref{firstexampleccv}, $\atr 1$ does \emph{not} write $\anaddr$ at $\apr 2$.

\subsection{Trace}

We define an abstract representation of executions that satisfy transaction isolation\footnote{We refer collectively to executions in $[\aprog]_\textsf{X}\mbox{ with }\textsf{X} \in \{\ccvt{},\ \scct{},\ \wcct{}\}$.}, called \emph{trace}. Essentially, a trace contains the summary of an execution (it forgets the order in which shared-variables are accessed inside a transaction) and several happens-before relations between events in its summary which record control-flow dependencies, the order between transactions issued in the same process, and data-flow dependencies, e.g. which transaction wrote the value read by another transaction. 
% forgets the order in which shared-variables are accessed inside a transaction, and \textcolor{red}{the order between transactions accessing different variables}. The trace of an execution $\rho$ is obtained by adding several standard relations between events in its summary which record the data-flow, e.g. which transaction wrote the value read by another transaction.

More precisely, the \emph{trace} of an execution $\rho$ is a tuple $\traceof{\rho} = (\tau, \po, \rfo, \sto, \cfo,\sametro)$ where $\tau$ is the summary of $\rho$, $\po$ is the \emph{program order}, which relates any two issue events $\issueact(\apr,\atr)$ and $\issueact(\apr,\atr')$ that occur in this order in $\tau$, $\rfo$ is
the \emph{write-read} relation (also called \emph{read-from}), which relates events of two transactions $\atr$ and $\atr'$ such that $\atr$ writes a value that $\atr'$ reads, $\sto$ is the \emph{write-write} order  (also called store-order), which relates events of two transactions that write to the same variable, $\cfo$ is the \emph{read-write} relation (also called \emph{conflict}), which relates events of two transactions $\atr$ and $\atr'$ such that $\atr$ reads a value overwritten by $\atr'$, and $\sametro$ is the \emph{same-transaction} relation, which relates events of the same transaction.

\begin{defi}[Trace]\label{def:traces}
Formally, the \emph{trace} of an execution $\rho$ satisfying transaction isolation is $\traceof{\rho} = (\tau, \po, \rfo, \sto, \cfo, \sametro)$ where $\tau$ is a summary of $\rho$, and
\begin{itemize}[leftmargin=1.2cm]%[topsep=0pt]
\item[$\po$:] relates the issue and store events $\issueact(\apr,\atr)$ and $\storeact(\apr,\atr)$ of $\atr$ and subsequently, the event $\storeact(\apr,\atr)$ with any issue event $\issueact(\apr,\atr')$ that occurs after it in $\tau$.
\item[$\rfo$:] relates  any store and issue events $\event_1=\storeact(\apr,\atr)$ and $\event_2=\issueact(\apr,\atr')$ that occur in this order in $\tau$ such that $\atr'$ performs an external read of $\anaddr$, and $\event_1$ is the last event in $\tau$ before $\event_2$ such that $\atr$ writes $\anaddr$ at $\apr$. To make the shared variable $\anaddr$ explicit, we may use $\rfo(x)$ to name the relation between $\event_1$ and $\event_2$.
\item[$\sto$:] relates events of two transactions that write to the same variable. More precisely, $\sto$ relates  any two store events $\event_1=\storeact(\apr,\atr_1)$ and $\event_2=\storeact(\apr,\atr_2)$ that occur in this order in $\tau$ provided that $\atr_1$ and $\atr_2$ both write the same variable $\anaddr$, and if $\rho$ is an execution under causal convergence, then $\atr_1$ and $\atr_2$ writes $\anaddr$ at $\apr$, and $\atr_1<\atr_2$.
To make the shared variable $\anaddr$ explicit, we may use $\sto(x)$ to name the relation between $\event_1$ and $\event_2$.
\item[$\cfo$:] relates events of two distinct transactions $\atr$ and $\atr'$ such that $\atr$ reads a value that is overwritten by $\atr'$. Formally, $\cfo(\anaddr)=\rfo^{-1}(\anaddr);\sto(\anaddr)$ (we use $;$ to denote the standard composition of relations) and $\cfo=\bigcup_{\anaddr\in\mathbb{V}} \cfo(x)$. If a transaction $\atr$ reads the initial value of $\anaddr$ then $\cfo(\anaddr)$ relates $\issueact(\apr,\atr)$ with every event $\storeact(\apr',\atr')$ with $\apr'\in\mathbb{P}$ of any other transaction $\atr'$ that writes to $\anaddr$ at $\apr'$. %For simplicity, we say that $\cfo(\anaddr)$ relates $\issueact(\apr,\atr)$ with the issue event $\issueact(\apr',\atr')$ of $\atr'$ (i.e., $(\issueact(\apr,\atr),\issueact(\apr',\atr'))\in\cfo(\anaddr)$). 
\item[$\sametro$:] relates issue events with store events of the same transaction. More precisely, $\sametro$ relates every event $\issueact(\apr,\atr)$ with every event $\storeact(\apr',\atr)$ with $\apr'\in\mathbb{P}$.
\end{itemize}
\end{defi}

\begin{figure}[t]
  \begin{minipage}[l]{0.67\textwidth}
  \scalebox{0.6}
  {
  \begin{tikzpicture}
  
    \node[shape=rectangle ,draw=none,font=\large] (A) at (0,0)  [] {$\issueact(p1,t1)$};
    \node[shape=rectangle ,draw=none,font=\large] (A1) at (1.85,0)  [] {$\storeact(p1,t1)$};
    \node[shape=rectangle ,draw=none,font=\large] (B) at (3.85,0)  [] {$\issueact(p2,t3)$};
    \node[shape=rectangle ,draw=none,font=\large] (B1) at (5.7,0)  [] {$\storeact(p2,t3)$};
    \node[shape=rectangle ,draw=none,font=\large] (C) at (7.55,0)  [] {$\storeact(p1,t3)$};
    \node[shape=rectangle ,draw=none,font=\large] (D) at (9.55,0)  [] {$\storeact(p2,t1)$};
    \node[shape=rectangle ,draw=none,font=\large] (E) at (11.55,0)  [] {$\issueact(p1,t2)$};
    \node[shape=rectangle ,draw=none,font=\large] (E1) at (13.4,0)  [] {$\storeact(p1,t2)$};
    \node[shape=rectangle ,draw=none,font=\large] (F) at (15.25,0)  [] {$\storeact(p2,t2)$};
    \node[shape=rectangle ,draw=none,font=\large] (G) at (17.3,0)  [] {$\issueact(p2,t4)$};

    \begin{scope}[every edge/.style={draw=black,very thick}]
    \path [->] (A1) edge [bend right=40,style={draw=red}] node [above,font=\large,xshift=-2.5mm,yshift=0.3mm] {$\sto$} (C);
    \path [->] (A) edge [bend right=40] node [below,font=\large] {$\po$} (A1);
    \path [->] (A) edge [bend left=40] node [above,font=\large,xshift=-3mm,yshift=-0.3mm] {$\sametro$} (A1);
    \path [->] (A) edge [bend left] node [above,font=\large,xshift=-3mm,yshift=-0.7mm] {$\sametro$} (D);
    \path [->] (A1) edge [bend left=30] node [above,font=\large] {$\po$} (E);
    \path [->] (B) edge [bend right=40] node [below,font=\large,xshift=7mm,yshift=1.7mm] {$\po$} (B1);
    \path [->] (B) edge [bend left=40] node [above,font=\large,xshift=-3mm,yshift=-0.3mm] {$\sametro$} (B1);
    \path [->] (B) edge [bend left] node [above,font=\large] {$\sametro$} (C);
    \path [->] (B1) edge [bend left=22] node [above,font=\large] {$\po$} (G);
    \path [->] (B) edge [bend right=40,style={draw=red}] node [above,font=\large,xshift=4mm,yshift=0.5mm] {$\cfo$} (D);
    \path [->] (B1) edge [bend right=20] node [above,font=\large] {$\rfo$} (G);
    \path [->] (E) edge [bend left] node [above,font=\large] {$\sametro$} (F);
    \path [->] (E) edge [bend right=40] node [below,font=\large] {$\po$} (E1);
    \path [->] (E) edge [bend left=40] node [above,font=\large,xshift=-3mm,yshift=-0.3mm] {$\sametro$} (E1);
    \path [->] (F) edge [bend left=60] node [above,font=\large] {$\rfo$} (G);
    \end{scope}
  
  \end{tikzpicture}}
  \end{minipage}
  \hfill
  \begin{minipage}[r]{0.22\textwidth}
  \scalebox{0.6}
  {
  \begin{tikzpicture}
  
    \node[shape=rectangle ,draw=none,font=\large] (A) at (0,0)  [] {$t1$};
    \node[shape=rectangle ,draw=none,font=\large] (B) at (1.7,0)  [] {$t3$};
    \node[shape=rectangle ,draw=none,font=\large] (C) at (3.4,0)  [] {$t2$};
    \node[shape=rectangle ,draw=none,font=\large] (D) at (5.1,0)  [] {$t4$};

    \begin{scope}[
                every edge/.style={draw=black,very thick}]
    \path [->] (A) edge [bend right,style={draw=red}] node [below,font=\large] {$\sto$} (B);
    \path [->] (A) edge [bend left=30] node [above,font=\large] {$\po$} (C);
    \path [->] (B) edge [bend left] node [above,font=\large] {$\po$} (D);
    \path [->] (B) edge [bend right=65,style={draw=red}] node [above,font=\large] {$\cfo$} (A);
    \path [->] (B) edge [bend right] node [below,font=\large] {$\rfo$} (D);
    \path [->] (C) edge [bend left=70] node [above,font=\large] {$\rfo$} (D);
    \end{scope}
  \end{tikzpicture}}
  \end{minipage}
  \caption{The trace of the execution in Fig.~\ref{firstexampleccv} and its transactional happens-before.}
  \label{firstexampleccvtrace}
\end{figure}

The following result states an important property of the store order relation $\sto$ that is enforced by the \ccvt{} semantics. It holds because the writes in different transactions are applied by different processes in the same order given by their timestamps, when visible (delivered) to those processes. 

\begin{lem}\label{lemma:CcvProperty}
Let $\tau\in \tracesconf_{\ccvt{}}(\aprog)$ be a trace. If $(\storeact(\apr_{0},\atr_0),\storeact(\apr_{0},\atr_1)) \in \sto(\anaddr)$, then for every other process $\apr$, $(\storeact(\apr,\atr_1),\storeact(\apr,\atr_0)) \not\in \sto(\anaddr)$.
%where $\storeact(\apr_{0},\atr_1)$ is the store of $\atr_1$ issued by $\apr_1$, then $(\issueact(\apr_{1},\atr_1),\storeact(\apr_{1},\atr_0)) \not\in \sto(\anaddr)$ and for every other process $\apr \not\in \{\apr_0, \apr_1\}$ we have that $(\storeact(\apr,\atr_1),\storeact(\apr,\atr_0)) \not\in \sto(\anaddr)$.
\end{lem}

We define the \emph{happens-before} relation $\hbo$ as the transitive closure of the union of all the relations in the trace, i.e., $\hbo = (\po \cup \rfo \cup \sto \cup \cfo \cup \sametro)^{+}$. Since we reason about only one trace at a time, we may say that a trace is simply a summary $\tau$, keeping the relations implicit. The trace of the \ccvt{} execution in Fig.~\ref{firstexampleccv} is shown on the left of Fig.~\ref{firstexampleccvtrace}. $\tracesconf(\aprog)_{\textsf{X}}$ denotes the set of traces of executions of a program $\aprog$ under $\textsf{X} \in \{\ccvt{},\ \scct{},\ \wcct{}\}$.

For readability, we write $\event_1\rightarrow_{\hbo}\event_2$ instead of $(\event_1,\event_2)\in\hbo$ and $\event_1$ and $\event_2$ can be either $\issueact(\apr,\atr)$ or $\storeact(\apr,\atr)$. We use the notation $\event_1\rightarrow_{\hbo^{1}}\event_2$ (resp., $(\event_1,\event_2)\in\hbo^{1}$) to denote $(\event_1,\event_2)\in(\po \cup \sto \cup \rfo \cup \sametro\cup \cfo)$.
  
The \emph{causal order} $\viso$ of a trace $\atrace=(\tau, \po, \rfo, \sto, \cfo,\sametro)$ is the transitive closure of the union of the program order, write-read relation, and the same-transaction relation, i.e., $\viso=(\po \cup \rfo\cup \sametro)^+$. For readability, we write $\event_1\rightarrow_{\viso}\event_2$ instead of $(\event_1,\event_2)\in\viso$.

Let $\atr_1$ and $\atr_2$  be two transactions issued in a trace $\atrace$ that originate from two different processes $\apr_1$ and $\apr_2$, respectively. If $(\issueact(\apr_1,\atr_1), \issueact(\apr_2,\atr_2)) \not\in \viso$ and $(\issueact(\apr_2,\atr_2), \issueact(\apr_1,\atr_1)) \not\in \viso$, then $\atr_1$ and $\atr_2$ are called  \emph{concurrent} transactions.

The happens-before relation between events is extended to transactions as follows: a transaction $\atr_1$ \emph{happens-before} another transaction $\atr_2\neq \atr_1$ if the trace $\atrace$ contains an event of transaction $\atr_1$ which happens-before an event of $\atr_2$. The happens-before relation between transactions is denoted by $\hbo_t$ and called \emph{transactional happens-before} (an example is given on the right of Fig.~\ref{firstexampleccvtrace}).

\begin{rem}\label{rem:axiomatic} 
The operational models of causal consistency we described are equivalent to the axiomatic models defined in~\cite{DBLP:conf/popl/BouajjaniEGH17}. These axiomatic models are defined as a set of constraints on abstractions of executions, called \emph{histories}, that consist of a set of read and write operations along with a program order, denoted by $\po'$, and a read-from relation, denoted by $\rfo'$: $\po'$ relates operations in the same process and $\rfo'$ associates every read operation to the write operation which wrote the read value. For instance, the axiomatic model of \wcct{} requires that the union of $\po'$ and $\rfo'$ (denoted $\viso'$) is acyclic\footnote{This constraint corresponds to the absence of the CyclicCO bad pattern in~\cite{DBLP:conf/popl/BouajjaniEGH17}.}, and its composition with a variation of the conflict relation, denoted by $\cfo'$, ($(a,b) \in \cfo'\footnote{$b$ is overwriting the value $a$ is reading.}\mbox{ iff } \exists\ c.\ (c, b) \in \viso' \land (c,a) \in \rfo'$) is irreflexive\footnote{This constraint corresponds to the absence of the WriteCORead bad pattern in~\cite{DBLP:conf/popl/BouajjaniEGH17}.}. These models can be extended easily to histories that contain transactions instead of operations by adapting the above relations. Note that every  program trace (cf. Definition~\ref{def:traces}) can be ``projected'' to a history where issue and store events from the same transaction in the trace are mapped to a single transaction in the history. Also, the read-from and the program order between trace events are mapped to the $\rfo'$ and $\po'$ of the history. 

To show equivalence between these models, it is sufficient to show that (1) every history corresponding to a trace in the operational model satisfies the constraints of the axiomatic model, and (2) every history that is valid under the axiomatic model is the ``projection'' of a trace of the operational model. 
% its history as defined above is valid under the axiomatic model, and for every history that is valid under the axiomatic model we can construct a trace under the semantics model such that this history corresponds to the history obtained from the trace. 
For instance, for \wcct{}, it is easy to see that the relation $\viso'=\po'\cup \rfo'$ in a history that is the projection of a trace $\tau\in \tracesconf_{\wcct{}}(\aprog)$ is acyclic because the causal order $\viso$ in $\tau$ is.
% based on $\viso$ is irrflexive since $\viso$ is. 
Also, the proof that $\viso';\cfo'$ is irreflexive can be derived easily by contradiction (for instance, if $(a,b)\in \cfo'$ and $(b,a)\in \viso'$, then there exists $c$ such that $(c,b)\in \viso'$ which means that by causal delivery, $a$ can never read the value written by $c$). 
\end{rem}

%\subsection{Relating Release-Acquire and Causal Consistency}\label{sec:RA}

%\begin{theorem} \label{theorem:CCvSRA} 
%$\hist= (\mathcal{T},\pot,\getlabel)$ is \axccvt{} iff $\hist$ is \axsra{}.
%\end{theorem}

%!TEX root = draft.tex
\section{Write-Write Race Freedom}\label{WWRF}

We say that an execution $\rho$ has a \emph{write-write race} on a shared variable $\anaddr$ if there exist two concurrent transactions $\atr_1$ and $\atr_2$ that were issued in $\rho$ and each transaction contains a write to the variable $\anaddr$.
We call $\rho$ write-write race free if there is no variable $\anaddr$ such that $\rho$ has a write-write race on $\anaddr$.
Also, we say a program $\aprog$ is \emph{write-write race free} under a consistency semantics $\textsf{X} \in \{\ccvt{},\ \scct{},\ \wcct{}\}$ iff for every $\rho \in \executionsconf_{\textsf{X}}(\aprog)$, $\rho$ is  write-write race free.

%\begin{rem}
%    We defined the causal order with using the same-transaction relation and program and read-from orders which means that even if a transaction is delivered before the other one is issued this doesn't necessarily imply that the two are causally related.
%\end{rem}

We show that if a given program has a write-write race under one of the three causal consistency models then it must have a write-write race under the remaining two.
The intuition behind this is that the three models coincide for programs without write-write races. Indeed, without concurrent transactions that write to the same variable, every process local valuation of a shared variable will be a singleton set under \wcct{} and no process will ever discard a write when applying an incoming transaction log under \ccvt{}.

\begin{thm}\label{theorem:wwraces}
Given a program $\aprog$ and two consistency semantics $\textsf{X}, \textsf{Y} \in \{\ccvt{},\ \scct{},\ \wcct{}\}$,
$\aprog$  has a write-write race under \textsf{X} iff  $\aprog$  has a write-write race under \textsf{Y}.
\end{thm}

\begin{proof}
Since \wcct{} is weaker than both \ccvt{} and \scct{}, it is sufficient to prove the following two cases: (1) if $\aprog$  has a write-write race under \wcct{}, then  $\aprog$  has a write-write race under \ccvt{} and (2) if $\aprog$ has a write-write race under \wcct{}, then  $\aprog$  has a write-write race under \scct{}. 

We prove the first case by induction on the number of transactions in $\aprog$. The second case can be proved in a similar way.

\noindent
\textbf{Base case:}  $\aprog$ is constituted of two transactions $\atr_1$ and $\atr_2$. 
Assume that $\aprog$ has a write-write race under \wcct{} then the transactions $\atr_1$ and $\atr_2$ must originate from different processes. Thus, in any trace $\tau$ of $\aprog$ under \ccvt{} where the transactions $\atr_1$ and $\atr_2$ are executed concurrently we will have a write-write race between these two transactions. Thus, $\aprog$  has a write-write race under \ccvt{}. 

\noindent
\textbf{Induction step:} If $n > 2$ is the number of transactions in $\aprog$, we assume that for any program $\aprog'$ with $n' < n$ transactions, if $\aprog'$  has a write-write race under \wcct{}, then  $\aprog'$  has a write-write race under \ccvt{}. 
Assume that $\aprog$  has a write-write race under \wcct{}. Let $\tau$ be a trace of $\aprog$  under \wcct{} where we have a write-write race between two transactions $\atr_1$ and $\atr_2$ that were issued by processes $\apr_1$ and $\apr_2$, respectively. 
Executing $\atr_1$ and $\atr_2$ concurrently while writing to a common variable is not possible under \ccvt{} only if the writes were enabled by some events that occurred before $\atr_1$ and $\atr_2$ under \wcct{} and are not possible under \ccvt{}. However, based on the semantic models of both \wcct{} and \ccvt{}, if all the transactions that write to common variables are causally related then such events cannot occur under \wcct{} but not \ccvt{}. Thus, we must have two other transactions $\atr'_1$ and $\atr'_2$ of $\aprog$ that were executed concurrently in $\tau$ under \wcct{} and occurred before $\atr_1$ (or $\atr_2$ or both) which write to a common variable. 
Without loss of generality, let $\aprog_1$ be the program resulting from removing the transaction $\atr_1$ from $\aprog$. 
We know that $\aprog_1$ admits a trace $\tau_1$ under \wcct{} where the transactions $\atr'_1$ and $\atr'_2$ are involved in a data race. 
Also, the size of $\aprog_1$  is $n-1 < n$. Thus, from the induction hypothesis we get that $\aprog_1$ has a write-write race under \ccvt{}. Because adding a new transaction to $\aprog_1$ will not eliminate existing data races, $\aprog$ has a write-write race under \ccvt{} as well.
%
% Thus, we differentiate between the possibility of executing $\atr'_1$ and $\atr'_2$ concurrently under \ccvt{} or not, if it is not possible we repeat the same mechanism above for $\atr'_1$ and $\atr'_2$ instead of $\atr_1$ and $\atr_2$. Since we are dealing with a finite prefix of the trace (i.e., events that were executed before $(\apr_1,\atr_1)$ and $(\apr_2,\atr_2)$) and we are backtracking, we will iteratively obtain the proof.
\end{proof}

The following result shows that indeed, the three causal consistency models coincide for programs which are write-write race free under any one of these three models.
% also show that if a given program has no write-write race under either one of three causal consistency models then the sets of executions of the underlying program under the three causal consistency semantics coincide.

\begin{thm}\label{theorem:traceswwracesOrig}
Let $\aprog$ be a program. Then, $\executionsconf_{\wcct{}}(\aprog) = \executionsconf_{\ccvt{}}(\aprog) = \executionsconf_{\scct{}}(\aprog)$ iff $\aprog$ has no write-write race under neither \wcct{}, \scct{}, and \ccvt{}.
\end{thm}

\begin{proof}
Left-to-right direction: By Theorem~\ref{theorem:wwraces}, it is sufficient to prove that $\aprog$ has no write-write race under \scct{}. Suppose by contradiction that $\aprog$ has a write-write race under \scct{}. Then, there must exist a trace $\tau \in \tracesconf_{\wcct{}}(\aprog)$ such that we have two concurrent transactions $\atr_1$ and $\atr_2$ that are issued in $\tau$ and write to a variable $\anaddr$. Assume w.l.o.g that the issue event of $\atr_1$ occurs before the issue event of $\atr_2$ in $\tau$. Since $\atr_1$ and $\atr_2$ are concurrent in $\tau$, the issue event of $\atr_1$ and the store events of $\atr_2$ are commutative, and the issue event of $\atr_2$ and the store events of $\atr_1$ are commutative. 
Then, $\tau' = \alpha \cdot \issueact(\apr_1,\atr_{1}) \cdot \storeact(\apr_1,\atr_{1})  \cdot\beta\cdot \issueact(\apr_2,\atr_{2}) \cdot \storeact(\apr_2,\atr_{2}) \cdot \storeact(\apr_1,\atr_{2}) \cdot \storeact(\apr_2,\atr_{1})$ where $\alpha$ and $\beta$ are sequences of events in $\tau$ that $\atr_1$ and $\atr_2$ causally depend on (since we are not interested in other events)\footnote{Note that other cases such as $\tau' = \alpha \cdot \issueact(\apr_1,\atr_{1})  \cdot\beta\cdot \issueact(\apr_2,\atr_{2}) \cdot \storeact(\apr_2,\atr_{2}) \cdot \storeact(\apr_1,\atr_{2}) \cdot \storeact(\apr_1,\atr_{1})$ implies that $\tau'' = \alpha \cdot \issueact(\apr_1,\atr_{1})  \cdot \storeact(\apr_1,\atr_{1}) \cdot\beta\cdot \issueact(\apr_2,\atr_{2}) \cdot \storeact(\apr_2,\atr_{2}) \cdot \storeact(\apr_1,\atr_{2}) \cdot \storeact(\apr_2,\atr_{1})$ is a trace of $\aprog$ as well since all events in $\beta$ are not causally dependent on $\atr_1$.}, is a trace of $\aprog$ under \scct{}.  In $\tau'$, both store events $\storeact(\apr_2,\atr_{1})$ and $\storeact(\apr_1,\atr_{2})$ do not discard any writes (guaranteed under \scct{}). Therefore, $(\storeact(\apr_1,\atr_{1}),\storeact(\apr_1,\atr_{2})) \in \sto(\anaddr)$ and $(\storeact(\apr_2,\atr_{2}),\storeact(\apr_2,\atr_{1})) \in \sto(\anaddr)$ since both $\atr_1$ and $\atr_2$ write to $\anaddr$. However, it is impossible to obtain $\tau'$ under \ccvt{} as we cannot have $(\storeact(\apr_2,\atr_{2}),\storeact(\apr_2,\atr_{1})) \in \sto(\anaddr)$ if $(\storeact(\apr_1,\atr_{1}),\storeact(\apr_1,\atr_{2})) \in \sto(\anaddr)$ which leads to a contradiction ($\aprog$ has different sets of traces under \scct{} and \ccvt{}).

Right-to-left direction: It is sufficient to prove the following two cases: if $\tau$ has no write-write race under \wcct{} then $\tau \in \tracesconf_{\wcct{}}$ implies $\tau \in \tracesconf_{\scct{}}$ and $\tau \in \tracesconf_{\ccvt{}}$ ($\tracesconf_{\ccvt{}}(\aprog) \subseteq \tracesconf_{\wcct{}}(\aprog)$ and $\tracesconf_{\scct{}}(\aprog) \subseteq \tracesconf_{\wcct{}}(\aprog)$ hold by definition). 

Let $\tau \in \tracesconf_{\wcct{}}$ be a trace under \wcct{}. Then, $\tau$ satisfies transactions isolation and causal delivery. It is important to notice that if $\tau$ has no write-write race then the contents of $store$ at a given variable will contain a single value at any time during $\tau$. This implies that $store$ can be simulated by a single value memory which does not discard writes. Thus, we obtain a program semantics that is the same as the one for \scct{}.
Thus, $\tau$ is also a trace of $\aprog$ under \scct{}. To prove that $\tau \in \tracesconf_{\ccvt{}}$, we also need to ensure that the transitive closure of store order in $\tau$ is acyclic which is enough to guarantee the existence of a total arbitration between transactions which is ensured by \ccvt{} semantics. 
Suppose by contradiction that the transitive closure of store order is cyclic then there must exist a sequence of events $\event_1\cdot\event_2\cdot\ldots\event_n$ in $\tau$  such that $(\event_{i},\event_{i+1})\in \sto$, for all $1\leq i\leq n-1$ and $(\event_{n},\event_{1})\in \sto$.  Since $\tau$ has no write-write races then $(\event_{i},\event_{i+1})\in \sto$ implies that the issue events corresponding to $\event_{i}$ and $\event_{i+1}$ must be related by causal ordered (since the corresponding transactions must be causally related to prevent concurrency which will lead to write-write races for transactions that write to a common variable). For all $i$ s.t. $1\leq i\leq n-1$, let $\event'_{i}$ and $\event'_{i+1}$ denote these issue events then $(\event'_{i},\event'_{i+1})\in \viso$ which implies that the causal order $\viso$ is cyclic. This is a contradiction since it is not possible under \wcct{}. Thus, there exists a total order between transactions in $\tau$ that includes both the causal order and the transitive closure of store order. Thus, $\tau$ is also a trace of $\aprog$ under \ccvt{}. 
\end{proof}

%The proof of this theorem is the combination of the two Lemmas \ref{lemma:wwr1} and \ref{lemma:wwr2}.
%Note that if we serialize transactions that write to the same location using some synchronization primitives, then the set of executions of the resulting program under any of three causal consistency semantics is the intersection of the execution sets of the original program under the three causal consistency semantics.
%\begin{corollary}
%Given a program $\aprog$ which has write-write races. Let $\aprog'$ be a write-write synchronized restriction of $\aprog$. Then, $\executionsconf_{\wcct{}}(\aprog') = \executionsconf_{\ccvt{}}(\aprog') =  \executionsconf_{\scct{}}(\aprog') = \executionsconf_{\scct{}}(\aprog) \cap \executionsconf_{\ccvt{}}(\aprog) \cap \executionsconf_{\wcct{}}(\aprog)$.
%\end{corollary} 

%!TEX root = draft.tex
\section{Program Robustness}\label{sec:PRs}

\subsection{Program Semantics Under Serializability}

The semantics of a program under serializability~\cite{DBLP:journals/jacm/Papadimitriou79b} can be defined using a transition system where the configurations keep a single shared-variable valuation (accessed by all processes) with the standard interpretation of read or write statements. Each transaction executes in isolation. 
Alternatively, the serializability semantics can be defined as a restriction of $[\aprog]_{\textsf{X}}$, $\textsf{X} \in \{\ccvt{},\ \scct{},\ \wcct{}\}$, to the set of executions where each transaction is \emph{immediately} delivered to all processes, i.e., each event $\commitact(\apr,\atr)$ is immediately followed by all $\storeact(\apr',\atr)$ with $\apr'\in\mathbb{P}$. Such executions are called \emph{serializable} and the set of serializable executions of a program $\aprog$ is denoted by $\executionsconf_{\serc{}}(\aprog)$. The latter definition is easier to reason about when relating executions under causal consistency and serializability, respectively.

Given a trace $\atrace=(\tau, \po, \rfo, \sto, \cfo,\sametro)$ of a serializable execution, we have that every event $\issueact(\apr,\atr)$ in $\tau$ is immediately followed by all $\storeact(\apr',\atr)$ with $\apr'\in\mathbb{P}$. For simplicity, we write $\tau$ as a sequence of ``atomic macro-events'' $(\apr,\atr)$ where $(\apr,\atr)$ denotes a sequence $\issueact(\apr,\atr)\cdot\storeact(\apr,\atr)\cdot \storeact(\apr_1,\atr)\cdot\ldots\cdot \storeact(\apr_n,\atr)$ with $\mathbb{P}=\{\apr,\apr_1,\ldots,\apr_n\}$. We say that $\atr$ is \emph{atomic}. In Fig.~\ref{firstexampleccvtrace}, $\atr 3$ is atomic and we can use $(\apr 2,\atr 3)$ instead of $\issueact(\apr 2,\atr 3)\cdot\storeact(\apr 2,\atr 3)\cdot\storeact(\apr 1,\atr 3)$.

The following result characterizes traces of serializable executions, and follows from previous works~\cite{Adya99,DBLP:journals/toplas/ShashaS88} that considered a notion of history/trace that corresponds to our notion of transactional happens-before. The transactional happens-before of any trace under $\serc{}$ is acyclic, and conversely, any trace obtained under a weaker semantics $\textsf{X} \in \{\ccvt{},\ \scct{},\ \wcct{}\}$ with an acyclic transactional happens-before can be transformed into a trace under $\serc{}$ by successive swaps of consecutive events in its summary, which are not related by happens-before (the happens-before relations remain the same). Indeed, note that multiple executions/traces can have the same (transactional) happens-before (an example for traces is given in Fig.~\ref{fig:onetracetwoexecutions}). In particular, it is possible that a trace $\atrace$ produced by a variation of causal consistency has an acyclic transactional happens-before even though $\issueact(\apr,\atr)$ events are not immediately followed by the corresponding $\storeact(\apr',\atr)$ events. However, $\atrace$ would be equivalent, up to reordering of consecutive summary events that are not related by happens-before to a serializable trace. 

%consisting of a set of transactions along with a happens-before relation. 
%It can be extended to our notion of trace easily by mapping the issue event and store events of the same transaction to a single transaction in the history and taking $\hbo_t$ to be the happens-before of the history.  

\begin{thmC}[\cite{Adya99,DBLP:journals/toplas/ShashaS88}]\label{th:acyclicity}
 For any trace $\atrace\in \tracesconf_{\serc{}}(\aprog)$, the transactional happens-before of $\atrace$ is acyclic. Moreover, for any trace $\atrace=(\tau, \po, \rfo, \sto, \cfo,\sametro)\in \tracesconf_{\textsf{X}}(\aprog)$ with $\textsf{X} \in \{\ccvt{},\ \scct{},\ \wcct{}\}$, if  the transactional happens-before of $\atrace$ is acyclic, then there exists a permutation $\tau'$ of $\tau$ such that $(\tau', \po, \rfo, \sto, \cfo,\sametro)\in \tracesconf_{\serc{}}(\aprog)$.
 
\end{thmC}

As a consequence of Theorem~\ref{th:acyclicity}, we define a trace $\atrace$ to be \emph{serializable} if it \emph{has the same happens-before relations} as a trace of a serializable execution. Let $\tracesconf_{\serc{}}(\aprog)$ denote the set of serializable traces of a program $\aprog$. 
%The following is a direct consequence of Theorem~\ref{th:acyclicity:interm} and of the fact that events that are not related by happens-before can be reordered without ``breaking'' the feasibility of the execution.

%\begin{thm}\label{th:acyclicity}
% A trace is serializable iff its transactional happens-before $\hbo_t$ is acyclic.
%\end{thm}

%Given a serializable trace $\atrace=(\tau, \po, \rfo, \sto, \cfo,\sametro)$ we have that every event $\issueact(\apr,\atr)$ in $\tau$ is immediately followed by all $\storeact(\apr',\atr)$ with $\apr'\in\mathbb{P}$. For simplicity, we write $\tau$ as a sequence of ``atomic macro-events'' $(\apr,\atr)$ where $(\apr,\atr)$ denotes a sequence $\issueact(\apr,\atr)\cdot\storeact(\apr,\atr)\cdot \storeact(\apr_1,\atr)\cdot\ldots\cdot \storeact(\apr_n,\atr)$ for some $\apr\in\mathbb{P}=\{\apr,\apr_1,\ldots,\apr_n\}$. We say that $\atr$ is ``atomic". In Fig.~\ref{firstexampleccvtrace}, $\atr 3$ is atomic and we can use $(\apr 2,\atr 3)$ instead of $\issueact(\apr 2,\atr 3)\storeact(\apr 2,\atr 3)\storeact(\apr 1,\atr 3)$.

\begin{figure}[t]
  \begin{minipage}[c]{0.31\textwidth}
  \begin{subfigure}{\linewidth}
\lstset{basicstyle=\ttfamily\scriptsize}
\begin{minipage}[c]{19mm}
\begin{lstlisting}[language=Java10]
    p1:
t1: [ r = y 
      x = 1]
\end{lstlisting}
\end{minipage}
\begin{minipage}[c]{2mm}
\footnotesize{$||$}
\end{minipage}
\begin{minipage}[c]{19mm}
\begin{lstlisting}[language=Java10]
    p2:
t2: [y = 2]
\end{lstlisting}
\end{minipage}
\label{fig:twoExesProg}
\end{subfigure}
  \end{minipage}
  \hfill
  \begin{minipage}[c]{0.33\textwidth}
  \begin{subfigure}{\linewidth}
  \scalebox{0.61}
  {\begin{tikzpicture}
  
   \node[shape=rectangle ,draw=none,font=\large] (A) at (0,0)  [] {$\issueact(\apr 1,\atr 1)$};
   \node[shape=rectangle ,draw=none,font=\large] (A1) at (1.9,0)  [] {$\storeact(\apr 1,\atr 1)$};
    \node[shape=rectangle ,draw=none,font=\large] (B) at (4,0)  [] {$(\apr 2,\atr 2)$};
    \node[shape=rectangle ,draw=none,font=\large] (C) at (6,0)  [] {$\storeact(\apr 2,\atr 1)$};
  
    \begin{scope}[ every edge/.style={draw=red,very thick}]
    \path [->] (A) edge  [bend left] node [above,font=\large] {$\cfo$} (B);
    \end{scope}
  \end{tikzpicture}}
  \label{fig:twoExes1}
  \end{subfigure}
  \end{minipage}
 \hfill
  \begin{minipage}[c]{0.33\textwidth}
  \begin{subfigure}{\linewidth}
  \scalebox{0.61}
  {
  \begin{tikzpicture}
  
   \node[shape=rectangle ,draw=none,font=\large] (A) at (0,0)  [] {$\issueact(\apr 1,\atr 1)$};
   \node[shape=rectangle ,draw=none,font=\large] (A1) at (1.9,0)  [] {$\storeact(\apr 1,\atr 1)$};
    \node[shape=rectangle ,draw=none,font=\large] (B) at (3.8,0)  [] {$\storeact(\apr 2,\atr 1)$};
    \node[shape=rectangle ,draw=none,font=\large] (C) at (6,0)  [] {$(\apr 2,\atr 2)$};
  
    \begin{scope}[ every edge/.style={draw=red,very thick}]
    \path [->] (A) edge [bend left] node [above,font=\large] {$\cfo$} (C);
    \end{scope}
  \end{tikzpicture}}
  \label{fig:twoExes2}
  \end{subfigure}
  \end{minipage}
  \caption{Two executions of the same serializable trace.}
  \label{fig:onetracetwoexecutions}
  \end{figure}

\subsection{Robustness Problem}\label{ssec:rob}

We consider the problem of checking whether the causally-consistent semantics of a program produces only serializable traces (it produces all serializable traces because every issue event can be immediately followed by all  the corresponding store events). % as the serializability semantics.
\begin{defi}
A program $\aprog$ is called \emph{robust} against a semantics $\textsf{X} \in \{\ccvt{},\ \scct{},\ \wcct{}\}$ iff $\tracesconf_{\textsf{X}}(\aprog)= \tracesconf_{\serc{}}(\aprog)$.
\end{defi}
%Since $\tracesconf_{\serc{}}(\aprog) \subseteq \tracesconf_{\textsf{X}}(\aprog)$, the problem of checking robustness of a program $\aprog$ against a semantics $\textsf{X}$ boils down to checking whether there exists 
A trace  $\atrace \in \tracesconf_{\textsf{X}}(\aprog) \setminus \tracesconf_{\serc{}}(\aprog)$ is called a \emph{robustness violation} (or \emph{violation}, for short). By Theorem~\ref{th:acyclicity}, the transactional happens-before $\hbo_t$ of $\atrace$ is cyclic.
%When $\textsf{X}$ is clear from the context, we say that a program is simply robust.
\begin{figure}[t]
\lstset{basicstyle=\ttfamily\scriptsize}
\begin{subfigure}{55mm}
\begin{minipage}[c]{25mm}
\begin{lstlisting}[language=Java10]
    p1
t1 [r1 = x    //0
    x  = r1 + 1]
\end{lstlisting}
\end{minipage}
\begin{minipage}[c]{2mm}
\footnotesize{$||$}
\end{minipage}
\begin{minipage}[c]{25mm}
\begin{lstlisting}[language=Java10]
    p2
t2 [r2 = x   //0
    x = r2 + 1]
\end{lstlisting}
\end{minipage}
\caption{Lost Update ($\mathsf{LU}$).}
\label{fig:rob0}
\end{subfigure}
\hspace{1mm}
\begin{subfigure}{49mm}
\begin{minipage}[c]{22mm}
\begin{lstlisting}[language=Java10]
    p1
t1 [x = 1	
    r1 = y] //0
\end{lstlisting}
\end{minipage}
\begin{minipage}[c]{2mm}
\footnotesize{$||$}
\end{minipage}
\begin{minipage}[c]{22mm}
\begin{lstlisting}[language=Java10]
    p2
t2 [y = 1
    r2 = x] //0
\end{lstlisting}
\end{minipage}
\caption{Store Buffering ($\mathsf{SB}$).}
\label{fig:rob1}
\end{subfigure}
\hspace{2mm}
\begin{subfigure}{40mm}
\begin{minipage}[c]{12mm}
\begin{lstlisting}[language=Java10]
[a = 1
 z = 1
 x = 1
 y = 1]
\end{lstlisting}
\end{minipage}
\begin{minipage}[c]{2mm}
\footnotesize{$||$}
\end{minipage}
\begin{minipage}[c]{21mm}
\begin{lstlisting}[language=Java10]
if (a == 1)
  [x = 2
   r1 = z  //0
   r2 = y  //1
   r3 = x] //2
\end{lstlisting}
\end{minipage}
\caption{Without transactions, non-robust against \ccvt{}.}
\label{fig:rob3}
\end{subfigure}

\begin{subfigure}{46mm}
\begin{minipage}[c]{19mm}
\begin{lstlisting}[language=Java10]
[a = 1
 x = 1
 r1 = x] //2	
\end{lstlisting}
\end{minipage}
\begin{minipage}[c]{2mm}
\footnotesize{$||$}
\end{minipage}
\begin{minipage}[c]{23mm}
\begin{lstlisting}[language=Java10]
if (a == 1)
  [x = 2
   r2 = x] //1
\end{lstlisting}
\end{minipage}
\caption{Without transactions, non-robust only against \scct{}.}
\label{fig:rob4}
\end{subfigure}
\hspace{3mm}
\begin{subfigure}{43mm}
\begin{minipage}[c]{19mm}
\begin{lstlisting}[language=Java10]
if ( * )
  [x = 1]
else
  [r1 = x]
\end{lstlisting}
\end{minipage}
\begin{minipage}[c]{2mm}
\footnotesize{$||$}
\end{minipage}
\hspace{.1mm}
\begin{minipage}[c]{19mm}
\begin{lstlisting}[language=Java10]
if ( * )
  [x = 2]
else
  [r2 = x]
\end{lstlisting}
\end{minipage}
\caption{Robust against both \scct{} and \ccvt{}.}
\label{fig:rob5}
\end{subfigure}
\hspace{3mm}
\begin{subfigure}{51mm}
\begin{minipage}[c]{13mm}
\begin{lstlisting}[language=Java10]
[x = 1]
[r1 = y]
\end{lstlisting}
\end{minipage}
\begin{minipage}[c]{3mm}
\footnotesize{$||$}
\end{minipage}
\hspace{.1mm}
\begin{minipage}[c]{19mm}
\begin{lstlisting}[language=Java10]
[r2 = x
 if (r2 == 1)
  y = 1]
\end{lstlisting}
\end{minipage}
\caption{Robust against both \scct{} and \ccvt{}.}
\label{fig:rob7}
\end{subfigure}
\caption{(Non-)robust programs. For non-robust programs, the read instructions are commented with the values they return in robustness violations.
The condition of ${\tt if}$-${\tt else}$ is checked inside a transaction whose demarcation is omitted for readability (${\tt *}$ denotes non-deterministic choice). }
\end{figure}

We discuss several examples of programs which are (non-) robust against both \scct{} and \ccvt{} or only one of them. Robustness violations are presented in terms of ``observable'' behaviors, tuples of values that can be read in the different transactions and that are not possible under the serializability semantics (they correspond to traces with acyclic transactional happens-before).
Fig.~\ref{fig:rob0} and Fig.~\ref{fig:rob1} show examples of programs that are \emph{not} robust against both \scct{} and \ccvt{}, which have also been discussed in the literature on weak memory models, e.g.~\cite{DBLP:journals/toplas/AlglaveMT14}. The execution of Lost Update under both \scct{} and \ccvt{} allows that the two reads of ${\tt x}$ in transactions $\atr 1$ and $\atr 2$ return $0$ although this cannot happen under serializability. Also, executing Store Buffering under both \scct{} and \ccvt{} allows that the reads of ${\tt x}$ and ${\tt y}$ return $0$ although this would not be possible under serializability. These values are possible because the transaction in each of the processes may not be delivered to the other process. %Then, the two transactions of the $\mathsf{IRIW}$ program that write a variable may be delivered in different orders to the processes that read the variables ${\tt x}$ and ${\tt y}$ (this is possible because the two transactions are not causally related). Therefore, the two processes on the right may read values which are not possible under serializability (the values are given as comments in Fig.~\ref{fig:rob2}).

Assuming for the moment that each instruction in Fig.~\ref{fig:rob3} and Fig.~\ref{fig:rob4} forms a different transaction, the values we give in comments show that the program in Fig.~\ref{fig:rob3}, resp., Fig.~\ref{fig:rob4}, is not robust against \ccvt{}, resp., \scct{}. The values in Fig.~\ref{fig:rob3} are possible assuming that the timestamp of the transaction ${\tt [x = 1]}$ is smaller than the timestamp of ${\tt [x = 2]}$ (which means that if the former is delivered after the second process executes ${\tt [x = 2]}$, then it will be discarded).
Moreover, enlarging the transactions as shown in Fig.~\ref{fig:rob3}, the program becomes robust against \ccvt{}.
The values in Fig.~\ref{fig:rob4} are possible under \scct{} because different processes do not need to agree on the order in which to apply transactions, each process applying the transaction received from the other process last. However, under  \ccvt{} this behavior is not possible, the program being actually robust against \ccvt{}.
As in the previous case, enlarging the transactions as shown in the figure leads to a robust program against \scct{}.

We end the discussion with several examples of programs that are robust against both \scct{} and \ccvt{}. These are simplified models of real applications reported in~\cite{Kurath}. The program in Fig.~\ref{fig:rob5} can be understood as the parallel execution of two processes that either create a new user of some service, represented abstractly as a write on a variable ${\tt x}$ or check its credentials, represented as a read of ${\tt x}$ (the non-deterministic choice abstracts some code that checks whether the user exists). Clearly this program is robust against both \scct{} and \ccvt{} since each process does a single access to the shared variable. Although we considered simple transactions that access a single shared-variable this would hold even for ``bigger'' transactions that access an arbitrary number of variables. The program in Fig.~\ref{fig:rob7} can be thought of as a process creating a new user of some service and reading some additional data in parallel to a process that updates that data only if the user exists. It is rather easy to see that it is also robust against both \scct{} and \ccvt{}.

%!TEX root = draft.tex
\section{Minimal Violations}\label{sec:RACC}

%The result concerning \wcct{} is derived from the one of \scct{}.
%We consider now the issue of checking robustness against the \ccvt{} and \scct{} variations of causal consistency.  
% Our reduction of the robustness checking problem uses two important facts: (1) reordering ``independent'' events doesn't change the fact whether a trace contains a cycle in the happens-before relation and (2) 
We define a class of robustness violations called \emph{minimal} violations. The particular shapes of these violations, that we determine through a series of results in this section, Section~\ref{sec:MVUCCV}, and Section~\ref{sec:MVUCM}, enables a polynomial-time reduction of robustness checking to a reachability problem in a program running under serializability. 

For simplicity, we use ``atomic macro-events'' $(\apr,\atr)$ even in traces obtained under causal consistency (recall that this notation was introduced to simplify serializable traces), i.e., we assume that any sequence of events formed of an issue $\issueact(\apr,\atr)$ followed immediately by all the store events $\storeact(\apr',\atr)$ is replaced by $(\apr,\atr)$. Then, all the relations that held between an event $\event$ of such a sequence and another event $\event'$, e.g., $(\event,\event')\in \po$, are defined to hold as well between the corresponding macro-event $(\apr,\atr)$ and $\event'$, e.g, $((\apr,\atr),\event')\in \po$.

\subsection{Happens-Before Through Relation}

To decide if two events in a trace are ``independent'' (or commutative) we use the information about the existence of a happens-before relation between the events. If two events are not related by  happens-before then they can be swapped while preserving the same happens-before. Thus, we extend the happens-before relation to obtain the \emph{happens-before through} relation as follows:

\begin{defi}\label{Definition:HBThrough}
Let $\tau = \alpha\cdot a\cdot \beta\cdot b\cdot \gamma$ be a trace where $a$ and $b$ are events (or atomic macro events), and $\alpha$, $\beta$, and $\gamma$ are sequences of events (or  atomic macro events) under a semantics $\textsf{X} \in \{\ccvt{},\ \scct{}\}$.
We say that \emph{$a$ happens-before $b$ through $\beta$} if there is a non empty sub-sequence $c_1 \cdots c_n$ of $\beta$ that satisfies:
$$c_i\rightarrow_{\hbo^{1}} c_{i+1}\quad\text{ for all }i\in[0, n]$$
where $c_0= a$, $c_{n+1}= b$.% and $\sametro?$ is the reflexive closure of $\sametro$.
\end{defi}

%\begin{rem}
%Note that in the above definition we use $\hbo^{1};\sametro?$ instead $\hbo^{1}$ because we may have two successive events $c_1$ and $c_2$ %that are $\hbo$-related but not $\hbo^{1}$-related. For instance, if $c_2$ is a store event of an issue event $c_3$ that occurs before %$c_1$, $c_1$ is an issue event, both $c_1$ and $c_3$ write to an address $x$, and $c_1$ reads the initial value of $x$. 
%In this case, $c_1$ is $\cfo$-related to $c_3$ and since $c_3$ is $\sametro$-related to $c_2$, then $c_1$ is $\hbo^{1};\sametro$-related %to $c_2$. On the other hand, in \ccvt{} the write to $x$ in $c_2$ are ignored because of the recent write to x in $c_1$ which means that %$c_1$ is not $\hbo^{1}$-related to $c_2$. Thus, $c_1$ and $c_2$ are $\hbo^{1};\sametro$-related but not $\hbo^{1}$-related. 
%\end{rem}

%An important property of the happens-before through relation is the stability of the relation when adding new events in the middle of a trace. 

%\begin{lem}\label{lemma:Insertion}
%Let $\tau=\alpha\cdot a\cdot\beta\cdot b\cdot\gamma$ and $\tau'=\alpha'\cdot a\cdot\beta'\cdot b\cdot\gamma'$ be two traces such that  $\projectionOf{\tau}{\apr}=\projectionOf{\tau'}{\apr}$ for every process $\apr$ (where $\projectionOf{\tau}{\apr}$ denotes the projection of the sequence $\tau$ on the set of events of process $\apr$), and $\beta$ a sub-sequence of $\beta'$. Then, if $a$ happens-before $b$ through $\beta$ then $a$ happens-before $b$ through $\beta'$.
%\end{lem}

The following result shows that any two events in a trace which are not related via the happens-before through relation can be reordered without affecting the happens-before or they can be placed one immediately after the other.

\begin{lem}\label{lemma:Reordering}
Let $\tau$ be a trace of a program $\aprog$ under a semantics $\textsf{X} \in \{\ccvt{},\ \scct{}\}$, and $a$ and $b$ be two events such that $\tau=\alpha\cdot a\cdot\beta \cdot b\cdot \gamma$. Then, one of the following holds:
\begin{enumerate}
  \item $a$ happens-before $b$ through $\beta$;
  \item $\tau'=\alpha\cdot\beta_1\cdot a\cdot b\cdot\beta_2\cdot\gamma\in \tracesconf_{\textsf{X}}(\aprog)$ where $(a,b)\in\hbo^{1}$ has the same happens-before as $\tau$;
  \item $\tau'=\alpha\cdot\beta_1\cdot b\cdot a\cdot\beta_2\cdot\gamma\in \tracesconf_{\textsf{X}}(\aprog)$ has the same happens-before as $\tau$.
\end{enumerate}
\end{lem}

\begin{proof}
%Let $\tau$ be a trace of a program $\aprog$ under a semantics $\textsf{X} \in \{\ccvt{},\ \scct{}\}$, and $a$ and $b$ two events such that $\tau=\alpha\cdot a\cdot\beta \cdot b\cdot \gamma$.
We prove that $\neg(1)\Rightarrow((2)\ \mathsf{or}\ (3))$ using induction on the size of $\beta$. %For $\length{\beta} = 0$ we have $\tau_{0} = \alpha\cdot a\cdot b\cdot \gamma$ where $a$ and $b$ are not $\hbo^{1}$-related which implies that $b$ can move left of $a$  which produces the trace $\tau'_{0} = \alpha\cdot b\cdot a\cdot \gamma$ that has the same happens-before as $\tau_{0}$. For $1 \leq \length{\beta}$, we use induction.

\noindent
\textbf{Base case:} 
If $\length{\beta} = 0$, then $\tau = \alpha\cdot a\cdot b\cdot \gamma$, which implies that $a$ does not happen-before $b$ through $\beta$ (by definition, $\beta$ cannot be empty). Thus, either $a$ and $b$ are $\hbo^{1}$-related, which corresponds to $(2)$, or $a$ and $b$ are not $\hbo^{1}$-related, which implies that $b$ can move to the left of $a$ producing the trace $\tau' = \alpha\cdot b\cdot a\cdot \gamma$ that has the same happens-before as $\tau$ and that corresponds to $(3)$.
%we have two possible cases either  or are not $\hbo^{1}$-related. The first case corresponds to $(2)$. Also, if $a$ and $b$ are not $\hbo^{1}$-related which implies that $b$ can move left of $a$  which produces the trace $\tau'_{0} = \alpha\cdot b\cdot a\cdot \gamma$ that has the same happens-before as $\tau_{0}$ which corresponds to $(3)$.

%$\length{\beta} = 1$ and $\beta = c$. Then $\tau_{1} = \alpha\cdot a\cdot c \cdot b\cdot \gamma$ where $a$ does not happen before $b$ through $c$. Therefore, either $a$ and $c$ are not $\hbo^{1}$-related  or $c$ and $b$ are not $\hbo^{1}$-related. If $a$ and $c$ are not $\hbo^{1}$-related, then we can swap the two obtaining  $\tau'_{1} = \alpha\cdot c\cdot a \cdot b\cdot \gamma$ which has the same happens-before as $\tau_{1}$. Also, if $a$ and $b$ are not $\hbo^{1}$-related then $b$ can move left of $a$ which produces the trace $\tau''_{1} = \alpha\cdot b\cdot a\cdot \gamma$ that has the same happens-before as $\tau'_{1}$. Otherwise, if $c$ and $b$ are not $\hbo^{1}$-related, then we can swap the two obtaining  $\tau'_{1} = \alpha\cdot a\cdot b \cdot c\cdot \gamma$ which has the same happens-before as $\tau_{1}$. Also, if $a$ and $b$ are not $\hbo^{1}$-related then $b$ can move left of $a$ which produces the trace $\tau''_{1} = \alpha\cdot b\cdot a\cdot c\cdot \gamma$ that has the same happens-before as $\tau'_{1}$.

\noindent
\textbf{Induction step:} We assume that the lemma holds for $\length{\beta} \leq n$.
Consider $\tau_{n+1}=\alpha\cdot a\cdot \beta \cdot b\cdot \gamma$ with $\length{\beta} = n + 1$.
Consider $c$ the last event in the sequence $\beta = \beta_1\cdot c$.
If $a$ does not happen before $b$ through $\beta$, then either $a$ does not happen before $c$ through $\beta_1$ and $a$ and $c$ are not $\hbo^{1}$-related, or $c$ and $b$ are not $\hbo^{1}$-related.

First case: suppose that $a$ does not happen before $c$ through $\beta_1$ and $a$ and $c$ are not $\hbo^{1}$-related.
Using the induction hypothesis over $\tau_{n+1}$ with respect to $a$ and $c$ (since $\length{\beta_1} \leq n$) results in $\tau'_{n+1} = \alpha\cdot\beta_{11}\cdot c\cdot a\cdot\beta_{12}\cdot b\cdot\gamma$ that has the same happens-before as $\tau_{n+1}$.
We know that if $a$ happens-before $b$ through $\beta_{12}$ then $a$ happens-before $b$ through $\beta$ because $\beta_{12}$ is a subset of $\beta$. Therefore, $a$ does not happen-before $b$ through $\beta_{12}$. Since $\length{\beta_{12}} \leq \length{\beta_1} \leq n$, then we can apply the induction hypothesis to $\tau'_{n+1}$ with respect to $a$ and $b$ which yields either $\tau''_{n+1} = \alpha\cdot\beta_{11}\cdot c\cdot\beta_{121}\cdot b\cdot a\cdot\beta_{122}\cdot\gamma$ which has the same happens-before as $\tau'_{n+1}$, if $a$ and $b$ are not $\hbo^{1}$-related, or $\tau''_{n+1} = \alpha\cdot\beta_{11}\cdot c\cdot\beta_{121}\cdot a\cdot b\cdot\beta_{122}\cdot\gamma$ which has the same happens-before as $\tau'_{n+1}$, otherwise.

Second case: suppose $c$ and $b$ are not $\hbo^{1}$-related. We apply the induction hypothesis to $\tau_{n+1}$ with respect to $c$ and $b$, and we get $\tau'_{n+1} = \alpha\cdot a\cdot\beta_{1}\cdot b\cdot c\cdot\gamma$ with the same happens-before as $\tau_{n+1}$. As we already know that $a$ does not happen before $b$ through $\beta$ then $a$ does not happen before $b$ through  $\beta_{1}$.
Subsequently by using the induction hypothesis over $\tau'_{n+1}$ with respect to $a$ and $b$, we obtain $\tau''_{n+1} =\alpha\cdot\beta_{11}\cdot b\cdot a\cdot\beta_{12}\cdot c\cdot\gamma$ where $\tau''_{n+1}$ has the same happens-before as $\tau_{n+1}'$, if $a$ and $b$ are not $\hbo^{1}$-related, or $\tau''_{n+1} =\alpha\cdot\beta_{11}\cdot a\cdot b\cdot\beta_{12}\cdot c\cdot\gamma$ where $\tau''_{n+1}$ has the same happens-before as $\tau_{n+1}'$, otherwise.
\end{proof}

We show next that a robustness violation should contain at least an issue and a store event of the same transaction
that are separated by another event that occurs after the issue and before the store and which is related to both via the happens-before relation. Otherwise, since any two events which are not related by happens-before could be swapped in order to derive a trace with the same happens-before, every store event could be swapped until it immediately follows the corresponding issue and the trace would be serializable.

\begin{lem}\label{lemma:OneDelayedTr}
Given a violation $\tau$, there must exist a transaction $\atr$ such that $\tau = \alpha \cdot \issueact(\apr,\atr) \cdot \beta \cdot \storeact(\apr_0,\atr) \cdot \gamma$ and $\issueact(\apr,\atr)$ happens-before $\storeact(\apr_0,\atr)$ through $\beta$.
\end{lem}

\begin{proof}
Assume by contradiction that the lemma does not hold. 
For every transaction $\atr$ of $\tau$ suppose there exist $\apr'\in\mathbb{P}$ such that $\storeact(\apr',\atr)$ does not occur immediately after $\issueact(\apr,\atr)$. Thus, $\tau = \alpha \cdot \issueact(\apr, \atr)\cdot \beta\cdot \storeact(\apr', \atr) \cdot\gamma$, and $(\issueact(\apr, \atr),\storeact(\apr', \atr) )\in \sametro \subset \hbo^{1}$. From Lemma \ref{lemma:Reordering}, $\tau' = \alpha \cdot \beta_1 \cdot\issueact(\apr, \atr)\cdot \storeact(\apr', \atr) \cdot\beta_2\cdot\gamma$ has the same happens-before as $\tau$ (since $\issueact(\apr, \atr)$ does not happens-before $\storeact(\apr', \atr)$ through $\beta$). 
Then, the trace $\tau^{*}$ where for every transaction $\atr$ of $\tau$ the store events occur immediately after the issue event has the same happens-before as $\tau$. Thus, $\tau^{*}$ is serializable which means that its $\hbo_t$ is acyclic which contradicts the fact  that $\tau$ is a violation. 
%In $\tau'$, $\storeact(\apr',\atr)$ occurs immediately after $\issueact(\apr,\atr)$ which contradicts the hypothesis that $\storeact(\apr',\atr)$ does not occur immediately after $\issueact(\apr,\atr)$. Therefore, $\issueact(\apr, \atr)$ happens before $\storeact(\apr', \atr)$ through $\beta$.
\end{proof}

The transaction $\atr$ in the trace $\tau$ above is called a \emph{delayed} transaction. The happens-before constraints imply that $\atr$ belongs to a transactional happens-before cycle in the trace. In the remainder of the paper, when given a violation $\tau = \alpha \cdot \issueact(\apr,\atr) \cdot \beta \cdot \storeact(\apr_0,\atr) \cdot \gamma$, we assume that $\atr$ is the \emph{first} delayed transaction in $\tau$. 

\subsection{Minimal Violations}

Given a trace $\tau = \alpha \cdot b \cdot \beta \cdot c \cdot \omega$ containing two events $b=\issueact(\apr,\atr)$ and $c$, the \emph{distance} between $b$ and $c$, denoted by $d_{\tau}(b,c)$,
is the number of events in $\beta$ that are causally related to $b$, excluding events that correspond to the delivery of $\atr$, i.e., $d_{\tau}(b,c) = \length{\{d \in \beta\ |\ (b,d) \in \viso\ \land d \neq \storeact(\apr',\atr)\mbox{ for every $\apr'\in \mathbb{P}$}\}}$

The \emph{number of delays} $\#(\tau)$ in a trace $\tau$ is the sum of all distances between issue and store events that originate from the same transaction:
$$\#(\tau)= \sum_{\issueact(\apr,\atr),\ \storeact(\apr',\atr)\ \in\ \tau}\ d_{\tau}(\issueact(\apr,\atr),\storeact(\apr',\atr))$$

\begin{defi}[Minimal Violation]
A robustness violation $\tau$ is called \emph{minimal} if it has the least number of delays among all robustness violations (for a given program $\aprog$ and semantics $\textsf{X} \in \{\wcct{},\ \ccvt{},\ \scct{}\}$).
\end{defi}

\begin{rem}
It is important to note that a non-robust program can admit multiple minimal violations with different happens-before relations.
For instance, Fig.~\ref{fig:equivtraces0} pictures two minimal violations that do not have the same happens-before and both traces have 0 delays. In the trace in Fig.~\ref{fig:equivtraces0b} a single transaction is delayed while in the trace in Fig.~\ref{fig:equivtraces0c} two transactions are delayed and are not causally related. For the trace $\tau_1$ in Fig.~\ref{fig:equivtraces0b}, we have that $\#(\tau_1)= d_{\tau_1}(\issueact(\apr 2,\atr 2),\storeact(\apr 2,\atr 2)) + d_{\tau_1}(\issueact(\apr 2,\atr 2),\storeact(\apr 3,\atr 2))=0$. For the trace $\tau_2$ in Fig.~\ref{fig:equivtraces0c}, we have that $\#(\tau_2)= d_{\tau_2}(\issueact(\apr 1,\atr 1),\storeact(\apr 1,\atr 1))+d_{\tau_2}(\issueact(\apr 1,\atr 1),\storeact(\apr 3,\atr 1))+d_{\tau_2}(\issueact(\apr 2,\atr 2),\storeact(\apr 3,\atr 2))=0$. Hence, the number of delays for both cases is 0.  
\end{rem}

\begin{figure}[!ht]
\centering
\begin{minipage}[c]{0.45\textwidth}
\begin{subfigure}{\linewidth}
\lstset{basicstyle=\ttfamily\scriptsize}
\begin{minipage}[c]{17mm}
\begin{lstlisting}[language=Java10]
    p1:
t1: [x = 1
     r1 = y]
\end{lstlisting}
\end{minipage}
\begin{minipage}[c]{2mm}
\footnotesize{$||$}
\end{minipage}
\begin{minipage}[c]{17mm}
\begin{lstlisting}[language=Java10]
    p2:
t2: [y = 2
     r2 = z]
\end{lstlisting}
\end{minipage}
\begin{minipage}[c]{2mm}
\footnotesize{$||$}
\end{minipage}
\begin{minipage}[c]{15mm}
\begin{lstlisting}[language=Java10]
    p3:
t3: [z = 3
     r3 = x
     r4 = y]
\end{lstlisting}
\end{minipage}
\caption{A program.}
\label{fig:equivtraces0a}
\end{subfigure}
\end{minipage}
\hfill
\begin{minipage}[c]{0.39\textwidth}
\begin{subfigure}{\linewidth}
\scalebox{0.59}
{
\begin{tikzpicture}

   \node[shape=rectangle ,draw=none,font=\large] (A) at (0,0)  [] {$\issueact(\apr 2,\atr 2)$};
   \node[shape=rectangle ,draw=none,font=\large] (A1) at (1.9,0)  [] {$\storeact(\apr 2,\atr 2)$};
  \node[shape=rectangle ,draw=none,font=\large] (B) at (4.1,0)  [] {$(\apr 3,\atr 3)$};
  \node[shape=rectangle ,draw=none,font=\large] (C) at (6.3,0)  [] {$\storeact(\apr 3,\atr 2)$};
  \node[shape=rectangle ,draw=none,font=\large] (D) at (8.3,0)  [] {$(\apr 1,\atr 1)$};

  \begin{scope}[ every edge/.style={draw=black,very thick}]
  \path [->,] (A) edge [bend right] node [above,font=\large] {$RW$} (B);
  \path [->] (B) edge [bend left] node [above,font=\large] {$RW$} (C);
  \path [->] (B) edge  [bend right] node [above,font=\large] {$RW$} (D);
  \path [->] (A) edge [bend left] node [above,font=\large] {$STO$} (C);
  \end{scope}

\end{tikzpicture}}
\caption{A minimal violation of $(a)$.}
\label{fig:equivtraces0b}
\end{subfigure}
\end{minipage}

\hfill\begin{minipage}[c]{.85\textwidth}
  \begin{subfigure}{\linewidth}
  \scalebox{0.59}
  {
  \begin{tikzpicture}
  
    \node[shape=rectangle ,draw=none,font=\large] (A) at (0,0)  [] {$\issueact(\apr 1,\atr 1)$};
    \node[shape=rectangle ,draw=none,font=\large] (A1)at (1.9,0)  [] {$\storeact(\apr 1,\atr 1)$};
    \node[shape=rectangle ,draw=none,font=\large] (B) at (4.1,0)  [] {$\issueact(\apr 2,\atr 2)$};
    \node[shape=rectangle ,draw=none,font=\large] (B1)at (6,0)  [] {$\storeact(\apr 2,\atr 2)$};
    \node[shape=rectangle ,draw=none,font=\large] (C) at (8.2,0)  [] {$\storeact(\apr 1,\atr 2)$};
    \node[shape=rectangle ,draw=none,font=\large] (D) at (10.4,0)  [] {$(\apr 3,\atr 3)$};
    \node[shape=rectangle ,draw=none,font=\large] (G) at (12.6,0)  [] {$\storeact(\apr 3,\atr 2)$};
    \node[shape=rectangle ,draw=none,font=\large] (H) at (14.8,0)  [] {$\storeact(\apr 3,\atr 1)$};
  
    \begin{scope}[ every edge/.style={draw=black,very thick}]
    \path [->] (A) edge [bend right] node [below,font=\large] {$STO$} (A1);
    \path [->] (A) edge [bend left] node [above,font=\large] {$RW$} (C);
    \path [->] (B) edge [bend right] node [above,font=\large] {$STO$} (C);
    \path [->] (B) edge [bend left] node [above,font=\large] {$RW$} (D);
    \path [->] (D) edge  [bend left] node [above,font=\large] {$RW$} (G);
    \path [->] (D) edge  [bend right] node [above,font=\large] {$RW$} (H);
    \end{scope}
  
  \end{tikzpicture}}
  \caption{Another minimal violation of $(a)$.}
  \label{fig:equivtraces0c}
  \end{subfigure}
  \end{minipage}\hfill
  \caption{Example of two minimal violation traces that do not have the same happens-before relation (possible under both \ccvt{} and \scct{}). Both traces have the same number of delays which is equal to $0$. The minimal violation in $(b)$ contains a single delayed transaction ($\atr 2$), and the minimal violation in $(c)$ contains two delayed transactions ($\atr 1$ and $\atr 2$). For readability, we do not show all $\po$ and $\sametro$ transitions.}
  \label{fig:equivtraces0}
  \end{figure}

Given a minimal violation $\tau = \alpha \cdot \issueact(\apr,\atr) \cdot \beta \cdot \storeact(\apr_0,\atr) \cdot \gamma$, the following lemma shows that we can assume w.l.o.g. that $\gamma$ contains only store events from transactions that were issued before $\storeact(\apr_0,\atr)$ in $\tau$. 
  
\begin{lem}\label{lemma:StoresSuffix}
  Let $\tau = \alpha \cdot \issueact(\apr,\atr) \cdot \beta \cdot \storeact(\apr_0,\atr) \cdot \gamma$ be a minimal violation such that $\issueact(\apr,\atr)$ happens-before $\storeact(\apr_0,\atr)$ through $\beta$. Then, $\tau' = \alpha \cdot \issueact(\apr,\atr) \cdot \beta \cdot \storeact(\apr_0,\atr) \cdot \gamma'$, such that $\gamma'$ contains only store events from transactions that were issued before $\storeact(\apr_0,\atr)$ in $\tau$, is also a minimal violation.
\end{lem}
\begin{proof}
The prefix $\alpha \cdot \issueact(\apr,\atr) \cdot \beta \cdot \storeact(\apr_0,\atr)$ has a cyclic transactional happens-before and it is already a minimal violation independently of whether $\gamma$ contains additional transactions. 
\end{proof}

%  Next, we show that for a minimal violation of the shape $\tau = \alpha \cdot \issueact(\apr, \atr)\cdot \beta\cdot \storeact(\apr_0, \atr) \cdot\gamma$ such that   $\issueact(\apr, \atr)$ happens-before $\storeact(\apr_0, \atr)$ through $\beta$, we must have an event $\issueact(\apr_0,\atr_0) \in \beta$ such that $\issueact(\apr_0,\atr_0)$ happens before $\storeact(\apr_0, \atr)$. 
  
The following result shows that for every minimal violation, we can extract another minimal violation of the shape $\tau = \alpha \cdot \issueact(\apr,\atr) \cdot \beta \cdot (\apr',\atr') \cdot \storeact(\apr',\atr) \cdot \gamma$ such that $(\issueact(\apr,\atr),(\apr',\atr'))\in \hbo$, and $((\apr',\atr'),\storeact(\apr', \atr))\in \hbo^{1}$.
  
\begin{lem}\label{lemma:LastMacroEvent}
If $\aprog$ is a program that is not robust against some $\textsf{X} \in \{\ccvt{},\ \scct{},\ \wcct{}\}$, then its set of traces under the semantics $\textsf{X}$ must admit a minimal violation of the shape $\tau = \alpha \cdot \issueact(\apr,\atr) \cdot \beta \cdot (\apr',\atr') \cdot \storeact(\apr',\atr) \cdot \gamma$ such that $(\issueact(\apr,\atr),(\apr',\atr'))\in \hbo$ and $((\apr',\atr'),\storeact(\apr', \atr))\in \hbo^{1}$.
\end{lem}
  
\begin{proof} 
Let $\tau = \alpha \cdot \issueact(\apr, \atr)\cdot \beta\cdot \storeact(\apr_0, \atr) \cdot\gamma$ be a minimal violation of $\aprog$, such that  $\issueact(\apr, \atr)$ happens-before $\storeact(\apr_0,\atr)$ through $\beta$. By Lemma~\ref{lemma:StoresSuffix}, we assume that $\gamma$ contains only store events. 
We prove by induction on the size of $\beta$ that $\aprog$ admits another minimal violation against $\textsf{X}$ of the form $\tau' = \alpha' \cdot \issueact(\apr_1,\atr_1) \cdot \beta' \cdot (\apr_2,\atr_2) \cdot \storeact(\apr_2,\atr_1) \cdot \gamma'$ such that $(\issueact(\apr_1,\atr_1),(\apr_2,\atr_2))\in \hbo$, $((\apr_2,\atr_2),\storeact(\apr_2, \atr_1))\in \hbo^{1}$, and $\tau'$ is a permutation of a subsequence of $\tau$. 

Note that $\issueact(\apr,\atr)$ happens-before $\storeact(\apr_0,\atr)$ through $\beta$ implies that there exists a sub-sequence $c_1 \cdots c_n$ of $\beta$ that satisfies: $c_i\rightarrow_{\hbo^{1}} c_{i+1}\quad\text{ for all }i\in[0, n]$ 
where $c_0= \issueact(\apr,\atr)$, $c_{n+1}= \storeact(\apr_0,\atr)$. Then, we have three possibilities for $c_{n}$: $(\apr',\atr')$, $\issueact(\apr',\atr')$, or $\storeact(\apr_0,\atr')$. 

\noindent
\textbf{Base case:} 
$\length{\beta} =1$ implies that $\beta = c_{n}$.  If $c_{n} = (\apr',\atr')$ then $\tau$ is a minimal violation s.t. $\issueact(\apr,\atr)\rightarrow_{\hbo} (\apr',\atr')$ and $(\apr',\atr')\rightarrow_{\hbo^{1}} \storeact(\apr_0,\atr)$.  If $c_{n} = \issueact(\apr',\atr')$ then we regroup together the issue event $\issueact(\apr',\atr')$ with its store events obtaining $\tau' = \alpha \cdot \issueact(\apr,\atr) \cdot  (\apr',\atr') \cdot \storeact(\apr_0,\atr) \cdot \gamma'$ to be a minimal violation as well (since the transactional happens-before of the trace resulting from reordering store events in $\storeact(\apr_0,\atr) \cdot \gamma'$ will always be cyclic). Since $(\apr',\atr')\rightarrow_{\hbo^{1}} \storeact(\apr_0,\atr)$ implies that $(\apr',\atr')\rightarrow_{\hbo^{1}} \storeact(\apr',\atr)\in \gamma$, then $\tau'' = \alpha \cdot \issueact(\apr,\atr) \cdot  (\apr',\atr') \cdot \storeact(\apr',\atr) \cdot \gamma''$, where the two store events $\storeact(\apr',\atr)$ and $\storeact(\apr',\atr)$ are reordered, is a minimal violation. $c_{n} = \storeact(\apr_0,\atr')$ is not possible since $\atr$ is the first delayed transaction in $\tau$. 

\noindent
\textbf{Induction step:} We assume that the induction hypothesis holds for $\length{\beta} \leq m$. 
The case $c_{n} = (\apr',\atr')$ is trivial. If $c_{n} = \issueact(\apr',\atr')$ then removing the issue events that occur after $c_n$ will not impact the happens-before. Thus, we remove every issue and atomic marco event that occurs after $\issueact(\apr',\atr')$ with all their store events and regroup together the event $\issueact(\apr',\atr')$ with its store events obtaining $\tau' = \alpha \cdot \issueact(\apr,\atr) \cdot \beta' \cdot  (\apr',\atr') \cdot \storeact(\apr_0,\atr) \cdot \gamma'$ to be a minimal violation. Similar to before, $\tau'' = \alpha \cdot \issueact(\apr,\atr) \cdot \beta'\cdot  (\apr',\atr') \cdot \storeact(\apr',\atr) \cdot \gamma''$ is a minimal violation.

If $c_{n} = \storeact(\apr_0,\atr')$, then the corresponding issue event $\issueact(\apr',\atr')$ must occur in $\beta$ ($\alpha$ contains only atomic macro events because $\atr$ is the first delayed transaction). If $\issueact(\apr',\atr')$ does not happen before $\storeact(\apr_0,\atr')$ (or any store event of $\atr'$ in $\beta \cdot \storeact(\apr_0,\atr) \cdot \gamma$) through a subsequence of $\beta$ (resp., $\beta \cdot \storeact(\apr_0,\atr) \cdot \gamma$) then we can regroup together the issue and store events of $\atr'$  and get that $\tau' = \alpha \cdot \issueact(\apr,\atr) \cdot \beta' \cdot  (\apr',\atr') \cdot \beta'' \cdot \storeact(\apr_0,\atr) \cdot \gamma'$ is a minimal violation.  Otherwise, if $\issueact(\apr',\atr')$ happens-before $\storeact(\apr_0,\atr')$ through a subsequence of $\beta$, then $\tau$ can be written as $\tau = \alpha \cdot \issueact(\apr,\atr) \cdot \beta_1\cdot \issueact(\apr',\atr')\cdot \beta_2 \cdot \storeact(\apr_0,\atr') \cdot \beta_3\cdot \storeact(\apr_0,\atr) \cdot \gamma$.  
Note that if there exists an issue event $\issueact(\apr_1,\atr_1)$ in $\beta_1\cdot \issueact(\apr',\atr')\cdot \beta_2$ s.t. $(\issueact(\apr_1,\atr_1), \storeact(\apr_0,\atr)) \in \cfo$ (or $(\issueact(\apr 1,\atr 1), \storeact(\apr_1,\atr)) \in \cfo$) then similar to before the following trace $\tau' = \alpha \cdot \issueact(\apr,\atr) \cdot \beta' \cdot  (\apr_1,\atr_1) \cdot \storeact(\apr_0,\atr) \cdot \gamma'$ (resp., $\tau' = \alpha \cdot \issueact(\apr,\atr) \cdot \beta' \cdot  (\apr_1,\atr_1) \cdot \storeact(\apr_1,\atr) \cdot \gamma'$) is a minimal violation. 
Assume now that there does not exist an issue event $\issueact(\apr_1,\atr_1)$. 
Then, let $\issueact(\apr_2,\atr_2)$ be the first issue event in $\issueact(\apr,\atr)\cdot\beta_1\cdot \issueact(\apr',\atr')$ s.t. $\tau = \alpha \cdot \issueact(\apr,\atr) \cdot \beta'_1\cdot \issueact(\apr_2,\atr_2)\cdot \beta'_2 \cdot \storeact(\apr_3,\atr_2) \cdot \beta'_3\cdot\gamma$ and $\issueact(\apr_2,\atr_2)$ happens-before $\storeact(\apr_3,\atr_2)$ through $\beta'_2$ and s.t. for every issue event in $\issueact(\apr,\atr)\cdot\beta'_1$ of a transaction $\atr_4$ there does not exist an event in $\beta'_1\cdot \issueact(\apr_2,\atr_2)\cdot \beta'_2$ that reads from a variable that $\atr_4$ overwrites. 
We can remove every issue event and atomic marco event which occur after $\storeact(\apr_3,\atr_2)$ with all related stores: 
$\tau' = \alpha \cdot \issueact(\apr,\atr) \cdot \beta'_1\cdot \issueact(\apr_2,\atr_2)\cdot \beta'_2 \cdot \storeact(\apr_3,\atr_2) \cdot \gamma'$ where $\gamma'$ contains only store events is a minimal violation. Then, not delaying the transactions in $\issueact(\apr,\atr)\cdot\beta'_1$ does not affect the reads in $\beta'_1\cdot \issueact(\apr_2,\atr_2)\cdot \beta'_2$, and thus, we get that $\tau'' = \alpha \cdot (\apr,\atr) \cdot \beta''_1\cdot \issueact(\apr_2,\atr_2)\cdot \beta''_2 \cdot \storeact(\apr_3,\atr_2) \cdot \gamma''$, where $\atr_2$ is the first delayed transaction in $\tau''$ and $\issueact(\apr_2,\atr_2)$ happens-before $\storeact(\apr_3,\atr_2)$ through $\beta''_2$, is a minimal violation. Note that $\length{\beta''_2} < \length{\beta} = m +1 $, and we can apply the induction hypothesis to $\tau''$ and conclude the proof. 
\end{proof}

Next, we show that a program which is not robust against \ccvt{} or \scct{} admits violations of particular shapes. 
For the remainder of the paper, we write a minimal violation in the shape $\tau = \atmacro \cdot \issueact(\apr,\atr) \cdot \beta \cdot (\apr',\atr') \cdot \storeact(\apr',\atr) \cdot \stmacro$ to say that all the events in the sequence $\atmacro$ are atomic macro events and all the events in the sequence $\stmacro$ are store events. As before, we assume that $\atr$ is the first delayed transaction in $\tau$, and by Lemma~\ref{lemma:LastMacroEvent}, we assume that $(\issueact(\apr,\atr),(\apr',\atr'))\in \hbo$ and $((\apr',\atr'),\storeact(\apr', \atr))\in \hbo^{1}$.

%, which enables reducing robustness checking to a reachability problem in a program running under serializability (presented in Section~\ref{sec:Instr}).

%!TEX root = draft.tex

\begin{figure}[t]
  \begin{minipage}[c]{\textwidth}
  \begin{subfigure}{\linewidth}
  \scalebox{0.71}
  {
  \begin{tikzpicture}
  \node[text width=1cm, anchor=west, left,font=\LARGE] at (0,0) {$\tau_{\ccvt{}1:}$};
   \node[shape=rectangle ,draw=none,font=\huge] (A0) at (0.5,0)  [] { };
   \node[shape=rectangle ,draw=none,font=\Large] (A1) at (1.5,0.3)  [] {$\mathbf{\atmacro}$};
   \node[] (A) at (1.5,0)  [] {};
   \node[shape=rectangle ,draw=none,font=\Large] (B1) at (2.8,0.5)  [] {$\issueact(\apr,\atr)$};
   \node[shape=rectangle ,draw=none,font=\huge] (B) at (2.8,0)  [] {$\circ$};
   \node[shape=rectangle ,draw=none,font=\Large] (B2) at (5,0.5)  [] {$\storeact(\apr,\atr)$};
   \node[shape=rectangle ,draw=none,font=\huge] (B3) at (5,0)  [] {$\circ$};
   \node[shape=rectangle ,draw=none,font=\Large] (C1) at (8,0.3)  [] {$\mathbf{\beta}$};
   \node[] (C) at (8,0)  [] {};
    \node[shape=rectangle ,draw=none,font=\Large] (D1) at (10.5,0.5)  [] {$(\apr',\atr')$};
    \node[shape=rectangle ,draw=none,font=\huge] (D) at (10.5,0)  [] {$\circ$};
    \node[shape=rectangle ,draw=none,font=\Large] (E1) at (13,0.5)  [] {$\storeact(\apr',\atr)$};
    \node[shape=rectangle ,draw=none,font=\huge] (E) at (13,0)  [] {$\circ$};
    \node[shape=rectangle ,draw=none,font=\Large] (F1) at (14.5,0.3)  [] {$\mathbf{\stmacro}$};
     \node[] (F) at (14.5,0)  [] {};
    \node[shape=rectangle ,draw=none,font=\huge] (F0) at (15.5,0)  [] {};
  
    \begin{scope}[ every edge/.style={draw=black}]
    \path [|-|] (A0) edge []  (B);
    \path [|-|] (B3) edge [] (D);
    \path [->] (B1) edge [bend left=20] node [above,font=\small] {$\hbo\setminus \viso$} node [pos=0.95,above,font=\small] {$\forall$} (C1);
    \path [->] (C1) edge [bend left=20] node [above,font=\small] {$\viso$} node [pos=0.05,above,font=\small] {$\forall$} (D1);
    \path [->] (B3)  edge [bend right=20] node  [below,font=\small] {$\sto(y)$} (D);
    \path [->] (D) edge [bend right=20] node [below,font=\small] {$\cfo(y)$} (E);
    \path [|-|] (E) edge [] (F0);
    \end{scope}
  \end{tikzpicture}}
  \label{fig:rob0traceccv}
  \end{subfigure}
  \end{minipage}
  \begin{minipage}[c]{\textwidth}
  \begin{subfigure}{\linewidth}
  \scalebox{0.71}
  {\begin{tikzpicture}
  \node[text width=1cm, anchor=west, left,font=\LARGE] at (0,0) {$\tau_{\ccvt{}2:}$};
   \node[shape=rectangle ,draw=none,font=\large] (A0) at (0.5,0)  [] {};
   \node[shape=rectangle ,draw=none,font=\Large] (A1) at (1.5,0.3)  [] {$\atmacro$};
   \node[] (A) at (1.5,0) [] {};
   \node[shape=rectangle ,draw=none,font=\Large] (B1) at (3,0.5)  [] {$\issueact(\apr,\atr)|_{\neg x}$};
   \node[shape=rectangle ,draw=none,font=\huge] (B) at (3,0)  [] {$\circ$};
   \node[shape=rectangle ,draw=none,font=\Large] (C1) at (5,0.3)  [] {$\beta_1|_{\neg x}$};
   \node[] (C) at (5,0) [] {};
   \node[shape=rectangle ,draw=none,font=\Large] (D1) at (7,0.5)  [] {$\issueact(\apr_{1},\atr_{1})$};
   \node[shape=rectangle ,draw=none,font=\huge] (D) at (7,0)  [] {$\circ$};
   \node[shape=rectangle ,draw=none,font=\Large] (E1) at (11,0.3)  [] {$\beta_2$};
   \node[] (E) at (11,0) [] {};
   \node[shape=rectangle ,draw=none,font=\Large] (H1) at (14,0.5)  [] {$(\apr',\atr')$};
   \node[shape=rectangle ,draw=none,font=\huge] (H) at (14,0)  [] {$\circ$};
   \node[shape=rectangle ,draw=none,font=\Large] (F1) at (17,0.5)  [] {$\storeact(\apr',\atr)$};
   \node[shape=rectangle ,draw=none,font=\huge] (F) at (17,0)  [] {$\circ$};
    \node[shape=rectangle ,draw=none,font=\Large] (G1) at (18.5,0.3)  [] {$\stmacro$};
    \node[] (G) at (18.5,0) [] {};
    \node[shape=rectangle ,draw=none,font=\large] (G0) at (19.5,0)  [] {};
  
    \begin{scope}[ every edge/.style={draw=black}]
    \path [|-|] (A0) edge []  (B);
    \path [|-|] (B) edge [] (D);
    \path [->] (D1) edge [bend left=20] node [above,font=\small] {$\hbo\setminus \viso$} node [pos=0.95,above,font=\small] {$\forall$} (E1);
    \path [|-|] (D) edge [] (H);
    \path [->] (B) edge [bend right=20] node [below,font=\small] {$\viso$} (D);
    \path [->] (D) edge [bend right=25] node [below,font=\small] {$\cfo(x)\cup(\sametro;\sto(x))$} node [pos=0.95,below,font=\small] {$\exists$} (E);
    \path [->] (E1) edge [bend left=20] node [above,font=\small] {$\hbo$} node [pos=0.05,above,font=\small] {$\forall$} (H1);
    \path [->] (H) edge [bend right=20] node [below,font=\small] {$\cfo(y\neq x)$}  (F);
    \path [|-|] (F) edge []  (G0);
    \end{scope}
  \end{tikzpicture}}
  \label{fig:rob1traceccv}
  \end{subfigure}
  \end{minipage}
  \vspace{-0.3cm}
  \caption{Robustness violation patterns under \ccvt{}. We use $a \xrightarrow{R \quad}\mathrel{^{\forall}} \beta$
  to denote $\forall\ b \in \beta.\ (a,b) \in R$. We use $\beta_1|_{\neg x}$ to say that all delayed transactions in $\beta_1$ do not access $x$. For violation $\tau_{\ccvt{}1}$, $\atr$ is the only delayed transaction.
  For $\tau_{\ccvt{}2}$, all delayed transactions are in $\issueact(\apr,\atr)\cdot\beta_1\cdot\issueact(\apr_{1},\atr_{1})$ and they form a causality chain that starts at $\issueact(\apr,\atr)$ and ends at $\issueact(\apr_{1},\atr_{1})$.}
  \label{fig:ccvrobtraces}
  \end{figure}

  \begin{figure}[t]
  \begin{minipage}[c]{0.4\textwidth}
  \begin{subfigure}{\linewidth}
  \scalebox{0.61}
  {
  \begin{tikzpicture}
  
   \node[shape=rectangle ,draw=none,font=\large] (A) at (0,0)  [] {$\issueact(\apr 1,\atr 1)$};
   \node[shape=rectangle ,draw=none,font=\large] (A1) at (1.9,0)  [] {$\storeact(\apr 1,\atr 1)$};
    \node[shape=rectangle ,draw=none,font=\large] (B) at (4.5,0)  [] {$(\apr 2,\atr 2)$};
    \node[shape=rectangle ,draw=none,font=\large] (C) at (6.2,0)  [] {$\storeact(\apr 2,\atr 1)$};
  
    \begin{scope}[ every edge/.style={draw=red,very thick}]
    \path [->] (A1) edge [bend right] node [above,font=\large] {$\sto$} (B);
    %\path [->] (A) edge node [pos=0.65,above,font=\large] {$\sametro$} (A1);
    \path [->] (B) edge [bend right] node [above,font=\large] {$\cfo$} (A1);
    \end{scope}
  
  \end{tikzpicture}}
  \caption{Violation of LU program in Fig.~\ref{fig:rob0}.}
  \label{fig:rob0LUtraceccv}
  \end{subfigure}
  \end{minipage}
  \hspace{1cm}
  \begin{minipage}[c]{0.4\textwidth}
  \begin{subfigure}{\linewidth}
  \scalebox{0.61}
  {\begin{tikzpicture}
  
   \node[shape=rectangle ,draw=none,font=\large] (A) at (0,0)  [] {$\issueact(\apr 1,\atr 1)$};
   \node[shape=rectangle ,draw=none,font=\large] (A1) at (1.9,0)  [] {$\storeact(\apr 1,\atr 1)$};
    \node[shape=rectangle ,draw=none,font=\large] (B) at (4,0)  [] {$(\apr 2,\atr 2)$};
    \node[shape=rectangle ,draw=none,font=\large] (C) at (6.2,0)  [] {$\storeact(\apr 2,\atr 1)$};
  
    \begin{scope}[ every edge/.style={draw=red,very thick}]
    \path [->] (A) edge  [bend left] node [above,font=\large] {$\cfo$} (B);
    \path [->] (B) edge [bend left] node [above,font=\large] {$\cfo$} (C);
    \end{scope}
  \end{tikzpicture}}
  \caption{Violation of SB program in Fig.~\ref{fig:rob1}.}
  \label{fig:rob1SBtraceccv}
  \end{subfigure}
  \end{minipage}
  \caption{(a) A $\tau_{\ccvt{}1}$ violation where $\beta_2=\epsilon$, $\stmacro=\epsilon$, and $\atr$ and $\atr'$ correspond to $\atr 1$ and $\atr 2$. (b) A $\tau_{\ccvt{}2}$ (resp., $\tau_{\scct{}2}$) violation where $\atr$ and $\atr_1$ coincide and correspond to $\atr 1$. Also, $\beta_1=\epsilon$, $\beta_2=(\apr 2,\atr 2)$, $\stmacro=\epsilon$, such that
  $(\issueact(\apr 1,\atr 1),(\apr 2,\atr 2)) \in \cfo(y)$ and $((\apr 2,\atr 2),\storeact(\apr 2,\atr 1)) \in  \cfo(x)$. In all traces, we show only the relations that are part of the happens-before cycle.}
  \label{fig:ccvrobtracesexamples}
  \end{figure}

  \section{Robustness Violations Under Causal Convergence} \label{sec:MVUCCV}

  In this section, we present a precise characterization of minimal violations under \ccvt{}. In particular, we show that in these violations, the first delayed transaction (which must exist by Lemma~\ref{lemma:OneDelayedTr}) is followed by a possibly-empty sequence of delayed transactions that form a ``causality chain'', i.e., the issue of every new delayed transaction is causally ordered after the issue of the first delayed transaction. Also, we show that the issue event of the last delayed transaction happens-before an event of another transaction that reads a variable updated by the first delayed transaction (which implies a cycle in the transactional happens-before). This characterization will allow us to build a monitor for detecting the existence of robustness violations that is linear in the size of the input program. 
  
Next, we give a precise definition of the ``causality chain''. It consists of a sequence of issue events such that the first issue is causally ordered before every other issue event and every issued transaction is delivered to the process executing the next issue event in the chain, before this issue event executes. 
  
  \begin{defi}
  We say that a sequence of issue events $\event_1\cdot\event_2\cdot\ldots\event_n$ forms a \emph{causality chain} that starts with $\event_1$ and ends at $\event_n$ in a trace $\tau$ if the followings hold:
  \begin{enumerate}
    \item $(\event_{1},\event_{i})\in \viso$, for all $2\leq i\leq n$
    \item for all $1\leq i\leq n-1$ such that $\event_{i} = \issueact(\apr_{i},\atr_{i})$, $\event_{i+1} = \issueact(\apr_{i+1},\atr_{i+1})$, the store event $\storeact(\apr_{i+1},\atr_{i})$ occurs before the issue event $\event_{i+1}$ in $\tau$.  
  \end{enumerate}
  \end{defi}
  
The characterization of robustness violations under \ccvt{} is stated in the following theorem and pictured in Fig.~\ref{fig:ccvrobtraces}. 

\begin{thm}\label{theorem:CcvMinViol}
A program $\aprog$ is not robust under \ccvt{} iff there exists a minimal violation in $\tracesconf(\aprog)_{\ccvt{}}$ of one of the following forms:
\begin{enumerate}[topsep=5pt]
\item $\tau_{\ccvt{}1}=\atmacro\cdot \issueact(\apr,\atr) \cdot \storeact(\apr,\atr) \cdot\beta_{2}\cdot (\apr',\atr') \cdot \storeact(\apr',\atr) \cdot \stmacro\mbox{ where:}$
\begin{enumerate}[label=(\alph*),topsep=5pt]
\item $\issueact(\apr,\atr)$ is the issue of the first and only delayed transaction (Lemma~\ref{lemma:CcvMinForm});
\item  $\exists\ y.$ s.t. $(\storeact(\apr,\atr),(\apr',\atr')) \in \sto(y)$ and $((\apr',\atr'),\storeact(\apr',\atr)) \in  \cfo(y)$ (Lemma \ref{lemma:CcvMinForm}); % and \ref{lemma:CCvstreamviol3}
\item   $\forall\ a \in \beta_{2}.\ (\issueact(\apr,\atr),a) \in \hbo\setminus\viso \mbox{ and } (a,(\apr',\atr')) \in \viso$ (Lemma \ref{lemma:CcvMinForm}).
%\item  $\stmacro$ contains only stores of $\atr$ (Lemma \ref{lemma:StoresSuffix}).
\end{enumerate}
\item $\tau_{\ccvt{}2}=\atmacro\cdot \issueact(\apr,\atr)\cdot \beta_1\cdot \issueact(\apr_{1},\atr_1) \cdot\beta_2\cdot (\apr',\atr')\cdot \storeact(\apr',\atr) \cdot \stmacro\mbox{ where:}$
\begin{enumerate}[label=(\alph*),topsep=5pt]
\item $\issueact(\apr,\atr)$ and $\issueact(\apr_{1},\atr_1)$ are the issues of the first and last delayed transactions (Lemmas \ref{lemma:CcvMinForm} and \ref{lemma:CcvMinViol});
\item  the issues of all delayed transactions are in $\beta_1$ and are included in a causality chain that starts with $\issueact(\apr,\atr)$ and ends at $\issueact(\apr_{1},\atr_{1})$ (Lemma \ref{lemma:CcvMinViol}); % and \ref{lemma:CCvstreamviol3}
\item  for every $a\in\beta_2$, we have that $(\issueact(\apr_1,\atr_1),a) \in \hbo\setminus\viso$ and $(a,(\apr',\atr')) \in \hbo$ (Lemma \ref{lemma:CcvMinForm});
\item  there exist $a \in \beta_2\cdot (\apr',\atr')$, $x$, and $y$ s.t. $x \neq y$, $(\issueact(\apr_1,\atr_1),a) \in \cfo(x)\cup(\sametro;\sto(x))$,  $(a,(\apr',\atr'))\in \hbo{}?$\footnote{$\hbo{}?$ is the reflexive closure of $\hbo{}$.}, and $((\apr',\atr'),\storeact(\apr',\atr)) \in \cfo(y)$ (Lemma~\ref{lemma:CcvMinForm});
\item  all delayed transactions in $\issueact(\apr,\atr) \cdot \beta_1$ do not access the variable $\anaddr$ (Lemma \ref{lemma:NoAccessX}).
\end{enumerate}
\end{enumerate}
\end{thm}

Above, $\tau_{\ccvt{}1}$ contains a single delayed transaction while $\tau_{\ccvt{}2}$ may contain arbitrarily many delayed transactions. In $\tau_{\ccvt{}1}$ the store event $\storeact(\apr,\atr)$ of the only delayed transaction happens before $(\apr',\atr')$ which is conflicting with $\atr$, thus resulting in a cycle in the transactional happens-before. In $\tau_{\ccvt{}2}$ the issue event of the last delayed transaction $\atr_1$, which is causally ordered after the issue of the first delayed transaction $\atr$, happens before $(\apr',\atr')$ which is conflicting with $\atr$, thus resulting in a cycle in the transactional happens-before as well. The theorem above allows $\atmacro=\epsilon$, $\beta_1=\epsilon$, $\beta_2=\epsilon$, $\beta=\epsilon$, $\stmacro=\epsilon$, $\apr=\apr_1$, $\atr=\atr_1$, and $\atr_1$ to be a read-only transaction. Fig.~\ref{fig:rob0LUtraceccv} and Fig.~\ref{fig:rob1SBtraceccv} show two violations under \ccvt{} where such equalities hold. If $\atr_1$ is a read-only transaction then $\issueact(\apr_{1},\atr')$ has the same effect as $(\apr_{1},\atr_1)$ since $\atr_1$ does not contain writes.

The minimality of the violation enforces the constraints stated above. For example, in the context of $\tau_{\ccvt{}2}$, the delayed transactions in $\beta_1$ cannot create a cycle in the transactional happens-before (otherwise, there exists a sequence of store events $\stmacro'$ such that $\atmacro \cdot \issueact(\apr,\atr)\cdot \storeact(\apr,\atr) \cdot \beta_1 \cdot \storeact(\apr_0,\atr) \cdot \stmacro'$ is a violation with a smaller measure, which contradicts minimality). Moreover, $(c)$ implies that $\beta_2$ contains no stores of delayed transactions from $\beta_1$. If this were the case, then these stores can either be reordered after $\storeact(\apr',\atr)$ or if this is not possible due to happens-before constraints, then there would exist an issue event which is after such a store in the happens-before order and thus causally after $\issueact(\apr,\atr)$, which would contradict the fact that $\issueact(\apr_1,\atr_1)$ is the last issue event in $\tau$ that is causally ordered after $\issueact(\apr,\atr)$. 
Also, if it were to have a delayed transaction $\atr_2$ in $\beta_2$ (resp., $\beta$ for $\tau_{\ccvt{}1}$), then it is possible to remove some transaction (the issue and all its store events) from the original trace and obtain a new violation trace with a smaller number of delays. For instance, in the case of $\beta_2$, if $\atr_1\neq\atr$, then we can remove the events of the last delayed transaction (i.e., $\atr_1$), that is causally related to $\issueact(\apr,\atr)$, since all events in $\beta_2\cdot\storeact(\apr_0,\atr)\cdot\stmacro$ neither read from the writes of $\atr_1$ nor are issued by the same process as $\atr_1$ (because of the $\hbo\setminus\viso$ relation between events $\beta_2$ and $\issueact(\apr_1,\atr_1)$). The resulting trace is still a robustness violation (because of the transactional happens-before cycle involving $\atr_2$ since it is delayed in $\beta_2$) but with a smaller measure. Note that all processes that delayed transactions, stop executing new transactions in $\beta_2$ (resp., $\beta$) because of the relation $\hbo\setminus \viso$, shown in Fig.~\ref{fig:ccvrobtraces}, between the delayed transaction $\atr_1$ (resp., $\atr$) and events in $\beta_2$ (resp., $\beta$).

In the following we give a series of lemmas that collectively imply Theorem~\ref{theorem:CcvMinViol}. Next lemma gives the decomposition of minimal violations under \ccvt{} into two possible patterns. It also characterizes the nature of the happens-before dependencies in these traces. For instance, we show that the last dependency in the happens-before cycle is always a conflict dependency. The lemma proof starts with a minimal violation as characterized in Lemma~\ref{lemma:LastMacroEvent} and uses induction to show that we can always obtain a minimal violation which follows one of the two patterns. The induction is based on the size of the sequence of events between the issue and delayed store events of the first delayed transaction (the sequence $\beta$ in Lemma~\ref{lemma:LastMacroEvent}). 

%Let $\tau = \atmacro \cdot \issueact(\apr,\atr) \cdot \beta_1\cdot  \issueact(\apr_1,\atr_1) \cdot \beta_2 \cdot \storeact(\apr_{0},\atr) \cdot \stmacro$ be a tace such that $(\issueact(\apr,\atr),\issueact(\apr_1,\atr_1)) \in \viso$ and $(\issueact(\apr_1,\atr_1),\storeact(\apr_{0},\atr)) \in \hbo$. In $\tau$, we say that $\issueact(\apr_1,\atr_1)$ was \emph{overstepped} by $\storeact(\apr_{0},\atr)$. 

%The following lemma characterizes the relation between the last delayed transaction that was overstepped by the first delayed transaction and the store of the first delayed transaction. Also, it shows 

\begin{lem}\label{lemma:CcvMinForm}
If $\aprog$ is a program that is not robust under \ccvt{}, then it must admit a minimal violation $\tau$ that satisfies one of the following:
\begin{enumerate}[topsep=5pt]
\item $\tau = \atmacro \cdot \issueact(\apr,\atr) \cdot \storeact(\apr,\atr) \cdot \beta \cdot (\apr',\atr') \cdot \storeact(\apr',\atr)  \cdot \stmacro \mbox{ where: }$
\begin{enumerate}[label=(\alph*),topsep=5pt]
\item $\exists\ y.$ s.t.  $(\storeact(\apr,\atr),(\apr',\atr')) \in \sto(y)$ and $((\apr',\atr'),\storeact(\apr',\atr)) \in  \cfo(y)$;
\item $\forall\ a \in \beta.\ (\issueact(\apr,\atr),a) \in \hbo\setminus\viso \mbox{ and } (a,(\apr',\atr')) \in \viso$.
\end{enumerate}
%\item $\tau = \atmacro \cdot \issueact(\apr,\atr) \cdot \beta \cdot  \storeact(\apr',\atr)\cdot \stmacro \mbox{ where: }$
%\begin{enumerate}[label=(\alph*)]
%  \item there exist $a$ and $b=(\apr',\atr')$ in $\beta$ s.t. $(\issueact(\apr,\atr),a) \in \cfo(x)\cup(\sametro;\sto(x))$, $(b,\storeact(\apr', \atr))\in \cfo(y)\cup\sto(y)$, and $(a,b)\in \hbo{}?$;
%  \item $\forall\ a \in \beta.\ (\issueact(\apr,\atr),a) \in \hbo\setminus\viso \mbox{ and }  (a,\storeact(\apr',\atr)) \in \hbo$.
%\end{enumerate}
\item $\tau = \atmacro \cdot \issueact(\apr,\atr) \cdot \beta_1 \cdot \issueact(\apr_{1},\atr_1)  \cdot \beta_2 \cdot (\apr',\atr')\cdot \storeact(\apr',\atr)  \cdot \stmacro \mbox{ where: }$
\begin{enumerate}[label=(\alph*),topsep=5pt]
\item $\issueact(\apr_{1},\atr_1)$ is the last issue event in $\{c \in \beta\ |\ (\issueact(\apr,\atr),c) \in \viso\}$;
\item $\exists\ x,\ y,\ \mbox{and }a \in \beta_2\cdot (\apr',\atr')$ s.t. $(\issueact(\apr_1,\atr_1),a) \in \cfo(x)\cup(\sametro;\sto(x))$, $(a,(\apr',\atr'))\in \hbo{}?$, and $((\apr',\atr'),\storeact(\apr', \atr))\in \cfo(y)$;
\item $\forall\ a \in \beta_2.\ (\issueact(\apr_1,\atr_1),a) \in \hbo\setminus\viso \mbox{ and }  (a,(\apr',\atr')) \in \hbo$.
\end{enumerate}
\end{enumerate}
\end{lem}

\begin{proof}
%It is important to note that $((\apr_{0},\atr_0),\issueact(\apr,\atr)) \in  \cfo(x)$ implies $((\apr_{0},\atr_0),\storeact(\apr_{0},\atr)) \in  \cfo(x);\sametro$. 
%We showed in Lemma \ref{lemma:LastMacroEvent} that a minimal violation trace $\tau$ can be rewritten in the shape $\tau = \atmacro \cdot \issueact(\apr,\atr) \cdot \beta' \cdot (\apr',\atr') \cdot \storeact(\apr_0,\atr) \cdot \stmacro$ such that $(\issueact(\apr,\atr),(\apr',\atr'))\in \hbo$ and $((\apr',\atr'),\storeact(\apr_0, \atr))\in \hbo^{1}$ (i.e., $\beta = \beta' \cdot (\apr',\atr')$). 

Let $\tau = \atmacro  \cdot \issueact(\apr,\atr) \cdot \beta  \cdot (\apr',\atr')\cdot \storeact(\apr',\atr)\cdot\stmacro$ be a minimal violation under \ccvt{} (cf. Lemma~\ref{lemma:LastMacroEvent}). We  prove by induction on the size of $\beta$ that there exists a minimal violation trace $\tau'$ that satisfies ($1$) or ($2$) and  $\tau'$ is obtained from $\tau$.  By the definition of the happens-before $((\apr',\atr'),\storeact(\apr', \atr))\in \hbo^{1}$ implies that $((\apr',\atr'), \storeact(\apr', \atr))\in \cfo\cup\sto$. Since $\atr'$ was issued after $\atr$ in $\tau$, then based on the total order of timestamps under \ccvt{}, we cannot have $((\apr',\atr'), \storeact(\apr', \atr))\in \sto$. 
Then, there must exist $y$ s.t. $((\apr',\atr'), \storeact(\apr', \atr))\in \cfo(y)$. 

\noindent
\textbf{Base case:} 
$\length{\beta} =0$. Since $(\issueact(\apr,\atr),(\apr',\atr'))\in \hbo$ then from the definition of the happens-before the only possible relation is $(\issueact(\apr,\atr),(\apr',\atr'))\in \cfo$. Thus, there must exist $x$ s.t. $(\issueact(\apr,\atr),(\apr',\atr')) \in \cfo(x)$. If $x=y$ then both $\atr$ and $\atr'$ write to $x$. Thus, by reordering the store event $\storeact(\apr,\atr)\in \stmacro$ to occur just after the corresponding issue event we get $\tau' = \atmacro  \cdot \issueact(\apr,\atr) \cdot \storeact(\apr,\atr) \cdot (\apr',\atr')\cdot \storeact(\apr',\atr)\cdot\stmacro'$ is also a minimal violation where 
$(\storeact(\apr,\atr),(\apr',\atr'))\in \sto(x)$ (since $\atr$ was issued before $\atr'$ and both write to $x$) and $((\apr',\atr'), \storeact(\apr', \atr))\in \cfo(x)$. $\tau'$ satisfies the first case of the lemma. If $x\neq y$ then $\tau = \atmacro  \cdot \issueact(\apr,\atr) \cdot (\apr',\atr')\cdot \storeact(\apr',\atr)\cdot\stmacro$ where there exist $x$ and $y$ s.t. $x\neq y$, $(\issueact(\apr,\atr),(\apr',\atr')) \in \cfo(x)$, and $((\apr',\atr'), \storeact(\apr', \atr))\in \cfo(y)$ satisfies the second case of the lemma where $\atr$ and $\atr_1$ coincide and $a$ corresponds to $(\apr',\atr')$.

\noindent
\textbf{Induction step:} We assume the induction hypothesis holds for $\length{\beta} \leq m$.  Let $\sigma=\{c \in \beta\ |\ (\issueact(\apr,\atr),c) \in \viso\}$, we will consider the following three possible cases: 

First, assume that $\sigma$ is empty. Since $(\issueact(\apr,\atr),(\apr',\atr'))\in \hbo$ then there must exist $a \in  \beta \cdot (\apr',\atr')$ s.t. $(\issueact(\apr,\atr),a) \in \hbo^{1}$ and $(a, (\apr',\atr')) \in \hbo?$. $\sigma$ is empty implies that $\beta$ does not contain events that are related to $\issueact(\apr,\atr)$ through $\viso$ (which includes $\po \cup \rfo \cup \sametro$), therefore, $(\issueact(\apr,\atr),a) \in \sto\cup\cfo$. It is impossible to have $(\issueact(\apr,\atr),a) \in \sto$ since $\issueact(\apr,\atr)$ does not contain writes. Thus, there must exist $x$ s.t. $(\issueact(\apr,\atr),a) \in \cfo(x)$. If $x=y$ then both the transaction of the event $a$, denoted $\atr_2$, and $\atr$ write to $x$. We consider the two cases of $(a, (\apr',\atr')) \in \hbo?$: i) $a= (\apr',\atr')$ (i.e., $\atr_2 = \atr'$), and ii) $(a, (\apr',\atr')) \in \hbo$. Assume $a=(\apr',\atr')$ then by reordering the store event $\storeact(\apr,\atr)\in \tau$ to occur just after the corresponding issue event (since the events in $\beta$ are not causally related to $\issueact(\apr,\atr)$) we get $\tau' = \atmacro  \cdot \issueact(\apr,\atr) \cdot \storeact(\apr,\atr) \cdot\beta \cdot(\apr',\atr')\cdot \storeact(\apr',\atr)\cdot\stmacro'$ is also a minimal violation where $(\storeact(\apr,\atr),(\apr',\atr'))\in \sto(x)$ (since $\atr$ was issued before $\atr'$ and both write to $x$) and $((\apr',\atr'), \storeact(\apr', \atr))\in \cfo(x)$. In $\tau'$ we remove all events in $\beta$ that are not causally ordered before $(\apr',\atr')$ since they do not contribute to the happens-before cycle. We obtain a new violation trace that satisfies the first case of the lemma. Assume now that $(a, (\apr',\atr')) \in \hbo$. This implies that $\issueact(\apr_2,\atr_2)\in \beta$ happens-before $(\apr',\atr')$ (since $a$ is an event $\atr_2$). 
Since both $\atr_2$ and $\atr$ write to $x$ and $\atr$ occurs before $\atr_2$ in $\tau$ then from the definition of store and conflict relations $((\apr',\atr'), \storeact(\apr', \atr))\in \cfo(x)$ implies that $((\apr',\atr'), \storeact(\apr', \atr_2))\in \cfo(x)$. Also, since in $\beta$ we do not have events that are causally related to $\issueact(\apr,\atr)$ then let $\tau'$ be the trace resulting from removing all events of $\atr$ in $\tau$:  $\tau' = \atmacro  \cdot \issueact(\apr_2,\atr_2) \cdot \beta' \cdot (\apr',\atr')\cdot \storeact(\apr',\atr_2)\cdot\stmacro'$ where $\tau'$ is a subsequence of $\tau$ and $\beta'$ is a subsequence of $\beta$. $\tau'$ is a minimal violation as well since it was obtained from $\tau$ by just removing events and $(\issueact(\apr_2,\atr_2), (\apr',\atr')) \in \hbo$ and $((\apr',\atr'), \storeact(\apr', \atr_2))\in \cfo(x)$. Since $\length{\beta'} \leq m$ then we can apply the induction hypothesis on $\tau'$. If $x \neq y$ we get that in $\tau$, $(\issueact(\apr,\atr),a) \in \cfo(x)$ and $((\apr',\atr'), \storeact(\apr', \atr))\in \cfo(y)$ which satisfies the second case of the lemma. 

Second, assume that $\sigma$ is not empty and all the elements of $\sigma$ are store events. 
Since $\atr$ is the first delayed transaction in $\tau$ then all stores in $\sigma$ are stores of $\atr$. 
Then, following the same analogy as before there must exist $x$ and an event $a \in  \beta \cdot (\apr',\atr')$ that is not a store event of $\atr$ s.t. $(\issueact(\apr,\atr),a) \in (\sametro;\sto(x))\cup\cfo(x)$ and $(a, (\apr',\atr')) \in \hbo?$. 
Similar to before we consider the two cases $x=y$ and $x \neq y$ and apply the induction hypothesis in the first case.

Third, assume that $\sigma$ is not empty and $\issueact(\apr_{1},\atr_1)$ is the last issue event in $\sigma$, i.e., $\beta = \beta_{1}  \cdot \issueact(\apr_{1},\atr_1) \cdot  \beta_{2}$ where all the events in $\beta_{2}$ are either stores of transactions that are causally related to $\issueact(\apr,\atr)$ (we can reorder these stores to be part of $\stmacro$ except the store $\storeact(\apr_{1},\atr_1)$) or other events that are not causally related to $\issueact(\apr,\atr)$. We also have that $\issueact(\apr,\atr)$ is causally ordered before $\issueact(\apr_{1},\atr_1)$. 
Since $(\issueact(\apr,\atr),(\apr',\atr'))\in \hbo$ then $(\issueact(\apr_1,\atr_1),(\apr',\atr'))\in \hbo$, otherwise, we remove  $\issueact(\apr_{1},\atr_1)$ and all related store events from $\tau$ and the resulting trace is still a violation and it has less delays since $\issueact(\apr,\atr)$ was not delayed after $\issueact(\apr_{1},\atr_1)$ in the trace. Thus, $(\issueact(\apr_1,\atr_1),(\apr',\atr'))\in \hbo$.  Similar to before we obtain that there exist $x$ and an event $a \in \beta_{2} \cdot (\apr',\atr')$ that is not a store event of $\atr_1$ s.t. $(\issueact(\apr_1,\atr_1),a) \in (\sametro;\sto(x))\cup\cfo(x)$ and $(a, (\apr',\atr')) \in \hbo?$. If $x=y$ then both the transaction of the event $a$, denoted $\atr_2$, and $\atr$ write to $x$. Thus, $(\issueact(\apr,\atr),a) \in \sametro;\sto(x)$. Then, since the events in $\beta_{2} \cdot (\apr',\atr')$ do not causally depend on $\issueact(\apr_1,\atr_1)$ then we can remove the events of $\atr_1$ and obtain $\tau'$ where $(\issueact(\apr,\atr),a) \in \sametro;\sto(x)$, $(a, (\apr',\atr')) \in \hbo?$, and $((\apr',\atr'), \storeact(\apr', \atr))\in \cfo(y)$ where $\atr$ was not delayed after $\issueact(\apr_{1},\atr_1)$ in the trace, which means that $\tau'$ has less delays than $\tau$ (a contradiction to $\tau$ being a minimal violation). Therefore, we must have $x \neq y$ s.t. $(\issueact(\apr_1,\atr_1),a) \in (\sametro;\sto(x))\cup\cfo(x)$ and $(a, (\apr',\atr')) \in \hbo?$ and $((\apr',\atr'), \storeact(\apr', \atr))\in \cfo(y)$ which satisfies the second case of the lemma.
\end{proof}

We use $\traceccvone$ and $\traceccvtwo$ to denote the class of minimal violations that satisfy the first and second case in Lemma~\ref{lemma:CcvMinForm}, respectively. The following lemma shows that we can always obtain a minimal violation trace in either $\traceccvone$ or $\traceccvtwo$ where $\beta$ and $\beta_2$ contain no delayed transactions, respectively.  
We distinguish two cases in the proof: i) a minimal violation in $\traceccvtwo$ where $\atr$ and $\atr_1$ are distinct transactions, and ii) a minimal violation in $\traceccvone$ or in $\traceccvtwo$ where $\atr$ and $\atr_1$ coincide. 
In the first case, we show that if it were to have a delayed transaction in $\beta_2$, then it is possible to remove some transaction from $\tau$ that is causally dependent on the first delayed transaction in $\tau$, and obtain a new violation with a smaller number of delays (which contradicts the minimality assumption). 
The second case is proved by induction on the size of $\beta$ (note that if $\atr$ and $\atr_1$ coincide, then $\beta = \beta_2$) where the base case is trivial (i.e., $\beta= \epsilon$), and in the induction step, we show that if it were to have a delayed transaction in $\beta$ then we can remove one of the delayed transactions in the trace and obtain another violation with the same number of delays as the original violation and for which we can apply the induction hypothesis. 

%In the rest of this section we focus on minimal violations in $\traceccvtwo$ (the proofs concerning the ones in $\traceccvone$ is similar and simpler). 

%For a minimal violation $\tau=\atmacro\cdot \issueact(\apr,\atr) \cdot\beta_1\cdot \issueact(\apr_{1},\atr_1) \cdot\beta_2\cdot \storeact(\apr',\atr)\cdot\stmacro$, the transactions $\atr$ and $\atr_1$ are involved in a cycle in the transactional happens-before. By an abuse of terminology, we say that $(\issueact(\apr,\atr), \issueact(\apr_{1},\atr_1), \storeact(\apr',\atr))$ is a happens-before cycle.

\begin{lem}\label{lemma:CcvMinViol}
Let $\tau$ be a minimal violation in $\traceccvone$ or $\traceccvtwo$. Then, there exist a violation $\tau_1$ in $\traceccvone$ where $\beta$ contains no delayed transactions or a violation $\tau_2$ in $\traceccvtwo$ where $\beta_2$ contains no delayed transactions. %, and $\tau_1$ (resp., $\tau_2$) is obtained from $\tau$ ($\tau_1$ (resp., $\tau_2$) and $\tau$ can be equal).
\end{lem}

\begin{proof}
We consider two cases: i) $\tau$ in $\traceccvtwo$ where $\atr_1$ and $\atr$ are two distinct transactions, ii) $\tau$ in  $\traceccvone$ or $\tau$ in  $\traceccvtwo$ where $\atr_1$ and $\atr$ coincide. We prove the first case by contradiction and the second case by induction on the size $\beta$ (we abused terminology here and considered $\beta_2=\beta$ since $\beta_1=\epsilon$ in the second case). 

%in which we prove by contradiction that $\beta_2$ cannot contain  a delayed transaction. The second case is when $\atr_1$ and $\atr$ correspond to the same transaction, in which we extract new minimal violation trace where we remove the transaction $\atr$ from the original trace.

First case: let $\tau=\atmacro\cdot \issueact(\apr,\atr) \cdot\beta_{1} \cdot \issueact(\apr_{1},\atr_1) \cdot \beta_2 \cdot (\apr',\atr') \cdot \storeact(\apr',\atr)\cdot\stmacro$ and suppose by contradiction that $\beta_2$ contains a delayed transaction $\atr_0$ issued by a process $q \neq p$. W.l.o.g., we assume that the delayed store events of $\atr_0$ occur in $\beta_2$. Thus, $\beta_2 = \beta_{21} \cdot \issueact(q,\atr_0)  \cdot\beta_{22} \cdot \storeact(q',\atr_0)\cdot\beta_{23}$ and $\tau=\atmacro\cdot  \issueact(\apr,\atr) \cdot\beta_{1} \cdot \issueact(\apr_{1},\atr_1) \cdot \beta_{21} \cdot \issueact(q,\atr_0)  \cdot\beta_{22} \cdot \storeact(q',\atr_0)\cdot\beta_{23}\cdot  (\apr',\atr') \cdot\storeact(\apr',\atr)\cdot\stmacro$. 
In $\tau$, $\issueact(q,\atr_0)$ happens-before $\storeact(q',\atr_0)$ through $\beta_{22}$. 
Hence, we deduce that we can get a robustness violation when the event $\storeact(q',\atr_0)$ is executed, thus we can remove all issued transactions from $\beta_{23}\cdot  (\apr',\atr')$ except stores of already issued transactions and we obtain:
$\tau'=\atmacro\cdot  \issueact(\apr,\atr) \cdot\beta_{1} \cdot \issueact(\apr_{1},\atr_1) \cdot \beta_{21} \cdot \issueact(q,\atr_0)  \cdot\beta_{22} \cdot \storeact(q',\atr_0)\cdot\beta_{23}' \cdot\storeact(\apr',\atr)\cdot\stmacro$ which is a minimal violation because $\issueact(q,\atr_0)$ happens-before $\storeact(q',\atr_0)$ through $\beta_{22}$ and its number of delays is less or equal to the one of $\tau$. We know that in $\beta_{21} \cdot \issueact(q,\atr_0)  \cdot\beta_{22} \cdot \storeact(q',\atr_0)\cdot\beta_{23}' \cdot\storeact(\apr',\atr)\cdot\stmacro$ there are no transactions from the process $\apr_1$ or that see the effect of transactions from $\apr_1$  (because of the $\hbo\setminus\viso$ relation between events $\beta_2$ and $\issueact(\apr_1,\atr_1)$). Therefore, $\issueact(\apr_1,\atr_1)$ is the last issued transaction from $\apr_1$ and we do not have any transaction in $\tau'$ that depends on it. Thus, we can remove $\issueact(\apr_1,\atr_1)$ and we obtain the following trace:
$\tau'' =\atmacro\cdot  \issueact(\apr,\atr) \cdot\beta_{1} \cdot  \beta_{21} \cdot \issueact(q,\atr_0)  \cdot\beta_{22} \cdot \storeact(q',\atr_0)\cdot\beta_{23}' \cdot\storeact(\apr',\atr)\cdot\stmacro'$, which is a robustness violation because $\issueact(q,\atr_0)$ happens-before $\storeact(q',\atr_0)$ through $\beta_{22}$. $\tau''$ has less delays than $\tau'$ ($\storeact(\apr',\atr)$ was not delayed after $\issueact(\apr_1,\atr_1)$ which was removed), which is a contradiction to the fact that  $\tau$ is a minimal violation.

Second case: let $\tau=\atmacro\cdot \issueact(\apr,\atr) \cdot\beta \cdot (\apr',\atr') \cdot \storeact(\apr',\atr)\cdot\stmacro$. We show by induction that we can construct either $\tau_1$ in $\traceccvone$ where $\beta$ of $\tau_1$ contains no delayed transactions or $\tau_2$ in $\traceccvtwo$ where $\beta_2$ of $\tau_2$ contains no delayed transactions. 

\noindent
\textbf{Base case:} 
$\length{\beta} =0$ is trivial. 

\noindent
\textbf{Induction step:} We assume that the induction hypothesis holds for $\length{\beta} \leq m$. Let $\atr_0$ be the first delayed transaction in $\beta$. Similar to before, we assume w.l.o.g. that the delayed store events of $\atr_0$ occurs in $\beta$. Thus, $\beta = \beta_{01} \cdot \issueact(q,\atr_0)  \cdot\beta_{02} \cdot \storeact(q',\atr_0)\cdot\beta_{03}$ and 
$\tau=\atmacro\cdot  \issueact(\apr,\atr) \cdot \beta_{01} \cdot \issueact(q,\atr_0)  \cdot\beta_{02} \cdot \storeact(q',\atr_0)\cdot\beta_{03}\cdot (\apr',\atr')\cdot  \storeact(\apr',\atr)\cdot\stmacro$ where $\issueact(q,\atr_0)$ happens-before $\storeact(q',\atr_0)$ through $\beta_{02}$. 
Using the same arguments as before, we can remove the event $\issueact(\apr,\atr)$, its related stores in $\tau$, and all issued transactions in $\beta_{03}\cdot  (\apr',\atr')$. We obtain: $\tau'=\atmacro' \cdot \issueact(q,\atr_0) \cdot\beta_{02} \cdot \storeact(q',\atr_0)\cdot\stmacro'$ where $\atmacro' = \atmacro\cdot \beta_{01}$ and $\issueact(q,\atr_0)$ happens-before $\storeact(q',\atr_0)$ through $\beta_{02}$. $\tau'$ is a robustness violation, and it has the same number of delays as $\tau$.  We now consider two possible case of $\tau'$: i) $\tau'$ is in $\traceccvtwo$ where $\atr_0$ and $\atr_{01}$, the last delayed transaction causally dependent on $\issueact(q,\atr_0)$ in $\tau'$, are two distinct transactions, or ii) $\tau$ in $\traceccvone$ or $\tau$ in $\traceccvtwo$ where $\atr_0$ and $\atr_{01}$ coincide. From the first part of the proof, it is guaranteed that in the first case  there are no delayed transactions after $\atr_{01}$. For the second case, we use the induction hypothesis since  $\length{\beta_{02}} \leq m$ ($\beta_{02}$ is a strict subsequence of $\beta$). 
\end{proof}

We have now showed all the necessary characterizations for minimal violations that fall under the first pattern (i.e., $\traceccvone$). In the rest of this section, we focus on minimal violations that fall under the second pattern (i.e., $\traceccvtwo$). In particular, we look at minimal violations in $\traceccvtwo$ where $\atr$ and $\atr_1$ are distinct transactions. 
%For simplicity, we assume that the class $\traceccvtwo$ is restricted to violations where $\beta_2$ contains no delayed transactions. 
In the following lemma, we show that for these minimal violations the issue events of delayed transactions in $\issueact(\apr,\atr) \cdot\beta_1\cdot \issueact(\apr_{1},\atr_1)$ constitute a causality chain. 
Our proof can be decomposed to two parts. In the first part, we show that we cannot have an issue event of a delayed transaction in $\beta_{1}\cdot \issueact(\apr_{1},\atr_1)$ that is not causally dependent on $\issueact(\apr,\atr)$. We prove this by showing that if this were possible then we can remove a transaction that is causally dependent on one of the two delayed transactions and obtain a new violation trace with less delays than the original violation (which contradicts the minimality assumption). For the second part, we show that for a given minimal violation, we can construct a happens-before equivalent trace where for every two successive issue events of delayed transactions in $\issueact(\apr,\atr) \cdot\beta_1\cdot \issueact(\apr_{1},\atr_1)$, the transaction in the first issue is delivered to the process executing the second issue before this event happens.  %in all delayed itransactions in $\beta_{1}$ constitute causality chain that starts with $\issueact(\apr,\atr)$ and ends at $\issueact(\apr_{1},\atr_1)$.

\begin{lem}\label{lemma:CCvcausalitychain}
Let $\tau=\atmacro\cdot \issueact(\apr,\atr) \cdot\beta_1\cdot \issueact(\apr_{1},\atr_1) \cdot\beta_2\cdot(\apr',\atr')\cdot \storeact(\apr',\atr)\cdot\stmacro$ be a minimal violation in $\traceccvtwo$ s.t $\atr \neq \atr_1$ and $\beta_2$ contains no delayed transactions (cf. Lemma \ref{lemma:CcvMinViol}). Then, there exists a violation $\tau' = \atmacro\cdot \issueact(\apr,\atr) \cdot\beta_1'\cdot \issueact(\apr_{1},\atr_1) \cdot\beta_2\cdot (\apr',\atr')\cdot \storeact(\apr',\atr)\cdot\stmacro'$ obtained from $\tau$ where $\beta_1'\cdot\stmacro'$ is a subsequence of $\beta_1\cdot\stmacro$ and the sequence of issue events of delayed transactions forms a causality chain that starts with $\issueact(\apr,\atr)$ and ends at $\issueact(\apr_1,\atr_1)$. 
\end{lem}

\begin{proof}
First, we show that we can obtain a violation $\tau'$ from $\tau$ where all delayed transactions in $\beta_1'\cdot \issueact(\apr_{1},\atr_1)$ are causally dependent on $\issueact(\apr,\atr)$. From the definition of $\atr_1$ in Lemma \ref{lemma:CcvMinForm}, we already have that 
$(\issueact(\apr,\atr),\issueact(\apr_{1},\atr_1)) \in \viso$. In the proof, we assume w.l.o.g that in $\beta_1\cdot \issueact(\apr_{1},\atr_1)\cdot\beta_2$ there is no event $a$ that reads a value that $\atr$ overwrites, otherwise, we can shortcut the trace by removing $(\apr',\atr')$ and instead using the conflict relation between $a$ and a store event of $\atr$ to build the transactional happens-before cycle. 
Now, assume that $\beta_{1}$ contains a delayed transaction $\atr_0$ from another process $q \neq p$ that is not causally dependent on $\issueact(\apr,\atr)$. We show that we either can obtain a contradiction or we can remove all events of $\atr_0$ and obtain a new violation trace $\tau'$ that has the same number of delays as $\tau$. We have three possible cases based on whether the delayed store event $\storeact(q',\atr_0)$ of $\atr_0$ occurs in $\beta_1$, $\beta_2$ or $\stmacro$. Hence, we get that $\tau$ can be one of the following:
\begin{enumerate}[label=(\alph*)]
\item $\tau=\atmacro\cdot  \issueact(\apr,\atr) \cdot\beta_{11}\cdot \issueact(q,\atr_0) \cdot \beta_{12} \cdot\storeact(q',\atr_0)\cdot\beta_{13}\cdot \issueact(\apr_{1},\atr_1) \cdot\beta_2\cdot(\apr',\atr')\cdot \storeact(\apr',\atr)\cdot\stmacro$
\item $\tau=\atmacro\cdot  \issueact(\apr,\atr) \cdot\beta_{11}\cdot \issueact(q,\atr_0) \cdot \beta_{12}\cdot \issueact(\apr_{1},\atr_1) \cdot\beta_{21}\cdot \storeact(q',\atr_0)\cdot\beta_{22}\cdot(\apr',\atr')\cdot \storeact(\apr',\atr)\cdot\stmacro$
\item $\tau=\atmacro\cdot  \issueact(\apr,\atr) \cdot\beta_{11}\cdot \issueact(q,\atr_0) \cdot \beta_{12}\cdot \issueact(\apr_{1},\atr_1) \cdot\beta_2\cdot(\apr',\atr')\cdot \storeact(\apr',\atr)\cdot\stmacro^{1}\cdot \storeact(q',\atr_0)\cdot \stmacro^{2}$
\end{enumerate}
In case $(a)$ (resp., $(b)$) we can notice that since $\issueact(q,\atr_0)$ happens-before $\storeact(q',\atr_0)$ through $\beta_{12}$ (resp., $\beta_{12}\cdot \issueact(\apr_{1},\atr_1) \cdot\beta_{21}$) then after executing $\storeact(q',\atr_0)$ we obtain a cycle in the transactional happens-before. Thus, we can remove $(\apr',\atr')$ from both traces and still obtain a robustness violation. Let $\tau'$ be the resulting trace. $\tau'$ has the same number of delays as $\tau$. In $\tau'$, we do not have events that read values that $\atr$ overwrites. Therefore, we do not need to delay the transaction $\atr$ to ensure that that the trace is a violation. Let $\tau''$ be the resulting trace where the transaction $\atr$ executes atomically. In $\tau''$, the transaction $\atr$ was not delayed after 
the issue event of $\atr_1$ which means that $\tau''$ has less delays than $\tau$. This contradicts the fact that $\tau$ is a minimal violation.

Case $(c)$:  we assume that $\storeact(q',\atr_0)$ happens-before after $(\apr',\atr')$, otherwise, we can reorder it before $(\apr',\atr')$ and get case $(b)$. Since $\stmacro^{1}$ contains only store events, then by the happens-before definition, $\storeact(q',\atr_0)$ must be a store event executed by $\apr'$ which means that $q' = \apr'$.  
Let $e_1$ and $e_2$ be the read/write actions $\atr'$ that are the source of the conflict between $(\apr',\atr')$ and $\storeact(\apr',\atr)$ and the happens-before between $(\apr',\atr')$ and $\storeact(\apr',\atr_0)$, respectively. 
Similar to before we assume w.l.o.g that there is no event in $\beta_{12}\cdot \issueact(\apr_{1},\atr_1) \cdot\beta_2$ that reads a value that $\atr_0$ overwrites. We consider the two cases: i) $e_2$ occurs before $e_1$ in $\atr'$ or the two coincide, and ii) $e_1$ occurs before $e_2$ in $\atr'$. In the first case we can obtain a new violation where we do not delay the transaction $\atr$ which will not affect the action $e_2$ that is the source of the happens-before between $(\apr',\atr')$ and $\storeact(\apr',\atr_0)$ (since $e_1$ occurs after $e_2$ then it cannot disable it). 
The new trace $\tau'$ is a violation since the store event $\storeact(\apr',\atr_0)$ is delayed. Also, since the store event $\storeact(\apr',\atr)$ of $\atr$ was not delayed after $\issueact(\apr_{1},\atr_1)$ then $\tau'$ has less delays than $\tau$, which contradicts the fact that $\tau$ is a minimal violation. In the second case, if in $\beta_{12}\cdot \issueact(\apr_{1},\atr_1) \cdot\beta_2$ we do not have any event that is causally dependent on $\issueact(q,\atr_0)$ other than the store events of $\atr_0$, then we can remove all events of $\atr_0$ from $\tau$ without affecting the happens before between $\issueact(\apr,\atr)$ and $\storeact(\apr',\atr)$ through  $\beta_{12}\cdot \issueact(\apr_{1},\atr_1) \cdot\beta_2 \cdot (\apr',\atr')$. Let $\tau' = \atmacro\cdot \issueact(\apr,\atr) \cdot\beta_1'\cdot \issueact(\apr_{1},\atr_1) \cdot\beta_2\cdot (\apr',\atr')\cdot \storeact(\apr',\atr)\cdot\stmacro'$ be the resulting trace which has the same number of delays as $\tau$. Otherwise, if in $\beta_{12}\cdot \issueact(\apr_{1},\atr_1) \cdot\beta_2$ we have an event $a$ that is causally dependent on $\issueact(q,\atr_0)$ that is not a store event of $\atr_0$, then the new trace $\atr'$ resulting from not delaying $\atr_0$ is a violation. This is because the store event $\storeact(\apr',\atr)$ is delayed. $\tau'$ has less delays than $\tau$ since the store event $\storeact(\apr',\atr_0)$ of $\atr_0$ was not delayed after $a$. This contradicts the fact that $\tau$ is a minimal violation.

Now, we show that for every two successive issue events of delayed transactions in $\tau'$, we can deliver the first to the process of the second before the second is issued. 
Let $\event_{i} = \issueact(\apr_{i},\atr_{i})$ and $\event_{j} = \issueact(\apr_{j},\atr_{j})$ be two successive issue events of delayed transactions in $\beta_1\cdot \issueact(\apr_{1},\atr_{1})$ s.t. either $(\event_{i},\event_{j})\in\hbo$ or $(\event_{i},\storeact(\apr_{i},\atr_{j}))\in\hbo$. 
Note that the only case where the store event $\storeact(\apr_{j},\atr_{i})$ cannot be moved to occur before $\event_{j}$ in 
$\beta_1$ is when the two events are related by a happens-before relation, i.e., $(\event_{j},\storeact(\apr_{j},\atr_{i}))\in\hbo$. In this case, we get that the transactions $\atr_i$ and $\atr_j$ are involved in a cycle in the transactional happens-before in $\tau'$ which means that 
$\tau'' = \atmacro\cdot (\apr,\atr) \cdot\beta_1\cdot \issueact(\apr_{1},\atr_{1}) \cdot \stmacro$ is a violation which has less delays than $\tau$ (since $\atr$ was not delayed after $\issueact(\apr_{1},\atr_{1})$). Therefore, the trace $\tau''$ where the store event $\storeact(\apr_{j},\atr_{i})$ occurs before $\event_{2}$ is happens-before equivalent to $\tau'$.  Similarly, when the two events are concurrent, the trace $\tau''$ where the store event $\storeact(\apr_{j},\atr_{i})$ occurs before $\event_{j}$ is happens-before equivalent to $\tau'$. Thus, given the sequence of issue events $\event_1\cdot\event_2\cdot\ldots\event_n$ of delayed transactions in $\tau'$ s.t. $\event_{1} = \issueact(\apr,\atr)$ and $\event_{n} = \issueact(\apr_{1},\atr_{1})$, the trace $\tau''$ where  
 for every $1\leq k\leq n-1$ s.t. $\event_{k} = \issueact(\apr_{k},\atr_{k})$ and $\event_{k+1} = \issueact(\apr_{k+1},\atr_{k+1})$, we have the store event $\storeact(\apr_{k+1},\atr_{k})$ occurs before the issue event $\event_{k+1}$ is happens-before equivalent to $\tau'$.  Also, in $\tau''$ for every $2\leq k\leq n$, we have that $\event_{k}$ is causally dependent on $\event_{1} = \issueact(\apr,\atr)$. Thus, in $\tau''$ the sequence of issue events $\event_1\cdot\event_2\cdot\ldots\event_n$ of delayed transactions forms a causality chain. 
\end{proof}

Next, we show that for minimal violations in $\traceccvtwo$ where $\atr$ and $\atr_1$ are distinct transactions, all delayed transactions in $\issueact(\apr,\atr) \cdot \beta_1$ do not access the shared variable $\anaddr$ that starts the happens-before path in $\beta_2$ (Lemma~\ref{lemma:CcvMinForm}) between $\issueact(\apr_1,\atr_1)$ and $(\apr',\atr')$. If this were not the case, then the events of $\atr_1$ can be removed and we still guarantee a happens-before path to $\storeact(\apr',\atr)$ (starting in the delayed transaction accessing the variable $\anaddr$), thus obtaining a new robustness violation trace with less delays (since $\storeact(\apr',\atr)$ was not delayed after $\issueact(\apr_1,\atr_1)$), which contradicts the minimality assumption.

%Also, we show that $\anaddr$ is different than the variable $y$ that ends the happens-before path in $\beta_2$ (Lemma~\ref{lemma:CcvMinForm}). %$\tau=\atmacro\cdot \issueact(\apr,\atr) \cdot\beta_1\cdot \issueact(\apr_{1},\atr_1) \cdot\beta_2\cdot \storeact(\apr',\atr)\cdot\stmacro$ be a minimal violation s.t. there exist $a$ and $b$ in $\beta_2$ s.t. $(\issueact(\apr_{1},\atr_1),a)\in\sto(\anaddr)\cup (\sametro;\cfo(\anaddr))$  and $(b,\storeact(\apr', \atr))\in \cfo(y)\cup\sto(y)$.

\begin{lem}\label{lemma:NoAccessX}
Let $\tau$ be a minimal violation in $\traceccvtwo$ where $\atr_1$ and $\atr$ are two distinct transactions. Then, all the delayed transactions in $\issueact(\apr,\atr) \cdot \beta_1$ do not access the variable $\anaddr$ from Lemma~\ref{lemma:CcvMinForm}.
\end{lem}

\begin{proof}
Suppose by contradiction that we have an issue event $\issueact(\apr_{2},\atr_2)$ in $\issueact(\apr,\atr) \cdot \beta_1$ (i.e., $\issueact(\apr,\atr) \cdot \beta_1 = \issueact(\apr,\atr) \cdot\beta_{11}\cdot \issueact(\apr_{2},\atr_2)\cdot\beta_{12}$)  which accesses the shared variable $\anaddr$ with either a read or a write instruction. Then, since there exists an event $a\in \beta_2$ s.t. $(\issueact(\apr_{1},\atr_1),a) \in \sto(\anaddr) \cup (\sametro;\cfo(\anaddr))$, we have that $(\issueact(\apr_{2},\atr_2),a) \in \sto(\anaddr) \cup (\sametro;\cfo(\anaddr))$. Moreover, because $\beta_2\cdot (\apr',\atr')\cdot \storeact(\apr',\atr)\cdot\stmacro$ does not contain any transaction that causally depends on $\issueact(\apr_{1},\atr_1)$, we get that $\issueact(\apr_{1},\atr_1)$ is the issue event by the process $\apr_{1}$ and we can remove it together with all the related stores in $\stmacro$ to obtain: $\tau'=\atmacro\cdot \issueact(\apr,\atr) \cdot\beta_{11}\cdot \issueact(\apr_{2},\atr_2)\cdot\beta_{12}\cdot\beta_2\cdot (\apr',\atr')\cdot \storeact(\apr',\atr)\cdot\stmacro'$ which is a violation because $\issueact(\apr,\atr)$ happens-before $\storeact(\apr',\atr)$ through $\beta_{11}\cdot \issueact(\apr_{2},\atr_2)\cdot\beta_{12}\cdot\beta_2\cdot (\apr',\atr')$. Furthermore, $\tau'$ has less delays than $\tau$ since $\storeact(\apr',\atr)$ was not delayed after  $\issueact(\apr_{1},\atr_1)$. This contradicts the fact that $\tau$ is a minimal violation. 
%If $\atr$ and $\atr_1$ are different then $\anaddr \neq y$, otherwise, $\atr$ accesses $\anaddr$. In the case $\atr_1 = \atr$, then since $a$ is an atomic macro event then $(\storeact(\atr,\atr),a) \in \sto(x)$ thus we cannot have $(b,\storeact(\apr', \atr)) \in \cfo(y) \mbox{ and } x=y$ which implies that when $x=y$ we can only have $(\storeact(\atr,\atr),a) \in \sto(x)$, $(a,b) \sto(x)$, and $(b,\storeact(\apr', \atr)) \in \cfo(x)$ which is impossible since $\sto$ is acyclic under \ccvt{}.
\end{proof}

\begin{figure}[t]
  \begin{minipage}[c]{\textwidth}
  \begin{subfigure}{\linewidth}
  \scalebox{0.71}
  {
  \begin{tikzpicture}
  \node[text width=1cm, anchor=west, left,font=\LARGE] at (0,0) {$\tau_{\scct{}1:}$};
   \node[shape=rectangle ,draw=none,font=\large] (A0) at (0.5,0)  [] {};
   \node[shape=rectangle ,draw=none,font=\Large] (A1) at (1.5,0.3)  [] {$\atmacro$};
   \node[] (A) at (1.5,0) [] {};
   \node[shape=rectangle ,draw=none,font=\Large] (B1) at (2.8,0.5)  [] {$\issueact(\apr,\atr)$};
   \node[shape=rectangle ,draw=none,font=\huge] (B) at (2.8,0)  [] {$\circ$};
   \node[shape=rectangle ,draw=none,font=\Large] (B2) at (5,0.5)  [] {$\storeact(\apr,\atr)$};
   \node[shape=rectangle ,draw=none,font=\huge] (B3) at (5,0)  [] {$\circ$};
   \node[shape=rectangle ,draw=none,font=\Large] (C1) at (8,0.3)  [] {$\beta$};
    \node[] (C) at (8,0)  [] {};
    \node[shape=rectangle ,draw=none,font=\Large] (D1) at (10.5,0.5)  [] {$(\apr',\atr')$};
    \node[shape=rectangle ,draw=none,font=\huge] (D) at (10.5,0)  [] {$\circ$};
    \node[shape=rectangle ,draw=none,font=\Large] (E1) at (13,0.5)  [] {$\storeact(\apr',\atr)$};
    \node[shape=rectangle ,draw=none,font=\huge] (E) at (13,0)  [] {$\circ$};
    \node[shape=rectangle ,draw=none,font=\Large] (F1) at (14.5,0.3)  [] {$\stmacro$};
    \node[] (F) at (14.5,0)  [] {};
    \node[shape=rectangle ,draw=none,font=\large] (F0) at (16,0)  [] {};
  
    \begin{scope}[ every edge/.style={draw=black}]
    \path [|-|] (A0) edge []  (B);
    \path [|-|] (B3) edge [] (D);
    \path [->] (B1) edge [bend left=20] node [above,font=\small] {$\hbo\setminus \viso$} node [pos=0.95,above,font=\small] {$\forall$} (C1);
    \path [->] (C1) edge  [bend left=20] node [above,font=\small] {$\viso$} node [pos=0.05,above,font=\small] {$\forall$} (D1);
    \path [->] (B3) edge [bend right=20] node [below,font=\small] {$\sto(x)$} (D);
    \path [->] (D) edge [bend right=30] node [below,font=\small] {$\sto(x)$} (E);
    \path [|-|] (E) edge [] (F0);
    \end{scope}
  \end{tikzpicture}}
  \label{fig:rob0tracecm}
  \end{subfigure}
  \end{minipage}
  \begin{minipage}[c]{\textwidth}
  \begin{subfigure}{\linewidth}
  \scalebox{0.71}
  {\begin{tikzpicture}
    \node[text width=1cm, anchor=west, left,font=\LARGE] at (0,0) {$\tau_{\scct{}2:}$};
    \node[shape=rectangle ,draw=none,font=\large] (A0) at (0.5,0)  [] {};
    \node[shape=rectangle ,draw=none,font=\Large] (A1) at (1.5,0.3)  [] {$\atmacro$};
    \node[] (A) at (1.5,0) [] {};
    \node[shape=rectangle ,draw=none,font=\Large] (B1) at (3,0.5)  [] {$\issueact(\apr,\atr)|_{\neg x}$};
    \node[shape=rectangle ,draw=none,font=\huge] (B) at (3,0)  [] {$\circ$};
    \node[shape=rectangle ,draw=none,font=\Large] (C1) at (5,0.3)  [] {$\beta_1|_{\neg x}$};
    \node[] (C) at (5,0) [] {};
    \node[shape=rectangle ,draw=none,font=\Large] (D1) at (7,0.5)  [] {$\issueact(\apr_{1},\atr_{1})$};
    \node[shape=rectangle ,draw=none,font=\huge] (D) at (7,0)  [] {$\circ$};
    \node[shape=rectangle ,draw=none,font=\Large] (E1) at (11,0.3)  [] {$\beta_2$};
    \node[] (E) at (11,0) [] {};
    \node[shape=rectangle ,draw=none,font=\Large] (H1) at (14,0.5)  [] {$(\apr',\atr')$};
    \node[shape=rectangle ,draw=none,font=\huge] (H) at (14,0)  [] {$\circ$};
    \node[shape=rectangle ,draw=none,font=\Large] (F1) at (17,0.5)  [] {$\storeact(\apr',\atr)$};
    \node[shape=rectangle ,draw=none,font=\huge] (F) at (17,0)  [] {$\circ$};
    \node[shape=rectangle ,draw=none,font=\Large] (G1) at (18.5,0.3)  [] {$\stmacro$};
    \node[] (G) at (18.5,0) [] {};
    \node[shape=rectangle ,draw=none,font=\large] (G0) at (19.5,0)  [] {};   
    \begin{scope}[ every edge/.style={draw=black}]
    \path [|-|] (A0) edge []  (B);
    \path [|-|] (B) edge [] (D);
    \path [->] (D1) edge [bend left=20] node [above,font=\small] {$\hbo\setminus \viso$} node [pos=0.95,above,font=\small] {$\forall$} (E1);
    \path [|-|] (D) edge [] (H);
    \path [->] (B) edge [bend right=20] node [below,font=\small] {$\viso$} (D);
    \path [->] (D) edge [bend right=25] node [below,font=\small] {$\cfo(x)$} node [pos=0.95,below,font=\small] {$\exists$} (E);
    \path [->] (E1) edge [bend left=20] node [above,font=\small] {$\hbo$} node [pos=0.05,above,font=\small] {$\forall$} (H1);
    \path [->] (H) edge [bend right=20] node [below,font=\small] {$\cfo(y\neq x)$}  (F);
    \path [|-|] (F) edge []  (G0);
    \end{scope}
  \end{tikzpicture}}
  \label{fig:rob1tracecm}
  \end{subfigure}
  \end{minipage}
  \caption{Robustness violation patterns under \scct{}. For violation $\tau_{\scct{}1}$, $\atr$ is the only delayed transaction.
  For $\tau_{\scct{}2}$, all delayed transactions are in $\issueact(\apr,\atr)\cdot\beta_1\cdot\issueact(\apr_{1},\atr_{1})$ and they form a causality chain that starts at $\issueact(\apr,\atr)$ and ends at $\issueact(\apr_{1},\atr_{1})$.}
  \label{fig:cmrobtraces}
  \end{figure}

\section{Robustness Violations Under Causal Memory} \label{sec:MVUCM}
The characterization of robustness violations under $\scct{}$ is at some level similar to that of robustness violations under $\ccvt{}$. However, some instance of the violation pattern under \ccvt{} is not possible under \scct{} and \scct{} admits some class of violations that is not possible under \ccvt{}. This reflects the fact that these consistency models are incomparable in general. 

The following theorem gives the characterization of minimal violations under \scct{} which is pictured in Fig.~\ref{fig:cmrobtraces}.  Roughly, a program is not robust iff it admits a violation that either contains two concurrent transactions that write to the same variable, or it is a restriction of the pattern $\tau_{\ccvt{}2}$ admitted by \ccvt{} where the last delayed transaction is related only by $\cfo$ to future transactions. The first pattern is not admitted by $\ccvt{}$ because the writes to each variable are executed according to the timestamp order (\scct{} does not satisfy the  $\ccvt{}$ property stated in Lemma \ref{lemma:CcvProperty}).

\begin{thm} \label{theorem:CmMinViol}
A program $\aprog$ is not robust under \scct{} iff there exists a minimal violation in $\tracesconf(\aprog)_{\scct{}}$ of one of the following forms:
\begin{enumerate}[topsep=5pt]
\item  $\tau_{\scct{}1}=\atmacro\cdot \issueact(\apr,\atr) \cdot \storeact(\apr,\atr) \cdot \beta \cdot (\apr',\atr') \cdot \storeact(\apr',\atr)\cdot\stmacro$, where:
\begin{enumerate}[label=(\alph*),topsep=5pt]
\item $\issueact(\apr,\atr)$ is the issue of the first and only delayed transaction;
\item  $\exists\ y.$ s.t. $(\storeact(\apr,\atr),(\apr',\atr')) \in \sto(y)$ and $((\apr',\atr'),\storeact(\apr',\atr)) \in  \sto(y)$ (Lemma \ref{lemma:CmMinForm});
\item $\forall\ a \in \beta.\ (\issueact(\apr,\atr),a) \in \hbo\setminus\viso \mbox{ and } (a,(\apr',\atr')) \in \viso$ (Lemma \ref{lemma:CmMinForm}).
%\item  $\stmacro$ contains only stores of $\atr$ (Lemma \ref{lemma:StoresSuffix}).\\
\end{enumerate}
\item $\tau_{\scct{}2}=\atmacro\cdot \issueact(\apr,\atr) \cdot\beta_1\cdot \issueact(\apr_{1},\atr_1) \cdot\beta_2\cdot (\apr',\atr')\cdot  \storeact(\apr',\atr)\cdot\stmacro$, where
\begin{enumerate}[label=(\alph*),topsep=5pt]
\item $\issueact(\apr,\atr)$ and $\issueact(\apr_{1},\atr_1)$ are the issues of the first and last delayed transactions (Lemma~\ref{lemma:CmMinForm});
\item the issues of all delayed transactions are in $\beta_1$ are included in a causality chain that starts with $\issueact(\apr,\atr)$ and ends with $\issueact(\apr_{1},\atr_1)$;
\item for every $a\in\beta_2$, we have that $(\issueact(\apr_1,\atr_1),a) \in \hbo\setminus\viso$ and $(a,(\apr',\atr')) \in \hbo$ (Lemma \ref{lemma:CmMinForm});
\item there exist $a \in \beta_2\cdot (\apr',\atr')$, $x$, and $y$ s.t. $x \neq y$, $(\issueact(\apr_1,\atr_1),a) \in \cfo(x)$, $(a,(\apr',\atr'))\in \hbo{}?$, and $((\apr',\atr'),\storeact(\apr',\atr)) \in \cfo(y)$ (Lemma \ref{lemma:CmMinForm});
\item  all delayed transactions in $\issueact(\apr,\atr) \cdot \beta_1$ do not access the variable $\anaddr$.
%\item  \mbox{$\stmacro$ contains only stores of delayed transactions that were issued in $\beta_1$ (Lemma \ref{lemma:StoresSuffix}).}
\end{enumerate}
\end{enumerate}
\end{thm}

\begin{wrapfigure}{r}{0.35\textwidth} 
  \scalebox{0.61}
  {
  \begin{tikzpicture}
  
   \node[shape=rectangle ,draw=none,font=\large] (A) at (0,0)  [] {$\issueact(\apr 1,\atr 1)$};
   \node[shape=rectangle ,draw=none,font=\large] (A1) at (1.9,0)  [] {$\storeact(\apr 1,\atr 1)$};
    \node[shape=rectangle ,draw=none,font=\large] (B) at (4,0)  [] {$(\apr 2,\atr 2)$};
    \node[shape=rectangle ,draw=none,font=\large] (C) at (6,0)  [] {$\storeact(\apr 2,\atr 1)$};
  
    \begin{scope}[ every edge/.style={draw=red,very thick}]
    \path [->] (A1) edge [bend left] node [above,font=\large] {$\sto$} (B);
    \path [->] (B) edge [bend left] node [above,font=\large] {$\sto$} (C);
    \end{scope}
  \end{tikzpicture}}
  \caption{Violation of LU program in Fig.~\ref{fig:rob0}. A $\tau_{\scct{}1}$ violation where $\beta_2=\stmacro=\epsilon$, and $\atr$ and $\atr'$ correspond to $\atr 1$ and $\atr 2$.}
  \label{fig:rob0LUtracecm}
\end{wrapfigure}
The violation pattern $\tau_{\scct{}2}$ is a restriction of the pattern $\tau_{\ccvt{}2}$ under \ccvt{}. For instance, the trace in Fig.\ref{fig:rob1SBtraceccv} is a valid minimal violation of the SB program under \scct{}. The violation pattern $\tau_{\scct{}1}$ implies the existence of a write-write race under \scct{}. Fig.\ref{fig:rob0LUtracecm} shows a minimal violation under \scct{} that corresponds to a write-write race in the LU program. Conversely, if a program $\aprog$ admits a trace $\tau$ which contains a write-write race under \scct{}, then $\aprog$ also admits a trace $\tau'$ where the two transactions $\atr_1$ and $\atr_2$ that caused the write-write race form a cycle in the store order %\footnote{That is, for the two concurrent transactions $\atr_1$ and $\atr_2$ that write on a common variable $x$, $\storeact(\apr_1,\atr_1)$ is in store order before $\storeact(\apr_1,\atr_2)$ and $\storeact(\apr_2,\atr_2)$ before $\storeact(\apr_2,\atr_1)$.} 
(the store events of $\atr_1$ and $\atr_2$ on the two processes $\apr_1$ and $\apr_2$ that issued them can be reordered to occur in opposite orders, i.e., $\storeact(\apr_1,\atr_1)$ before $\storeact(\apr_1,\atr_2)$ and $\storeact(\apr_2,\atr_2)$ before $\storeact(\apr_2,\atr_1)$, which implies that are also in opposite orders w.r.t. the store order).
Thus, $\aprog$ has a trace $\tau'$ with a cycle in the transactional happens-before which means that $\aprog$ is not robust against \scct{}.
Therefore, a program which is robust against \scct{} is also write-write race free under \scct{}.
Since without write-write races, the \scct{} and the \ccvt{} semantics coincide, we get the following the result.

\begin{lem}\label{lemma:CcvCmrobust}
If a program $\aprog$ is robust against \scct{}, then $\aprog$ is robust against \ccvt{}.
\end{lem}

Next, we discuss the proof of Theorem~\ref{theorem:CmMinViol}. The following lemma reveals the two possible minimal violation patterns under causal memory. The characterization of the patterns in this lemma can be refined further using arguments similar to the case of \ccvt{} (see the discussion at the end of this section).
%This is the only lemma that needs different arguments than those used for causal convergence.
%This lemma also characterizes the happens-before dependency between the events of the first delayed transaction, which corresponds to the cause of a cycle in the transactional happens-before.

\begin{lem}\label{lemma:CmMinForm}
If $\aprog$ is a program that is not robust under \scct{}, then it must admit a minimal violation $\tau$ that satisfies one of the following:
\begin{enumerate}[topsep=5pt]
  \item $\tau = \atmacro \cdot \issueact(\apr,\atr) \cdot \storeact(\apr,\atr) \cdot \beta \cdot (\apr',\atr') \cdot \storeact(\apr',\atr)  \cdot \stmacro \mbox{ where: }$
\begin{enumerate}[label=(\alph*),topsep=5pt]
\item $\exists\ y.$ s.t. $(\storeact(\apr,\atr),(\apr',\atr')) \in \sto(y)$ and $((\apr',\atr'),\storeact(\apr',\atr)) \in  \sto(y)$;
\item $\forall\ a \in \beta.\ (\issueact(\apr,\atr),a) \in \hbo\setminus\viso \mbox{ and } (a,(\apr',\atr')) \in \viso$. 
\end{enumerate}
%\item  $\tau = \atmacro \cdot \issueact(\apr,\atr) \cdot \beta \cdot \storeact(\apr',\atr) \cdot \stmacro \mbox{  where:}$
%\begin{enumerate}[label=(\alph*)]
%\item there exist $a$ and $b=(\apr',\atr')$ in $\beta$ such that $(\issueact(\apr,\atr),a) \in \cfo(x)$, $(b,\storeact(\apr',\atr)) \in  \cfo(y)$, and $(a,b)\in\hbo{}?$;
%\item $\forall\ a \in \beta.\ (\issueact(\apr,\atr),a) \in \hbo\setminus\viso \mbox{ and }  (a,\storeact(\apr',\atr)) \in \hbo$.
%\end{enumerate}
\item $\tau = \atmacro \cdot \issueact(\apr,\atr) \cdot \beta_1 \cdot \issueact(\apr_{1},\atr_1)  \cdot \beta_2 \cdot (\apr',\atr')\cdot \storeact(\apr',\atr)  \cdot \stmacro \mbox{ where: }$
\begin{enumerate}[label=(\alph*),topsep=5pt]
\item  $\issueact(\apr_{1},\atr_1)$ is the last issue event from $\{c \in \beta\ |\ (\issueact(\apr,\atr),c) \in \viso\}$ in $\tau$;
\item there exist two variables $x\neq y$, $a$ in $\beta_2\cdot (\apr',\atr')$, and $b=(\apr',\atr')$ such that $(\issueact(\apr_1,\atr_1),a) \in \cfo(x)$, $(b,\storeact(\apr',\atr)) \in  \cfo(y)$, and $(a,b)\in\hbo{}?$;
\item $\forall\ a \in \beta_2.\ (\issueact(\apr_1,\atr_1),a) \in \hbo\setminus\viso \mbox{ and }  (a,(\apr',\atr')) \in \hbo$.
\end{enumerate}
\end{enumerate}
\end{lem}
\begin{proof}
The proof will contain many arguments which are similar to those used in the proof of Lemma~\ref{lemma:CcvMinForm}. 
Let $\tau = \atmacro  \cdot \issueact(\apr,\atr) \cdot \beta  \cdot (\apr',\atr')\cdot \storeact(\apr',\atr)\cdot\stmacro$ be a minimal violation under \scct{} (cf. Lemma~\ref{lemma:LastMacroEvent}). We  prove that there exists a minimal violation trace $\tau'$ obtained from $\tau$ that satisfies ($1$) or ($2$). Similar to Lemma~\ref{lemma:CcvMinForm}, we get that there must exist $y$ s.t. $((\apr',\atr'), \storeact(\apr', \atr))\in \cfo(y)\cup\sto(y)$. 

%\noindent
%\textbf{Base case:} $\length{\beta} =0$. This implies that $(\issueact(\apr,\atr),(\apr',\atr'))\in \cfo$. 

%We will consider two cases: i) $((\apr',\atr'), \storeact(\apr', \atr))\in \sto(y)$, and ii)  $((\apr',\atr'), \storeact(\apr', \atr))\in \cfo(y)$.  Assume that $((\apr',\atr'), \storeact(\apr', \atr))\in \sto(y)$, then by reordering the store event $\storeact(\apr,\atr)\in \stmacro$ to occur just after the corresponding issue event we get $\tau' = \atmacro  \cdot \issueact(\apr,\atr) \cdot \storeact(\apr,\atr) \cdot (\apr',\atr')\cdot \storeact(\apr',\atr)\cdot\stmacro'$ is also a minimal violation where  $(\storeact(\apr,\atr),(\apr',\atr'))\in \sto(y)$ and $((\apr',\atr'), \storeact(\apr', \atr))\in \sto(y)$. $\tau'$ satisfies the first case of the lemma. Now assume that $((\apr',\atr'), \storeact(\apr', \atr))\in \cfo(y)$. Let $x$ s.t. $(\issueact(\apr,\atr),(\apr',\atr'))\in \cfo(x)$. If $x=y$ then both $\atr$ and $\atr'$ write to $x$. Thus, we get  $((\apr',\atr'), \storeact(\apr', \atr))\in \sto(x)$ as in the fist case which leads to the first case of the lemma. If $x\neq y$ then we get the second case of the lemma. 

%\noindent
%\textbf{Induction step:} We assume the induction hypothesis holds for $\length{\beta} \leq m$.   We now consider the same two cases as in the previous paragraph. 
We consider two cases:  i) $((\apr',\atr'), \storeact(\apr', \atr))\in \sto(y)$, and ii)  $((\apr',\atr'), \storeact(\apr', \atr))\in \cfo(y)$.  If $((\apr',\atr'), \storeact(\apr', \atr))\in \sto(y)$, then by reordering the store event $\storeact(\apr,\atr)\in \stmacro$ to occur just after the corresponding issue and removing all events in $\beta$ (and all related stores in $\stmacro$) that are not causally ordered before $(\apr',\atr')$ (since they do not contribute to the transactional happens-before cycle) we obtain a trace $\tau' = \atmacro  \cdot \issueact(\apr,\atr) \cdot \storeact(\apr,\atr) \cdot \beta' \cdot (\apr',\atr')\cdot \storeact(\apr',\atr)\cdot\stmacro'$ that is also a minimal violation and where $(\storeact(\apr,\atr),(\apr',\atr'))\in \sto(y)$ and $((\apr',\atr'), \storeact(\apr', \atr))\in \sto(y)$. The trace $\tau'$ satisfies the first case of the lemma. 

Now assume that $((\apr',\atr'), \storeact(\apr', \atr))\in \cfo(y)$, and let $\sigma=\{c \in \beta\ |\ (\issueact(\apr,\atr),c) \in \viso\}$. We consider the following three cases.

First, assume that $\sigma$ is empty. As in the proof of Lemma~\ref{lemma:CcvMinForm}, we obtain that there exist $a\in\beta\cdot (\apr',\atr')$ and $x$ s.t. $(\issueact(\apr,\atr),a) \in \cfo(x)$ and $(a, (\apr',\atr')) \in \hbo?$. If $x=y$ then both $\atr$ and the transaction $\atr_2$ by a process $\apr_2$ of the event $a$ write to $x$. Similar to before we can reorder the store event $\storeact(\apr,\atr)\in \stmacro$ to occur just after the corresponding issue and remove all issue events in $\beta\cdot (\apr',\atr')$ that occur after the issue event of $\atr_2$ and all their related stores. Also, we remove all events in $\beta$ that are not causally ordered before the issue event of $\atr_2$. We obtain $\tau' = \atmacro  \cdot \issueact(\apr,\atr) \cdot \storeact(\apr,\atr) \cdot \beta' (\apr_2,\atr_2)\cdot \storeact(\apr_2,\atr)\cdot\stmacro'$. In $\tau'$ the events of $\atr_2$ are assembled together, $\storeact(\apr_2,\atr)\in \stmacro$ is reordered to occur just after $(\apr_2,\atr_2)$, and $(\storeact(\apr,\atr),(\apr_2,\atr_2))\in \sto(y)$ and $((\apr_2,\atr_2), \storeact(\apr_2, \atr))\in \sto(y)$. Thus, $\tau'$ is a minimal violation and it satisfies the first case of the lemma. If $x\neq y$ then we get the second case of the lemma. 

Second, assume that $\sigma$ is not empty and all the elements of $\sigma$ are store events. As in the proof of Lemma~\ref{lemma:CcvMinForm}, we obtain that there exist $x$ and an event $a \in  \beta \cdot (\apr',\atr')$ that is not a store event of $\atr$ s.t. $(\issueact(\apr,\atr),a) \in (\sametro;\sto(x))\cup\cfo(x)$ and $(a, (\apr',\atr')) \in \hbo?$. If $(\issueact(\apr,\atr),a) \in (\sametro;\sto(x))$ or $x=y$ then both $\atr$ and the transaction $\atr_2$ by a process $\apr_2$ of the event $a$ write to $x$. Using the same procedure as in the previous paragraph we can obtain $\tau' = \atmacro  \cdot \issueact(\apr,\atr) \cdot \storeact(\apr,\atr) \cdot \beta' \cdot (\apr_2,\atr_2)\cdot \storeact(\apr_2,\atr)\cdot\stmacro'$ that satisfies the first case of the lemma. Similarly, if $(\issueact(\apr,\atr),a) \in \cfo(x)$ and $x \neq y$ then we get the second case of the lemma. 

Third, assume that $\sigma$ is not empty and $\issueact(\apr_{1},\atr_1)$ is the last issue event in $\sigma$, i.e., $\beta = \beta_{1}  \cdot \issueact(\apr_{1},\atr_1) \cdot  \beta_{2} \cdot (\apr',\atr')$. As in the proof of Lemma~\ref{lemma:CcvMinForm}, we obtain that there exist $x$ and an event $a \in  \beta_2 \cdot (\apr',\atr')$ that is not a store event of $\atr_1$ s.t. $(\issueact(\apr_1,\atr_1),a) \in (\sametro;\sto(x))\cup\cfo(x)$ and $(a, (\apr',\atr')) \in \hbo?$. 
If $x=y$ then both $\atr$ and the transaction $\atr_2$ by a process $\apr_2$ of the event $a$ write to $x$. Using the same procedure as before we can obtain a trace $\tau' = \atmacro  \cdot \issueact(\apr,\atr) \cdot \storeact(\apr,\atr) \cdot \beta_1' \cdot \beta_2' \cdot (\apr_2,\atr_2)\cdot \storeact(\apr_2,\atr)\cdot\stmacro'$ that is a minimal violation. $\tau'$ has less delays than $\tau$ since the store of $\atr$ was not delayed after $\issueact(\apr_{1},\atr_1)$. This contradicts the fact that 
$\tau$ is a minimal violation. Assume now that $x\neq y$. We assume w.l.o.g. that all events in $\beta_{2}$ do not read values that any transaction with an issue event in $\issueact(\apr,\atr) \cdot beta_{1}  \cdot \issueact(\apr_{1},\atr_1)$ overwrites. 
If $(\issueact(\apr_1,\atr_1),a) \in (\sametro;\sto(x))$ and $a \neq (\apr',\atr')$ then we can remove all issue events in $\beta_2 \cdot (\apr',\atr')$ that occur after the issue event of $\atr_2$ including $(\apr',\atr')$ and assemble together the events of $\atr_2$. We obtain that $(\storeact(\apr_1,\atr_1),(\apr_2,\atr_2)) \in \sto(x)$ and $((\apr_2,\atr_2),\storeact(\apr_2,\atr_1)) \in \sto(x)$ where we do not need to delay the transaction $\atr$ and obtain $\tau' = \atmacro  \cdot (\apr,\atr) \cdot \beta_1' \cdot \issueact(\apr_{1},\atr_1)\cdot\storeact(\apr_1,\atr_1) \cdot \beta_2' \cdot (\apr_2,\atr_2)\cdot \storeact(\apr_2,\atr_1)\cdot\stmacro'$ that is a violation and has less delays than $\tau$. This contradicts the fact that $\tau$ is a minimal violation. If $(\issueact(\apr_1,\atr_1),a) \in (\sametro;\sto(x))$ and $a = (\apr',\atr')$ (i.e., $\atr' = \atr_2$) then we construct $\tau'$ such that all transactions that have issue events in $\sigma$ and $\atr$ are executed atomically after all the events in $(\beta_1 \setminus \sigma)\cdot \beta_2 \cdot \issueact(\apr',\atr') \cdot \storeact(\apr',\atr')$ are executed first, i.e., $\tau' = \atmacro  \cdot \beta_{11} \cdot \beta_2 \cdot \issueact(\apr',\atr') \cdot \storeact(\apr',\atr') \cdot  (\apr,\atr) \cdot \beta_{12} \cdot  \beta' \cdot (\apr_1,\atr_1)\cdot \storeact(\apr_1,\atr')\cdot\stmacro'$. $\tau'$ is a robustness violation since $(\storeact(\apr',\atr'),(\apr_1,\atr_1)) \in \sto(x)$ and $((\apr_1,\atr_1),\storeact(\apr_1,\atr')) \in \sto(x)$. Also, $\tau'$ has less delays than $\tau$ since $\atr'$ was not delayed after a causally dependent event other than its store events and $\atr$ is no longer delayed after the issue event of $\atr_1$. This contradicts the fact that $\tau$ is a minimal violation. Finally, the only remaining possibility is $(\issueact(\apr_1,\atr_1),a) \in \cfo(x)$ where $x\neq y$ which corresponds to the second case of the lemma.
\end{proof}

We use $\tracecmone$ and $\tracecmtwo$ to denote the class of minimal violations that satisfy the first and second case in Lemma~\ref{lemma:CmMinForm}, respectively. To show that for a non robust program, we can always find a minimal violation in either 
$\tracecmone$ or $\tracecmtwo$ where $\beta$ and $\beta_2$ do not contain delayed transactions we can use the same proof arguments as in Lemma~\ref{lemma:CcvMinViol}. For minimal violations in $\tracecmtwo$ where $\atr$ and $\atr_1$ are distinct transactions, the two properties that issue events of all delayed transactions form a causality chain and that delayed transactions in $\issueact(\apr,\atr) \cdot \beta_1$ do not access the shared variable $\anaddr$ can also be proved in the same manner as in Lemmas~\ref{lemma:CCvcausalitychain} and~\ref{lemma:NoAccessX}, respectively.

\section{Robustness Violations Under Weak Causal Consistency}\label{sec:MVWCC}

If a program is robust against \scct{}, then it must not contain a write-write race under \scct{} (note that this is not true for \ccvt{}). Therefore, by Theorem~\ref{theorem:traceswwracesOrig}, a program which is robust against \scct{} has the same set of traces under both \scct{} and \wcct{}, which implies that it is also robust against \wcct{}. Conversely, since \wcct{} is weaker than \scct{} (i.e., $\tracesconf_{\scct{}}(\aprog)\subseteq \tracesconf_{\wcct{}}(\aprog)$ for any $\aprog$), if a program is robust against \wcct{} then it is robust against \scct{}. Thus, we obtain the following result.

\begin{thm}\label{theorem:CCMinViol}
A program $\aprog$ is robust against \wcct{} iff it is robust against \scct{}.
\end{thm}

%has execution traces that are in $\tracesconf_{\wcct{}}(\aprog)$ and not in $\tracesconf_{\scct{}}(\aprog)$ then these executions must contain a write-write race. Thus, checking robustness against \wcct{} can be decomposed to detecting write-write races and checking robustness against \scct{}.checking write-write races under \wcct{} is the same as checking write-write races under \scct{}. Then, we can reuse the first pattern of violation under \scct{} given in Theorem \ref{theorem:CmMinViol}. In particular, we note that the execution prefix $\alpha$ of $\tau_{\scct{}1}$ is possible under both semantics since there is no write-write race which occurs before the one detected in $\tau_{\scct{}1}$. Otherwise, we can shortcut $\tau_{\scct{}1}$, and consider only the prefix $\alpha$ where the first write-write race occurs and obtain a new \scct{} violation. If there are no execution traces of the form of $\tau_{\scct{}1}$, then the underlying program has no write-write races under \wcct{}. Thus, the program has the same set of traces under the three models of causal consistency. Then, checking the robustness against \wcct{} can be reduced to checking for the violation $\tau_{\scct{}2}$ (because we already checked for the first pattern while we were searching for write-write race). 
%!TEX root = draft.tex
\section{Reduction to SC Reachability}\label{sec:Instr}

We describe a reduction of robustness checking to a reachability problem in a program executing under the serializability semantics, which can be simulated on top of standard sequential consistency (SC) by considering that each transaction is an atomic section (guarded by a fixed global lock).
Essentially, given a program $\aprog$ and a semantics $\textsf{X} \in \{\ccvt{},\ \scct{},\ \wcct{}\}$, we define an instrumentation of $\aprog$ such that $\aprog$ is not robust against $\textsf{X}$ iff the instrumentation reaches an error state under the serializability semantics. The instrumentation uses auxiliary variables in order to simulate the robustness violations (in particular, the delayed transactions) satisfying the patterns given in Fig.~\ref{fig:ccvrobtraces} and Fig.~\ref{fig:cmrobtraces}.
We will focus our presentation on the second violation pattern of \ccvt{} (which is similar to the second violation pattern of \scct{}):
$\tau_{\ccvt{}2} =\atmacro\cdot \issueact(\apr,\atr)\cdot \beta_1\cdot \issueact(\apr_{1},\atr_1) \cdot\beta_2\cdot (\apr',\atr')\cdot \storeact(\apr',\atr) \cdot \stmacro$. %For lack of space, we describe the instrumentation only informally. The precise definition is given in the Appendix.

The process $\apr$ that delayed the first transaction $\atr$ is called the \emph{Attacker}. The other processes delaying transactions in $\beta_1\cdot \issueact(\apr_{1},\atr_{1})$ are called \emph{Visibility Helpers}. Recall that all the delayed transactions must be causally ordered after $\issueact(\apr,\atr)$. The processes that execute transactions in $\beta_2\cdot (\apr',\atr')$ and contribute to the happens-before path between $\issueact(\apr_{1},\atr_{1})$ and $\storeact(\apr',\atr)$ are called \emph{Happens-Before Helpers}. A happens-before helper cannot be the attacker or a visibility helper since this would contradict the causal delivery guarantee provided by causal consistency (a transaction of a happens-before helper is not delayed, so visible immediately to all processes, and it cannot follow a delayed transaction). $\stmacro$ contains the stores of the delayed transactions in $\issueact(\apr,\atr)\cdot\beta_1\cdot \issueact(\apr_{1},\atr_{1})$.
It is important to notice that we may have $\atr=\atr_1$. In this case, $\beta_1 =\epsilon$ and the only delayed transaction is $\atr$. Also, all delayed transactions in $\beta_1$ including $\atr_1$ may be issued by the same process as $\atr$. In all of these cases, the set of Visibility Helpers is empty.

The instrumentation uses two copies of the set of shared variables in the original program. We use primed variables $\anaddr'$ to denote the second copy. When a process becomes the attacker or a visibility helper, it will write only to the second copy that is visible only to these processes (and remains invisible to the other processes including the happens-before helpers). The writes made by other processes including the happens-before helpers are made visible to all processes, i.e., they are applied on both copies of every shared variable.

To establish the causality chains of the delayed transactions issued by the attacker and the visibility helpers, we look whether a transaction can extend the causality chain started by the first delayed transaction issued by the attacker. This is to ensure that all such transactions are causally related to the first delayed transaction (of the attacker). In order for a transaction to ``join'' the causality chain, it has to satisfy one of the following conditions:
\begin{itemize}
\item the transaction is issued by a process that has already another transaction in the causality chain. Thus, we ensure the continuity of the causality chain through program order;
\item the transaction is reading from a variable that was updated by a previous transaction in the causality chain. Hence, we ensure the continuity of the  causality chain through the write-read relation.
\end{itemize}
We introduce a flag for each shared variable to mark the fact that it was updated by a previous transaction in the causality chain. These flags are used by the instrumentation to establish whether a transaction ``joins'' a causality chain.
Enforcing a happens-before path starting in the last delayed transaction, using transactions of the happens-before helpers, can be done in the same way. Compared to causality chains, there are two more cases in which a transaction can extend a happens-before path:
\begin{itemize}
\item the transaction writes to a shared variable that was read by a previous transaction in the happens-before path. Hence, we ensure the continuity of the happens-before path through the read-write relation;
\item the transaction writes to a shared variable that was updated by a previous transaction in the happens-before path. Hence, we ensure the continuity of the happens-before path  through write-write order.
\end{itemize}
Thus, we extend the shared variables flags used for causality chains in order to record if a variable was read or written by a previous transaction (in this case, a previous transaction in the happens-before path). %Recall that $\beta_2$ contains no store events corresponding to delayed transactions.
Overall, the instrumentation uses a flag $\anaddr.event$ or $\anaddr'.event$ for each (copy of a) shared variable,
that stores the type of the last access (read or write) to the variable. %We will explain the meaning of these flags along with the instrumentation.
Initially, these flags and other flags used by the instrumentation as explained below are initialized to null ($\perp$).

In general, whether a process is an attacker, visibility helper, or happens-before helper is not enforced syntactically by the instrumentation, and can vary from execution to execution. The role of a process in an execution is set \emph{non-deterministically} during the execution using some additional process-local flags. Thus, during an execution, each process chooses to set to $\mytrue$ at most one of the flags $p.a$,  $p.vh$, and $p.hbh$, implying that the process becomes an attacker, visibility helper, or happens-before helper, respectively. At most one process can be an attacker, i.e., set $p.a$ to $\mytrue$.
%Before a process becomes an attacker or visibility helper it passes though a stage where it executes transactions in the usual way without delaying them.

\begin{figure}[t]
 \scriptsize
 \hspace*{-0.7cm}
 \begin{minipage}[l]{0.45\textwidth}
\begin{eqnarray}
%%%%%%%%%%%
%%%%%%%%%%%
&&\semidler{\thetransition{\alab_1}{\alab_2}{\beginact}} = \notag\\
&&\textbf{// Typical execution of begin}\notag\\
&&\thetransitionInstrumented{\alab_1}{\alab_{x1}}{\thecondition{\hb{} = \perp \land (p.a \neq \perp \lor  a_{\attacktransaction} = \perp)}}\ \ \ \ \ \  \label{Equation:AttackerBegin}\\
&&\thetransitionInstrumented{\alab_{x1}}{\alab_2}{\beginact}\notag\\
&&\textbf{// Begin of first delayed transaction}\notag\\
&&\thetransitionInstrumented{\alab_1}{\alab_{x2}}{\thecondition{\hb{} = \perp \land \  a_{\attacktransaction} = \perp}}\label{Equation:AttackerBeginDelay}\\
&&\thetransitionInstrumented{\alab_{x2}}{\alab_{x3}}{\beginact}\notag\\
&&\thetransitionInstrumented{\alab_{x3}}{\alab_{x4}}{\theassign{p.a}{\ajoin}}\notag\\
&&\thetransitionInstrumented{\alab_{x4}}{\alab_{x5}}{\theforeachassign{\anaddr}{\mathbb{V}}{\anaddr'}{\anaddr}}\notag\\
&&\thetransitionInstrumented{\alab_{x5}}{\alab_2}{\theassign{a_{\attacktransaction}}{\mytrue}}\label{Equation:AttackerTransaction}\\[1.5mm]
%%%%%%%%%%%
&&\semidler{\thetransition{\alab_1}{\alab_2}{\theload{\areg}{\anaddr}}}  =\notag\\
&&\textbf{// Read before delaying transactions}\notag\\
&&\thetransitionInstrumented{\alab_1}{\alab_{x1}}{\thecondition{ a_{\attacktransaction} = \perp}}\notag\\
&&\thetransitionInstrumented{\alab_{x1}}{\alab_{2}}{\theload{\areg}{\anaddr}}\notag\\
&&\textbf{// Read in delayed transactions}\notag\\
&&\thetransitionInstrumented{\alab_1}{\alab_{x2}}{\thecondition{ a_{\attacktransaction} \neq \perp \land \apr.a\neq \perp}}\notag\\
&&\thetransitionInstrumented{\alab_{x2}}{\alab_{x3}}{\theload{\areg}{\anaddr'}}\notag\\
&&\thetransitionInstrumented{\alab_{x3}}{\alab_{2}}{\theassign{\anaddr'.event}{\loadacc}}\notag\\
%%%%%%%%%%
&&\textbf{// Special read in last delayed transaction}\notag\\
&&\thetransitionInstrumented{\alab_{1}}{\alab_{x4}}{\thecondition{\anaddr'.event = \perp \land \ \apr.a\neq \perp }}\notag\\
&&\thetransitionInstrumented{\alab_{x4}}{\alab_{x5}}{\theload{\areg}{\anaddr'}}\notag\\
&&\thetransitionInstrumented{\alab_{x5}}{\alab_{x6}}{\thestore{\hb{}}{\mytrue}}\label{Equation:AttackerLastLoad}\\
&&\thetransitionInstrumented{\alab_{x6}}{\alab_{2}}{\theassign{\anaddr.event}{\loadacc}}\label{Equation:AttackerLastLoad1}
%%%%%%%%%%%
\end{eqnarray}
 \end{minipage}%\hspace{5mm}
 \hfill
 \begin{minipage}[r]{0.55\textwidth}
 \begin{eqnarray}
 %%%%%%%%%%%
&&\semidler{\thetransition{\alab_1}{\alab_2}{\thestore{\anaddr}{e}}} = \notag\\
&&\textbf{// Write before delaying transactions}\notag\\
&&\thetransitionInstrumented{\alab_1}{\alab_{x1}}{\thecondition{ a_{\attacktransaction} = \perp}}\notag\\
&&\thetransitionInstrumented{\alab_{x1}}{\alab_{2}}{\thestore{\anaddr}{e}}\notag\\
&&\textbf{// Write in delayed transactions}\notag\\
&&\thetransitionInstrumented{\alab_1}{\alab_{x2}}{\thecondition{ a_{\attacktransaction} \neq \perp \land \apr.a\neq \perp}}\notag\\
&&\thetransitionInstrumented{\alab_{x2}}{\alab_{x3}}{\thestore{\anaddr'}{e}}\label{Equation:AttackerDelayStore}\\
&&\thetransitionInstrumented{\alab_{x3}}{\alab_{2}}{\theassign{\anaddr'.event}{\storeacc}}\label{Equation:AttackerDelayStore0}\\
%%%%%%%%%%%
&&\textbf{// Special write in first delayed transaction}\notag\\
&&\thetransitionInstrumented{\alab_1}{\alab_{x4}}{\thecondition{a_{\attackstore} = \anaddr.event = \perp \land \ \apr.a\neq \perp}}\notag\\
&&\thetransitionInstrumented{\alab_{x4}}{\alab_{x5}}{\thestore{\anaddr'}{e}}\notag\\
&&\thetransitionInstrumented{\alab_{x5}}{\alab_{x6}}{\thestore{a_{\attackstore}}{`\anaddr`}}\label{Equation:AttackerStore}\\
&&\thetransitionInstrumented{\alab_{x6}}{\alab_{2}}{\theassign{\anaddr'.event}{\storeacc}}\notag\\
%%%%%%%%%%%
&&\textbf{// Special write in last delayed transaction}\notag\\
&&\thetransitionInstrumented{\alab_{1}}{\alab_{x7}}{\thecondition{\anaddr'.event = \perp \land\ \apr.a\neq \perp}}\notag\\
&&\thetransitionInstrumented{\alab_{x7}}{\alab_{x8}}{\thestore{\anaddr'}{e}}\notag\\
&&\thetransitionInstrumented{\alab_{x8}}{\alab_{x9}}{\thestore{\hb{}}{\mytrue}}\label{Equation:AttackerLastStore}\\
&&\thetransitionInstrumented{\alab_{x9}}{\alab_{2}}{\theassign{\anaddr.event}{\loadacc}}\label{Equation:AttackerLastStore1}\\[1.5mm]
%%%%%%%%%%%
&&\semidler{\thetransition{\alab_1}{\alab_2}{\commitact}}=\notag\\
&&\thetransitionInstrumentedb{\alab_1}{\thecondition{ \apr.a\neq \perp \land \ a_{\attackstore} = \perp}}\notag\\
&&\thetransitionInstrumented{\alab_{1}}{\alab_{2}}{\commitact}\notag
%%%%%%%%%%%
%%%%%%%%%%%
%%%%%%%%%%%
\end{eqnarray}
 \end{minipage}
\normalsize
\caption{Instrumentation of the Attacker. We use $`\anaddr`$ to denote the name of the shared variable $\anaddr$.}
\label{Figure:Attacker}
\end{figure}

\subsection{Instrumentation of the Attacker}\label{subsec:Attacker}

We provide in Fig. \ref{Figure:Attacker}, the instrumentation of the instructions for the attacker process. Such a process passes through an initial phase where it executes transactions that are visible immediately to all the other processes (i.e., they are not delayed), and then non-deterministically it can choose to delay a transaction. When the attacker randomly chooses the first transaction to start  delaying of transactions, it sets a \emph{global} flag $a_{\attacktransaction}$ to $\mytrue$ in the instruction $\beginact$ (line~\eqref{Equation:AttackerTransaction}). Then, it sets the flag $p.a$ to $\ajoin$ to indicate that the current process is the attacker. During the first delayed transaction, the attacker non-deterministically chooses a write instruction to a shared variable $y$ and stores the name of this variable in the flag  $a_{\attackstore}$ (line \eqref{Equation:AttackerStore}). The values written during delayed transactions are stored in the primed variables and are visible only to the attacker and the visibility helpers. For example, given a variable $z$, all the writes to $z$ from the original program are transformed into writes to the primed version $z'$ (line \eqref{Equation:AttackerDelayStore}). Each time the attacker writes to a variable $z'$, it sets the flag $z'.event$ to $\storeacc$ (line \eqref{Equation:AttackerDelayStore0}) which will allow other processes that read the same variable to join the set of visibility helpers and start delaying their transactions. Once the attacker delays a transaction, it will read only from the primed variables (i.e., $z'$). 

To start the happens-before path, the attacker has to execute a transaction that either reads or writes to a shared variable $x$ that was not accessed by a delayed transaction (i.e., $x'.event = \perp$). In this case, it sets the variable $\hb{}$ to $\mytrue$ (lines \eqref{Equation:AttackerLastLoad} and \eqref{Equation:AttackerLastStore}) to mark the start of the happens before path and the end of the visibility chains, and it sets the flag $x.event$ to $\loadacc$ (lines \eqref{Equation:AttackerLastLoad1} and \eqref{Equation:AttackerLastStore1}). We set $x.event$ to $\loadacc$ even in the case of a write to $x$ in order to simplify the instrumentation of the happens-before helpers (to check that this transaction is related to a transaction of a happens-before helper $p$ through $\sto(x)$ or $\cfo(x)$ it is enough that $p$ writes to $x$ and it ``observers'' the same value $\loadacc$ in $x.event$).
When the flag $\hb{}$ is set to $\mytrue$ the attacker stops executing new transactions.
We can notice that when the $\hb{}$ is set to $\mytrue$, we can no longer execute new transactions from the attacker (all conditions in lines~\eqref{Equation:AttackerBegin} and \eqref{Equation:AttackerBeginDelay} become $\myfalse$).

\begin{figure}[t]
 \scriptsize
 \hspace*{-0.7cm}
 \begin{minipage}[l]{0.435\textwidth}
\begin{eqnarray}
%%%%%%%%%%%
&&\semVhelperorig{\thetransition{\alab_1}{\alab_2}{\beginact}}=\notag\\
&&\textbf{// Before joining visibility helpers}\notag\\
&&\thetransitionInstrumented{\alab_1}{\alab_{x1}}{\thecondition{\hb{} = \perp \land \  (a_{\attacktransaction} = \perp \lor p.vh = \perp)}}\notag\\
&&\thetransitionInstrumented{\alab_{x1}}{\alab_2}{\beginact}\label{Equation:VHBegin}\\
&&\textbf{// Joining visibility helpers}\notag\\
&&\thetransitionInstrumented{\alab_1}{\alab_{x2}}{\thecondition{\hb{} =  p.vh =  p.a = \perp \land \  a_{\attacktransaction} \neq \perp}}\notag\\
&&\thetransitionInstrumented{\alab_{x2}}{\alab_{x3}}{\beginact}\notag\\
&&\thetransitionInstrumented{\alab_{x3}}{\alab_{x4}}{\theassign{p.vh}{\vhatt}}\notag\\
&&\thetransitionInstrumented{\alab_{x4}}{\alab_2}{\theforeachassign{\anaddr'}{\mathbb{V}}{\anaddr'.event'}{\anaddr'.event}}\notag\\
&&\textbf{// After joining visibility helpers}\notag\\
&&\thetransitionInstrumented{\alab_1}{\alab_{x5}}{\thecondition{\hb{} = \perp \land \  a_{\attacktransaction} \neq \perp \land \  p.vh}}\notag\\
&&\thetransitionInstrumented{\alab_{x5}}{\alab_2}{\beginact}\label{Equation:VHTransaction}\\[1.5mm]
%%%%%%%%%%%
%%%%%%%%%%%
&&\semVhelperorig{\thetransition{\alab_1}{\alab_2}{\theload{\areg}{\anaddr}}}=\notag\\
&&\textbf{// Before joining visibility helpers}\notag\\
&&\thetransitionInstrumented{\alab_{1}}{\alab_{x1}}{\thecondition{a_{\attacktransaction} = \perp \lor ( p.vh =  \apr.a = \perp)}}\notag\\
&&\thetransitionInstrumented{\alab_{x1}}{\alab_{2}}{\theload{\areg}{\anaddr}}\label{Equation:VHOriginLoad}\\
&&\textbf{// After joining visibility helpers}\notag\\
&&\thetransitionInstrumented{\alab_1}{\alab_{x2}}{\thecondition{ p.vh \neq \perp}}\notag\\
&&\thetransitionInstrumented{\alab_{x2}}{\alab_{x3}}{\theload{\areg}{\anaddr'}}\label{Equation:VHDelayLoad}\\
&&\thetransitionInstrumented{\alab_{x3}}{\alab_{x4}}{\thecondition{\anaddr'.event' = \storeacc \land \  \neg p.vh}}\notag\\
&&\thetransitionInstrumented{\alab_{x4}}{\alab_{2}}{\theassign{p.vh}{\vhjoin}}\label{Equation:VHMeetCondition}\\
&&\thetransitionInstrumented{\alab_{x3}}{\alab_{2}}{\thecondition{\anaddr'.event' \neq \storeacc \lor \  p.vh}}\notag\\
%%%%%%%%%%%
&&\textbf{// Last delayed transaction}\notag\\
&&\thetransitionInstrumented{\alab_{1}}{\alab_{x5}}{\thecondition{\anaddr'.event = \perp \land \  p.vh \neq \perp }}\notag\\
&&\thetransitionInstrumented{\alab_{x5}}{\alab_{x6}}{\thestore{\hb{}}{\mytrue}}\label{Equation:VHLastLoad}\\
&&\thetransitionInstrumented{\alab_{x6}}{\alab_{x7}}{\theassign{\anaddr.event}{\loadacc}}\label{Equation:VHLastLoad1}\\
&&\thetransitionInstrumented{\alab_{x7}}{\alab_{2}}{\theload{\areg}{\anaddr'}}\notag
%%%%%%%%%%%
\end{eqnarray}
 \end{minipage}
 \hfill 
 \begin{minipage}[r]{0.545\textwidth}
    \begin{eqnarray}
    %%%%%%%%%%%
    &&\semVhelperorig{\thetransition{\alab_1}{\alab_2}{\thestore{\anaddr}{\anexpr}}}=\notag\\
    &&\textbf{// Before attacker delays transactions}\notag\\
    &&\thetransitionInstrumented{\alab_{1}}{\alab_{x1}}{\thecondition{a_{\attacktransaction} = \perp}}\notag\\
    &&\thetransitionInstrumented{\alab_{x1}}{\alab_{2}}{\thestore{\anaddr}{e}}\ \ \label{Equation:VHOriginStore}\\
    &&\textbf{// Before joining visibility helpers}\notag\\
    &&\thetransitionInstrumented{\alab_1}{\alab_{x2}}{\thecondition{ a_{\attacktransaction} \neq \perp \land p.vh = \apr.a = \perp}}\notag\\
    &&\thetransitionInstrumented{\alab_{x2}}{\alab_{x3}}{\thestore{\anaddr'}{e}}\notag\\
    &&\thetransitionInstrumented{\alab_{x3}}{\alab_{2}}{\thestore{\anaddr}{e}}\notag\\
    &&\textbf{// After joining visibility helpers}\notag\\
    &&\thetransitionInstrumented{\alab_1}{\alab_{x4}}{\thecondition{p.vh \neq \perp}}\notag\\
    &&\thetransitionInstrumented{\alab_{x4}}{\alab_{x5}}{\thestore{\anaddr'}{e}}\label{Equation:VHDelayStore}\\
    &&\thetransitionInstrumented{\alab_{x5}}{\alab_{x6}}{\theassign{\anaddr'.event}{\storeacc}}\label{Equation:VHDelayStore1}\\
    &&\thetransitionInstrumented{\alab_{x6}}{\alab_{2}}{\theassign{\anaddr'.event'}{\perp}}\notag\\
    %%%%%%%%%%%
    &&\textbf{// Last delayed transaction}\notag\\
    &&\thetransitionInstrumented{\alab_{1}}{\alab_{x7}}{\thecondition{\anaddr'.event = \perp \land \  p.vh \neq \perp }}\notag\\
    &&\thetransitionInstrumented{\alab_{x7}}{\alab_{x8}}{\thestore{\hb{}}{\mytrue}}\label{Equation:VHLastStore}\\
    &&\thetransitionInstrumented{\alab_{x8}}{\alab_{x9}}{\theassign{\anaddr.event}{\loadacc}}\label{Equation:VHLastStore1}\\
    &&\thetransitionInstrumented{\alab_{x9}}{\alab_{2}}{\thestore{\anaddr'}{e}}\notag\\[1.5mm]
    %%%%%%%%%%%
    &&\semVhelperorig{\thetransition{\alab_1}{\alab_2}{\commitact}}=\notag\\
    &&\textbf{// Before joining visibility helpers}\notag\\
    &&\thetransitionInstrumented{\alab_1}{\alab_{x1}}{\thecondition{ a_{\attacktransaction} = \perp \lor (a_{\attacktransaction} \neq \perp \land \  p.vh = \perp)}}\notag\\
    &&\thetransitionInstrumented{\alab_{x1}}{\alab_2}{\commitact}\notag\\
    &&\textbf{// After joining visibility helpers}\notag\\
    &&\thetransitionInstrumented{\alab_1}{\alab_{x2}}{\thecondition{ a_{\attacktransaction} \neq \perp \land \  p.vh}}\notag\\
    &&\thetransitionInstrumented{\alab_{x2}}{\alab_2}{\commitact}\notag\\
    &&\textbf{// Failed to join visibility helpers}\notag\\
    &&\thetransitionInstrumentedb{\alab_1}{\thecondition{ a_{\attacktransaction} \neq \perp \land \  \neg p.vh}}\label{Equation:VHBLOCK}
    %%%%%%%%%%%
    %%%%%%%%%%%
    \end{eqnarray}
    \end{minipage}%\hfill
\normalsize
\caption{Instrumentation of the Visibility Helpers.}
\label{Figure:VHelpersInstrumentation}
\end{figure}

\subsection{Instrumentation of the Visibility Helpers}\label{subsec:VHelpers}

Fig. \ref{Figure:VHelpersInstrumentation} lists the instrumentation of the instructions of a process that belongs to the set of visibility helpers. Such a process passes through an initial phase where it executes the original code instructions (lines \eqref{Equation:VHOriginStore} and \eqref{Equation:VHOriginLoad}) until the flag $a_{\attacktransaction}$ is set to $\mytrue$ by the attacker. Then, it continues the execution of its original instructions but, whenever it stores a value it writes it to both the shared variable $z$ and the primed variable $z'$ so it is visible to all processes. Non deterministically it chooses a first transaction to delay, at which point it joins the set of visibility helpers. It sets the flag $p.vh$ to $\vhatt$ signaling its desire to join the visibility helpers, and it chooses a transaction (the $\beginact$ of this transaction is shown in line~\eqref{Equation:VHTransaction}) through which the process will join the set of visibility helpers. The process directly starts delaying its writes, i.e., writing to primed variables, and reading only from delayed writes, i.e., from primed variables, and behaving the same as the attacker. In order to check that it can extend the sequence of causal dependencies (required by the causal chain definition), it takes a snapshot of the $\_.event$ fields at the beginning of the transaction and stores it to $\_.event'$ fields (line $\alab_{x4}$ in the instrumentation of $\beginact$). This snapshot is necessary to check that it reads from writes made in other transactions (ignoring the writes in the current transaction). When a process choses a first transaction to delay (during the \plog{begin} instruction), it has made a pledge that during this transaction it will read from a variable that was updated by a another delayed transaction from either the attacker or some other visibility helper. This is to ensure that this transaction extends the visibility chain. Hence, the local process flag $p.vh$ will be set to $\vhjoin$ when the process meets its pledge (line \eqref{Equation:VHMeetCondition}). If the process does not keep its pledge (i.e., $p.vh$ is equal to $\vhatt$) at the end of the transaction (i.e., during the \plog{end} instruction) we block the execution. Thus, when executing the $\commitact$ instruction of the underlying transaction we check whether the flag $p.vh$ is null, if so we block the execution (line~\eqref{Equation:VHBLOCK}).

When a process joins the visibility helpers, it delays all writes and reads only from the primed variables (lines \eqref{Equation:VHDelayStore} and \eqref{Equation:VHDelayLoad}).
Similar to the attacker, a process in the visibility helpers delays a write to a shared variable $z$ by writing to $z'$, it sets the flag $z'.event$ to $\storeacc$ (line \eqref{Equation:VHDelayStore1}).
In order for a process in the visibility helpers to start the happens-before path, it has to either read or write a shared variable $x$ that was not accessed by a delayed transaction (i.e., $x'.event = \perp$). In this case we set the flag $\hb{}$ to $\mytrue$ (lines \eqref{Equation:VHLastStore} and \eqref{Equation:VHLastLoad}) to mark the start of the happens before path and the end of the visibility chains and set the flag $x.event$ to $\loadacc$ (lines \eqref{Equation:VHLastStore1} and \eqref{Equation:VHLastLoad1}). 
When the flag $\hb{}$ is set to $\mytrue$, all processes in the set of visibility helpers stop issuing new transactions because all conditions for executing the $\beginact$ instruction become $\myfalse$.

%Notice that in the underlying procedure for selecting a process to join the set of visibility helpers, it suffices that the first transaction executed by the process to be causally dependent on some transaction $\atr$ which was executed by the existing visibility helpers or the attacker. The transaction $\atr$ does not need to be the last transaction executed by the existing visibility helpers or the attacker. Therefore, the instrumentation allows executions that contain several ``concurrent'' causality chains that start from the first delayed transaction (of the attacker), however, at least one of these chains must end with the last delayed transaction in the execution. This is a slight deviation from the characterization of minimal violations in Theorem~\ref{theorem:CcvMinViol} and Theorem~\ref{theorem:CmMinViol} where all the delayed transactions must be causally related to the last delayed transaction, which however has no impact on the correctness of the instrumentation (stated in Theorem~\ref{th:final}). The delayed transactions that are not causally related to the last delayed transaction could be safely removed from the execution in order to get a minimal violation.

\begin{figure}[t]
 \scriptsize
 \hspace*{-0.7cm}
\begin{minipage}[l]{0.44\textwidth}
\begin{eqnarray}
%%%%%%%%%%%
&&\semHbhelperorig{\thetransition{\alab_1}{\alab_2}{\beginact}}=\notag\\
&&\textbf{// Before joining happens-before helpers}\notag\\
&&\thetransitionInstrumented{\alab_1}{\alab_{x1}}{\thecondition{\hb{} = p.vh =  p.a = \perp}}\notag\\
&&\thetransitionInstrumented{\alab_{x1}}{\alab_2}{\beginact}\notag\\
&&\textbf{// Joining happens-before helpers}\notag\\
&&\thetransitionInstrumented{\alab_1}{\alab_{x2}}{\thecondition{\hb{} \neq \perp \land\   p.hbh =  p.vh =  p.a = \perp}}\notag\\
&&\thetransitionInstrumented{\alab_{x2}}{\alab_{x3}}{\beginact}\label{Equation:HBHTransaction}\\
&&\thetransitionInstrumented{\alab_{x3}}{\alab_2}{\theforeachassign{\anaddr}{\mathbb{V}}{\anaddr.event'}{\anaddr.event}}\notag\\
&&\textbf{// After joining happens-before helpers}\notag\\
&&\thetransitionInstrumented{\alab_1}{\alab_{x4}}{\thecondition{\hb{} \neq \perp \land \   p.hbh \neq \perp}}\notag\\
&&\thetransitionInstrumented{\alab_{x4}}{\alab_2}{\beginact}\notag\\[1.5mm]
%%%%%%%%%%%
%%%%%%%%%%%
&&\semHbhelperorig{\thetransition{\alab_1}{\alab_2}{\thestore{\anaddr}{\anexpr}}}=\notag\\[1mm]
&&\textbf{// Before the first delayed transaction}\notag\\
&&\thetransitionInstrumented{\alab_{1}}{\alab_{x1}}{\thecondition{\hb{} = \perp \land \  a_{\attacktransaction} = \perp}}\notag\\
&&\thetransitionInstrumented{\alab_{x1}}{\alab_{2}}{\thestore{\anaddr}{\anexpr}}\notag\\
&&\textbf{// After the first delayed transaction}\notag\\
&&\thetransitionInstrumented{\alab_{1}}{\alab_{x2}}{\thecondition{\hb{} = p.vh = p.a = \perp \land \  a_{\attacktransaction} \neq \perp}}\notag\\
&&\thetransitionInstrumented{\alab_{x2}}{\alab_{x3}}{\thestore{\anaddr'}{\anexpr}}\label{Equation:HBHStore1}\\
&&\thetransitionInstrumented{\alab_{x3}}{\alab_{2}}{\thestore{\anaddr}{\anexpr}}\label{Equation:HBHStore2}\\
&&\textbf{// After the last delayed transaction}\notag\\
&&\thetransitionInstrumented{\alab_1}{\alab_{x4}}{\thecondition{ \hb{} \neq \perp \land\ p.vh =  p.a = \perp}}\notag\\
&&\thetransitionInstrumented{\alab_{x4}}{\alab_{x5}}{\thestore{\anaddr}{\anexpr}}\notag\\
&&\thetransitionInstrumented{\alab_{x5}}{\alab_{x6}}{\theassign{\anaddr.event}{\storeacc}}\label{Equation:HBHFlagStore}\\
&&\thetransitionInstrumented{\alab_{x6}}{\alab_{x7}}{\thecondition{\anaddr.event' \neq \perp \land \  p.hbh = \perp}}\notag\\
&&\thetransitionInstrumented{\alab_{x7}}{\alab_{2}}{\theassign{p.hbh}{\hbhjoin}}\label{Equation:HBHMeetConditionStore}\\
&&\thetransitionInstrumented{\alab_{x6}}{\alab_{2}}{\thecondition{\anaddr.event' = \perp \lor \  p.hbh \neq \perp}}\notag
%%%%%%%%%%%
\end{eqnarray}
 \end{minipage}
 \hfill
\begin{minipage}[r]{0.555\textwidth}
\begin{eqnarray}
%%%%%%%%%%%
&&\semHbhelperorig{\thetransition{\alab_1}{\alab_2}{\theload{\areg}{\anaddr}}}=\notag\\[1mm]
&&\textbf{// Before the last delayed transaction}\notag\\
&&\thetransitionInstrumented{\alab_{1}}{\alab_{x1}}{\thecondition{\hb{} = \perp \land\ p.vh = p.a = \perp}}\notag\\
&&\thetransitionInstrumented{\alab_{x1}}{\alab_{2}}{\theload{\areg}{\anaddr}}\label{Equation:HBHLoad0}\\
&&\textbf{// After the last delayed transaction}\notag\\
&&\thetransitionInstrumented{\alab_1}{\alab_{x2}}{\thecondition{ \hb{} \neq \perp \land\ p.vh =  p.a = \perp}}\notag\\
&&\thetransitionInstrumented{\alab_{x2}}{\alab_{x3}}{\theload{\areg}{\anaddr}}\notag\\
&&\thetransitionInstrumented{\alab_{x3}}{\alab_{x4}}{\thecondition{\anaddr.event' = \storeacc \land\ p.hbh = \perp}}\notag\\
&&\thetransitionInstrumented{\alab_{x4}}{\alab_{2}}{\theassign{p.hbh}{\hbhjoin}}\label{Equation:HBHMeetConditionLoad}\\
&&\thetransitionInstrumented{\alab_{x3}}{\alab_{x5}}{\thecondition{\anaddr.event = \perp}}\notag\\
&&\thetransitionInstrumented{\alab_{x5}}{\alab_{2}}{\theassign{\anaddr.event}{\loadacc}}\label{Equation:HBHFlagLoad}\\
&&\thetransitionInstrumented{\alab_{x3}}{\alab_{2}}{\thecondition{\anaddr.event \neq \perp \lor \  p.hbh \neq \perp}}\notag\\[1.5mm]
%%%%%%%%%%%
&&\semHbhelperorig{\thetransition{\alab_1}{\alab_2}{\commitact}}=\notag\\
&&\textbf{// Before joining happens-before helpers}\notag\\
&&\thetransitionInstrumented{\alab_1}{\alab_{x1}}{\thecondition{\hb{} = p.vh =  p.a = \perp}}\notag\\
&&\thetransitionInstrumented{\alab_{x1}}{\alab_2}{\commitact}\notag\\
&&\textbf{// After joining happens-before helpers}\notag\\
&&\thetransitionInstrumented{\alab_1}{\alab_{x2}}{\thecondition{\hb{} \neq \perp \land \  p.hbh \neq \perp}}\notag\\
&&\thetransitionInstrumented{\alab_{x2}}{\alab_{x3}}{\commitact}\notag\\
&&\thetransitionInstrumented{\alab_{x3}}{\alab_{x4}}{\theload{\tilde r}{a_{\attackstore}}}\label{Equation:HBHFINISH1}\\
&&\thetransitionInstrumented{\alab_{x4}}{\alab_{x5}}{\theassign{\tilde r}{\tilde r .event}}\label{Equation:HBHFINISH2}\\
&&\thetransitionInstrumentedf{\alab_{x5}}{\thecondition{\tilde r \neq \perp}}\label{Equation:HBHFINISH3}\\
&&\thetransitionInstrumented{\alab_{x5}}{\alab_{2}}{\thecondition{\tilde r = \perp}}\notag\\
&&\textbf{// Failed to join happens-before helpers}\notag\\
&&\thetransitionInstrumentednotogo{\alab_1}{\thecondition{\hb{} \neq \perp \land\ p.hbh = p.vh = p.a = \perp}}\notag\\
&&\thetransitionInstrumentedbfinal{\ \ \ \ \ \ \ \ \ \ \ \ \ \ \ \ \ \ \ \ \ \ \  \ \ \ \ \ \ \ \  \  \ \ \ \ }\label{Equation:HBHBLOCK}
%%%%%%%%%%%
\end{eqnarray}
 \end{minipage}
\normalsize
\caption{Instrumentation of Happens-Before Helpers.}
\label{Figure:HBHelpersInstrumentation}
\end{figure}

\subsection{Instrumentation of the Happens-Before Helpers}\label{subsec:HBHelpers}

The remaining processes, which are not the attacker or a visibility helper, can become happens-before helpers. Fig. \ref{Figure:HBHelpersInstrumentation} lists the instrumentation of the instructions of a happens-before helper process. Similar to above, when the flag $a_{\attacktransaction}$ is set to $\mytrue$ by the attacker, other processes enter a phase where they continue executing their instructions, however, when they store a value they write it in both the shared variable $z$ and the primed variable $z'$ (lines \eqref{Equation:HBHStore1} and \eqref{Equation:HBHStore2}). However, they only read from the original shared variables (line \eqref{Equation:HBHLoad0}). Once the flag $\hbo$ is set to $\hbhjoin$, a process that cannot be the attacker (i.e., the flag $p.a$ is null) or a visibility helper (i.e., the flag $p.vh$ is null) chooses non-deterministically a transaction $\atr$ (the $\beginact$ of this transaction is shown in line~\eqref{Equation:HBHTransaction}) through which it wants to join the set of happens-before helpers, i.e., continue the happens-before path created by the existing happens-before helpers. Similar to visibility helpers, when a process choses the transaction $\atr$, it makes a pledge (while executing the \plog{begin} instruction) that during this transaction it will either read a variable updated by another happens-before helper or write to a variable that was accessed (read or written) by another happens-before helper (every process that executes a transaction after $\hb{}$ is set to $\mytrue$ makes this pledge). When the pledge is met, the process sets the flag $p.hbh$ to $\hbhjoin$ (lines \eqref{Equation:HBHMeetConditionLoad} and \eqref{Equation:HBHMeetConditionStore}). The execution is blocked if a process does not keep its pledge (i.e., the flag $p.hbh$ is null) at the end of the transaction (line~\eqref{Equation:HBHBLOCK}). We use a flag $x.event$ for each variable $x$ to record the type (read $\loadacc$ or write $\storeacc$) of the last access made by a happens-before helper (lines \eqref{Equation:HBHFlagLoad} and \eqref{Equation:HBHFlagStore}).
Moreover, once $\hb{}$ is set to $\mytrue$ (i.e., there are no more delayed transactions), the process can write and read only the original shared variables, since the primed versions are no longer in use. A particular case is when the transaction $\atr$ is from the first process trying to join the happens-before helpers, in which the transaction must contain a read accessing the variable $x$ that was read or written to by a transaction from the attacker of a visibility helper.

The happens-before helpers continue executing their instructions, until one of them reads from the shared variable $y$ whose name was stored in $a_{\attackstore}$. This establishes a happens-before path between the last delayed transaction and a ``fictitious'' store event corresponding to the first delayed transaction that could be executed just after this read of $y$. The execution does not have to contain this store event explicitly since it is always enabled. Therefore, at the end of every transaction, the instrumentation checks whether the transaction read $y$. If it is the case, then the execution stops and goes to an error state to indicate that this is a robustness violation.
The happens-before helpers processes continue executing their instructions, until one of them executes a load that reads from the shared variable $y$ that was stored in $a_{\attackstore}$ which implies the existence of a happens-before cycle.
Thus, when executing the instruction $\commitact$ at the end of every transaction, we have a conditional check to detect if we have a load or a write  accessing the variable $y$ (lines~\eqref{Equation:HBHFINISH1}, \eqref{Equation:HBHFINISH2}, and \eqref{Equation:HBHFINISH3}).
When the check detects that the variable $y$ was accessed, the execution goes to an error state (line~\eqref{Equation:HBHFINISH3}) to indicate that it has produced a robustness violation.

In Fig. \ref{Figure:InstrExmp}, we show an excerpt of the instrumentations of the two transactions of the $\mathsf{SB}$ program. 
In particular, we only give the instructions of the instrumented $\mathsf{SB}$ that are reached during the execution that leads to an error state. The attacker instrumentation is applied to the transaction $\atr 1$ of $\apr 1$ and the happens-before helpers instrumentation is applied to the transaction $\atr 2$ of $\apr 2$. The first conflict order from $\atr 1$ to $\atr 2$ (shown in Fig.~\ref{fig:ccvrobtraces}) is simulated by the fact that at line \ref{Equation:RW12}, $y.event' = \loadacc$ (see lines \ref{Equation:RW10} and \ref{Equation:RW11}). Also, the second conflict order from $\atr 2$ to $\atr 1$ is simulated by the fact that at line \ref{Equation:RW22} we reach the error state where $a_{\attackstore}.event = x.event = \loadacc$ (see lines \ref{Equation:RW20} and \ref{Equation:RW21}). 

%!TEX root = draft.tex

\begin{figure}[t]
 \scriptsize
 \hspace*{-0.75cm}
 \begin{minipage}[l]{0.41\textwidth}
    \begin{eqnarray}
      &&\semidler{\thetransition{\alab_1}{\alab_2}{\beginact}} = \notag\\[1mm]
      &&\thetransitionInstrumented{\alab_1}{\alab_{b2}}{\thecondition{\hb{} = \perp \land \  a_{\attacktransaction} = \perp}}\notag\\
      &&\thetransitionInstrumented{\alab_{b2}}{\alab_{b3}}{\beginact}\notag\\
      &&\thetransitionInstrumented{\alab_{b3}}{\alab_{b4}}{\theassign{p1.a}{\ajoin}}\notag\\
      &&\thetransitionInstrumented{\alab_{b4}}{\alab_{b5}}{\theforeachassign{z}{\mathbb{V}}{z'}{z}}\notag\\
      &&\thetransitionInstrumented{\alab_{b5}}{\alab_2}{\theassign{a_{\attacktransaction}}{\mytrue}}\notag\\[1.5mm]
      %%%%%%%%%%%
      &&\semidler{\thetransition{\alab_2}{\alab_3}{\thestore{x}{1}}} = \notag\\[1mm]
      &&\thetransitionInstrumented{\alab_2}{\alab_{s4}}{\thecondition{a_{\attackstore} = x.event = \perp \land\ p1.a\neq \perp}}\notag\\
      &&\thetransitionInstrumented{\alab_{s4}}{\alab_{s5}}{\thestore{x'}{1}}\notag\\
      &&\thetransitionInstrumented{\alab_{s5}}{\alab_{s6}}{\thestore{a_{\attackstore}}{`x`}}\label{Equation:RW20}\\
      &&\thetransitionInstrumented{\alab_{s6}}{\alab_{3}}{\theassign{x'.event}{\storeacc}}\notag\\[1.5mm]
      %%%%%%%%%%%
      &&\semidler{\thetransition{\alab_3}{\alab_4}{\theload{\areg 1}{y}}}  =\notag\\[1mm]
      &&\thetransitionInstrumented{\alab_{3}}{\alab_{l4}}{\thecondition{y'.event = \perp \land\ a_{\attacktransaction} \neq \perp }}\notag\\
      &&\thetransitionInstrumented{\alab_{l4}}{\alab_{l5}}{\theload{\areg 1}{y'}}\notag\\
      &&\thetransitionInstrumented{\alab_{l5}}{\alab_{l6}}{\thestore{\hb{}}{\mytrue}}\notag\\
      &&\thetransitionInstrumented{\alab_{l6}}{\alab_{4}}{\theassign{y.event}{\loadacc}}\label{Equation:RW10}\\[1.5mm]
      %%%%%%%%%%%
      &&\semidler{\thetransition{\alab_4}{\alab_5}{\commitact}}=\notag\\[1mm]
      &&\thetransitionInstrumented{\alab_{4}}{\alab_{5}}{\commitact}\notag
      %%%%%%%%%%%
    \end{eqnarray}
     \end{minipage}
     \hfill
    \begin{minipage}[r]{0.585\textwidth}
    \begin{eqnarray}
    %%%%%%%%%%%
    %%%%%%%%%%%
    &&\semHbhelperorig{\thetransition{\alab_1}{\alab_2}{\beginact}}=\notag\\
    &&\thetransitionInstrumentedThree{\alab_1}{\alab_{b2}}{\thecondition{\hb{} \neq \perp \land\ p2.hbh = p2.vh = p2.a = \perp}}\notag\\
    &&\thetransitionInstrumented{\alab_{b2}}{\alab_{b3}}{\beginact}\notag\\
    &&\thetransitionInstrumentedTwo{\alab_{b3}}{\alab_2}{\thestore{x.event'}{x.event}}{\thestore{y.event'}{y.event}}\label{Equation:RW11}\\[1.5mm]
    %%%%%%%%%%%
    &&\semHbhelperorig{\thetransition{\alab_3}{\alab_4}{\thestore{y}{1}}}=\notag\\[1mm]
    &&\thetransitionInstrumented{\alab_3}{\alab_{s4}}{\thecondition{ \hb{} \neq \perp \land\ p2.vh =  p2.a = \perp}}\notag\\
    &&\thetransitionInstrumented{\alab_{s4}}{\alab_{s5}}{\thestore{y}{1}}\notag\\
    &&\thetransitionInstrumented{\alab_{s5}}{\alab_{s6}}{\theassign{y.event}{\storeacc}}\notag\\
    &&\thetransitionInstrumented{\alab_{s6}}{\alab_{s7}}{\thecondition{y.event' \neq \perp \land \  p2.hbh = \perp}}\label{Equation:RW12}\\
    &&\thetransitionInstrumented{\alab_{s7}}{\alab_{4}}{\theassign{p2.hbh}{\hbhjoin}}\notag\\[1.5mm]
    %%%%%%%%%%%
    &&\semHbhelperorig{\thetransition{\alab_2}{\alab_3}{\theload{\areg2}{x}}}=\notag\\[1mm]
    &&\thetransitionInstrumented{\alab_2}{\alab_{l2}}{\thecondition{ \hb{} \neq \perp \land\ p.vh =  p.a = \perp}}\notag\\
    &&\thetransitionInstrumented{\alab_{l2}}{\alab_{l3}}{\theload{\areg2}{x}}\notag\\
    &&\thetransitionInstrumented{\alab_{l3}}{\alab_{l5}}{\thecondition{x.event = \perp}}\notag\\
    &&\thetransitionInstrumented{\alab_{l5}}{\alab_{3}}{\theassign{x.event}{\loadacc}}\label{Equation:RW21}\\[1.5mm]
    %%%%%%%%%%%
    &&\semHbhelperorig{\thetransition{\alab_4}{\alab_5}{\commitact}}=\notag\\
    &&\thetransitionInstrumented{\alab_4}{\alab_{e2}}{\thecondition{\hb{} \neq \perp \land \  p2.hbh \neq \perp}}\notag\\
    &&\thetransitionInstrumented{\alab_{e2}}{\alab_{e3}}{\commitact}\notag\\
    &&\thetransitionInstrumented{\alab_{e3}}{\alab_{e4}}{\theload{\tilde r}{a_{\attackstore}}}\notag\\
    &&\thetransitionInstrumented{\alab_{e4}}{\alab_{e5}}{\theassign{\tilde r}{\tilde r .event}}\notag\\
    &&\thetransitionInstrumentedf{\alab_{e5}}{\thecondition{\tilde r \neq \perp}}\label{Equation:RW22}%
%%%%%%%%%%%
%%%%%%%%%%%
    \end{eqnarray}
     \end{minipage}
\normalsize
\caption{Instrumentation of $\mathsf{SB}$ program in Fig. \ref{fig:rob1}.}
\label{Figure:InstrExmp}
\end{figure} 

\subsection{Correctness}
As we have already mentioned, the role of a process in an execution is chosen non-deterministically at runtime. Therefore, the final instrumentation of a given program $\aprog$, denoted by $\semTwo{\aprog}$, is obtained by replacing each labeled instruction $\langle linst\rangle$ with the concatenation of the instrumentations corresponding to the attacker, the visibility helpers, and the happens-before helpers, i.e., $\semTwo{\langle linst\rangle} ::= \semidler{\langle linst\rangle} \hspace{0.17cm}   \semVhelperorig{\langle linst\rangle} \hspace{0.17cm}   \semHbhelperorig{\langle linst\rangle}$. 
The instrumented program $\semTwo{\aprog}$ reaches the error state iff $\aprog$ admits a violation of the pattern $\tau_{\ccvt{}2}$. Let $\semOne{\aprog}$ be the instrumented program that reaches an error state iff $\aprog$ admits a violation of the pattern $\tau_{\ccvt{}1}$. The instrumentation $\semOne{\ }$ does not include the visibility helpers since only a single transaction is delayed in $\tau_{\ccvt{}1}$, and it can be obtained in the same manner as $\semTwo{\ }$. The following theorem states the correctness of the instrumentation. 

\begin{thm}\label{th:final}
A program $\aprog$ is not robust against \ccvt{} iff either $\semOne{\aprog}$ or $\semTwo{\aprog}$ reaches the error state.
\end{thm}

%The aim of the instrumentation procedure is to reduce the problem of checking the existence of the violation patterns described above to reachability under SC of en error state by the instrumented version of a program. The instrumentation procedure is sound and complete iff if 

The proof of this theorem relies on the explanations given above. One can define a bijection between executions of the instrumentation that reach an error state and executions of the original program that satisfy the constraints in one of the two violation patterns. The former can be rewritten to the latter by roughly, removing all accesses to the auxiliary variables used by the instrumentation, replacing the writes to shared variable copies by writes to the original variables, delivering delayed transactions only to visibility helpers, and appending store events for all the delayed transactions. For the reverse, given a robustness violation $\tau = \atmacro\cdot \issueact(\apr,\atr)\cdot \beta_1\cdot \issueact(\apr_{1},\atr_1) \cdot\beta_2\cdot (\apr',\atr')\cdot \storeact(\apr',\atr) \cdot \stmacro$ of type $\tau_{\ccvt{}2}$, we can build an execution of the instrumentation that reaches an error state, where $\apr$ is the attacker, the processes delaying transactions in $\beta_1\cdot  \issueact(\apr_1,\atr_1)$ are visibility helpers, and the processes that issue transactions between $\issueact(\apr_1,\atr_1)$ and $\storeact(\apr',\atr)$ and that are part of the happens-before path between these two events are the happens-before helpers.

The following result states the complexity of checking robustness for finite-state programs\footnote{That is, programs where the number of variables and the data domain are bounded.} against one of the three variations of causal consistency considered in this work (we use causal consistency as a generic name to refer to all of them). The upper bound is a direct consequence of Theorem~\ref{th:final} and of previous results concerning the reachability problem in concurrent programs running over SC, with a fixed~\cite{DBLP:conf/focs/Kozen77} or parametric number of processes~\cite{DBLP:journals/tcs/Rackoff78}. 
For the lower bound, given an instance of the reachability problem under sequential consistency, denoted by $(\aprog,\ell)$\footnote{That is, whether the program $\aprog$ reaches the control location $\ell$ under SC.}, we construct a program $\aprog'$ where each statement $s$ of $\aprog$ is a different transaction (guarded by a global lock), and where reaching the location $\ell$ enables the execution of a ``gadget'' that corresponds to the $\mathsf{SB}$ program in Figure~\ref{fig:rob1}. Executing each statement under a global lock ensures that every execution of $\aprog'$ under causal consistency is serializable, and faithfully represents an execution of the original $\aprog$ under sequential consistency. Moreover, $\aprog$ reaches $\ell$ iff $\aprog'$ contains a robustness violation, which is due to the execution of $\mathsf{SB}$.

\begin{cor}\label{theorem:InstRobus}
Checking robustness of finite-state programs against causal consistency is PSPACE-complete when the number of processes is fixed and EXPSPACE-complete, otherwise.
\end{cor}

\begin{rem}
The reduction to reachability does not manipulate transaction identifiers and it is insensitive to the number of transactions executed by one process. Thus, all our results extend to processes that include unbounded loops of transactions. This includes programs where each process can call a statically known set of transactions (with parameters) an arbitrary number of times. 
\end{rem}
%!TEX root = draft.tex
\section{Related Work}

Causal consistency is one of the oldest consistency models for distributed systems \cite{DBLP:journals/cacm/Lamport78}. Formal definitions of several variants of causal consistency, suitable for different types of applications, have been introduced recently~\cite{DBLP:conf/popl/BurckhardtGYZ14,DBLP:journals/ftpl/Burckhardt14,DBLP:conf/ppopp/PerrinMJ16,DBLP:conf/popl/BouajjaniEGH17}. The definitions in this paper are inspired from these works and coincide with those given in \cite{DBLP:conf/popl/BouajjaniEGH17}. In that paper, the authors address the decidability and the complexity of verifying that an implementation of a storage system is causally consistent (i.e., all its computations, for every client, are causally consistent).
%This problem is different from the robustness problem we consider in this work that concerns client programs of storage systems: the computations of a given transactional program may be serializable even if the storage system it uses does not ensure serializability but only causal consistency.

%The concept of robustness we consider in this paper is {\em trace-based} in the sense that a program is robust if each of its computations according to the weak semantics has a counterpart according to the strong semantics.
While our paper focuses on {\em trace-based} robustness, \emph{state-based robustness} requires that a program is robust if the set of all its reachable states under the weak semantics is the same as its set of reachable states under the strong semantics. While state-robustness is the necessary and sufficient concept for preserving state-invariants, its verification, which amounts in computing the set of reachable states under the weak semantics, is in general a hard problem.
The decidability and the complexity of this problem has been investigated in the context of relaxed memory models such as TSO and Power, and it has been shown that it is either decidable but highly complex (non-primitive recursive), or undecidable \cite{DBLP:conf/popl/AtigBBM10,DBLP:conf/esop/AtigBBM12}. As far as we know, the decidability and complexity of this problem has not been investigated for causal consistency.
Automatic procedures for approximate reachability/invariant checking have been proposed using either abstractions or bounded analyses, e.g., \cite{DBLP:conf/cav/AtigBP11,DBLP:conf/cav/AlglaveKT13,DBLP:journals/cl/DanMVY17,DBLP:conf/tacas/AbdullaABN17}. Proof methods have also been developed for verifying invariants in the context of weakly consistent models such as \cite{DBLP:conf/icalp/LahavV15,DBLP:conf/popl/GotsmanYFNS16,DBLP:conf/eurosys/NajafzadehGYFS16,DBLP:conf/popl/AlglaveC17}. These methods, however, do not provide decision procedures.
%As far as we know, solving reachability has not been investigated for weakly consistent models of distributed systems such as eventual consistency, causal consistency, etc.

Decidability and complexity of trace-based robustness has been investigated for the TSO and Power memory models \cite{DBLP:conf/icalp/BouajjaniMM11, DBLP:conf/esop/BouajjaniDM13,DBLP:conf/icalp/DerevenetcM14}. The work we present in this paper borrows the idea of using minimal violation characterizations for building an instrumentation allowing to obtain a reduction of the robustness checking problem to the reachability checking problem over SC.
However, applying this approach to the case of causal consistency is not straightforward and requires different proof techniques. Dealing with causal consistency is far more tricky and difficult than dealing with TSO, and requires
%. In fact the TSO and causal consistency models differ in many subtle ways,
%and therefore dealing with causal consistency requires a novel investigation of the robustness problem taking into account the specificities of causal consistency, and
coming up with radically different arguments and proofs, for (1) characterizing in a finite manner the set of violations, (2) showing that this characterization is sound and complete, and (3) using effectively this characterization in the definition of the reduction to the reachability problem.
% define an instrumentation of the given program that leads to a reduction of the robustness checking problem to a reachability checking problem by sequentially consistent computations only.

As far as we know, our work is the first one that establishes results on the decidability and complexity issues of the robustness problem in the context of causal consistency, and taking into account transactions. The existing work on the verification of robustness for distributed systems consider essentially trace-based concepts of robustness and provide either over- or under-approximate analyses for checking it. In \cite{DBLP:conf/concur/0002G16,DBLP:conf/popl/BrutschyD0V17,DBLP:conf/pldi/BrutschyD0V18,DBLP:journals/jacm/CeroneG18}, static analysis techniques are proposed based on computing an abstraction of the set of computations that is used in searching for robustness violations. These approaches may return false alarms due to the abstractions they consider.
In particular, \cite{DBLP:conf/concur/0002G16} shows that a trace under causal convergence is not admitted by the serializability semantics iff it contains a (transactional) happens-before cycle with a $\cfo$ dependency, and another $\cfo$ or $\sto$ dependency. This characterization alone is not sufficient to prove our result concerning robustness checking. Our result relies on a characterization of more refined robustness violations and relies on different proof arguments.
In  \cite{DBLP:conf/concur/NagarJ18} a sound (but not complete) bounded analysis for detecting robustness violation is proposed. Our approach is technically different, is precise, and provides a decision procedure for checking robustness when the program is finite-state.
\section{Conclusion}

We have studied three variations of transactional causal consistency, showing that they are equivalent for programs without write-write races. 
We have shown that the problem of verifying that a transactional program is robust against causal consistency can be reduced, modulo a linear-size instrumentation, to a reachability problem in a transactional program running over a sequentially consistent shared memory.
This reduction leads to the first decidability result concerning the problem of checking robustness against a weak transactional consistency model. Furthermore, this reduction opens the door to the use of existing methods and tools for the analysis and verification of SC concurrent programs, in order to reason about weakly-consistent transactional programs. It can be used for the design of a large spectrum of static/dynamic tools for testing/verifying robustness against causal consistency. 

Our notion of robustness relies on a particular interpretation of behaviors as traces recording all happens-before dependencies. This is stronger than a more immediate notion of \emph{state-based} robustness that requires equality of sets of reachable states, which means that it could produce false alarms, i.e., robustness violations that are not also violations of the intended program specification. This trade-off is similar in spirit to data races being used as an approximation of concurrency errors (since data races are easier to detect, compared to violations of arbitrary specifications).

An interesting direction for future work is looking at the robustness problem in the context of \emph{hybrid} consistency models where some of the transactions in the program can be declared serializable. These models include synchronization primitives similar to lock acquire/release which allow to enforce a serialization order between some transactions. Such mechanisms can be used as a ``repair'' mechanism in order to make programs robust.

%\input{AxiomaticDefinitions}

%%
%% Bibliography
%%

%% Please use bibtex,
%% in general the use of bibtex is encouraged
\bibliographystyle{plain}
\bibliography{draft}
%\bibliography{dblp,misc}

%% Appendix
% \appendix
% \newpage
%\input{App2OperationalModels}
%\input{App3AxiomaticDefinitions}
%\input{App4ComparisonCCModels}
%\input{App1Figures}
%\input{App5TraceActReord}
%\input{App6MiniViolCCv}
%\input{App7MiniViolCM}
%\input{App8MiniViolCC}
%\input{App9Instrumentation}
%\input{App10SoundCompleInstru}
%\input{App11InstrumentationExamples}

\end{document}